\documentclass[acmsmall,screen]{acmart}
\usepackage{caption, subcaption}
\usepackage[noend]{algpseudocode}
\usepackage{algorithm}
\usepackage{tikz}
\usetikzlibrary{arrows,automata,positioning,backgrounds,calc,patterns}
\usetikzlibrary{arrows.meta}
\usepackage{mathpartir}
\usepackage{framed}
\usepackage{ulem}
\usepackage{adjustbox}
\usepackage{thmtools}
\usepackage{hyperref}
\usepackage{cleveref}
\usepackage{thm-restate}
\usepackage{wrapfig}
\usepackage{listings}

\tikzset{transaction state/.style={draw=black!0}}
\tikzset{
	arrow/.pic={\path[tips,every arrow/.try,->,>=#1] (0,0) -- +(.1pt,0);},
	pics/arrow/.default={triangle 90},
	lab dis/.store in=\LabDis,
	lab dis=0.3,
	->-/.style args={at #1 with label #2}{decoration={
			markings,
			mark=at position #1 with {\arrow{>}; \node at (0,\LabDis) {#2};}},postaction={decorate}},
	-<-/.style args={at #1 with label #2}{decoration={
			markings,
			mark=at position #1 with {\arrow{<}; \node at (0,\LabDis)
				{#2};}},postaction={decorate}},
	-*-/.style={decoration={
			markings,
			mark=at position #1 with {\fill (0,0) circle (1.5pt);}},postaction={decorate}}
}

\definecolor{pinegreen}{rgb}{0.0, 0.47, 0.44}
\definecolor{newVersionColor}{HTML}{000000}
\definecolor{appendixVersionColor}{HTML}{000000}

\definecolor{noAppendixVersionColor}{HTML}{000000}

\newcommand*\nver{\textcolor{newVersionColor}}
\newcommand\noAppendixver[1]{\iftrue\textcolor{noAppendixVersionColor}{#1}\fi}

\newcommand\oldver[1]{}
\newcommand\appendixver[1]{\iffalse{\leavevmode\textcolor{appendixVersionColor}{#1}}\fi}

\newcommand\red[1]{{\color{red}#1}}

\newcommand{\callout}[1]{\emph{#1}}

\newcommand{\key}{x}

\newcommand{\Events}{\mathcal{E}}

\definecolor{coColor}{HTML}{b37200}
\definecolor{soColor}{HTML}{800080}
\definecolor{poColor}{HTML}{C41E3A}
\definecolor{wrColor}{HTML}{008000}
\definecolor{orColor}{HTML}{0026ff}

\newcommand{\tup}[1]{{\left\langle{#1}\right\rangle}}
\renewcommand{\implies}{\Rightarrow}

\newcommand{\Vars}{\mathsf{Vars}}
\newcommand{\LVars}{\mathsf{LVars}}

\newcommand{\Vals}{\mathsf{Vals}}

\newcommand{\xvar}{{a}}
\newcommand{\prog}{{\mathsf{P}}}
\newcommand{\histOf}[2][]{{\mathsf{hist}_{#1}({#2})}}

\newcommand{\ebegin}{\mathsf{begin}}
\newcommand{\eabort}{\mathsf{abort}}
\newcommand{\ecommit}{\mathsf{commit}}
\newcommand{\erd}[1]{\mathsf{read}({#1})}
\newcommand{\ewrt}[1]{\mathsf{write}({#1})}

\newcommand{\yvar}{{y}}
\newcommand{\val}{{v}}

\newcommand{\rd}[2][]{\textsf{read}_{#1}({#2})}

\newcommand{\wrt}[3][]{\textsf{write}_{#1}({#2},{#3})}

\newcommand{\tr}{t}

\newcommand{\hist}{{h}}
\newcommand{\po}{\textcolor{poColor}{\mathsf{po}}}
\newcommand{\so}{\textcolor{soColor}{\mathsf{so}}}
\newcommand{\co}{\textcolor{coColor}{\mathsf{co}}}
\newcommand{\ora}{\textcolor{orColor}{\mathsf{or}}}
\newcommand{\wro}{\textcolor{wrColor}{\mathsf{wr}}}

\newcommand{\readOp}[1]{\mathsf{reads}({#1})}
\newcommand{\tlogs}[1]{\mathsf{tr}({#1})}
\newcommand{\transC}[1]{\mathsf{commTrans}({#1})}

\newcommand{\writeOp}[1]{\mathsf{writes}({#1})}
\newcommand{\writeVar}[2]{{#1}\ \mathsf{writes}\ {#2}}

\renewcommand{\hist}{{h}}

\newcommand{\trans}[2]{\mathsf{tr}({#1},{#2})}

\newcommand{\CC}{\textup{\texttt{CC}}}
\newcommand{\SER}{\textup{\texttt{SER}}}

\newcommand{\SI}{\textup{\texttt{SI}}}
\newcommand{\RA}{\textup{\texttt{RA}}}
\newcommand{\RC}{\textup{\texttt{RC}}}

\newcommand{\SC}{\textup{\texttt{SC}}}

\newcommand{\init}{\textup{\textbf{\texttt{init}}}}

\newcommand{\locals}{\mathsf{locals}}
\newcommand{\genericNext}{\textsc{Next}}
\newcommand{\genericEvaluate}{\textsc{Valid}}
\newcommand{\genericCompute}{\textsc{ComputeReorderings}}
\newcommand{\genericProtocol}{\textsc{Optimality}}
\newcommand{\genericSwap}{\textsc{Swap}}
\newcommand{\genericValidWrites}{\textsc{ValidWrites}}
\newcommand{\nextEvent}{\textup{\textsc{next}}}

\newcommand{\protocol}[1][\mathcal{I}]{\textup{\textsc{isSwappable}}_{#1}}
\newcommand{\swap}{\textup{\textsc{swap}}}

\newcommand{\oraRespectful}[2]{\texttt{R}^{\ora}(#1, #2)}
\newcommand{\oraRespectfulCanon}[1]{\texttt{R}^{\ora}(#1)}

\newcommand{\last}[1]{\textup{\texttt{last}}(#1)}

\newcommand{\swapped}[2]{\textup{\textsc{swapped}}(#1, #2)}
\newcommand{\maxCompletion}{\textup{\textsc{maxCompletion}}}

\newcommand{\prev}{\textup{\textsc{prev}}}

\newcommand{\isMaximallyAdded}[4][I]{\mathsf{readLatest}_{#1}(#2, #3, #4)}
\newcommand{\readsCasuallyLatest}[3][I]{\mathsf{readsCLatest}_{#1}(#2, #3)}
\newcommand{\dep}{\textup{\textsc{dep}}}
\newcommand{\minimalDependency}{\textup{\textsc{minimalDependency}}}

\newcommand{\events}[1]{\mathsf{events}({#1})}

\newcommand{\variable}[1]{\mathtt{var}({#1})}

\newcommand{\eqdef}{::=}

\newcommand{\ibegin}{\mathtt{begin}}

\newcommand{\icommit}{\mathtt{commit}}
\newcommand{\iabort}{\mathtt{abort}}
\newcommand{\iif}[2]{\mathtt{if}({#1}) \{ {#2} \}}
\newcommand{\iread}{\mathtt{read}}
\newcommand{\iwrite}{\mathtt{write}}

\algrenewcommand\algorithmicindent{1.0em}%
\algnewcommand\algorithmicswitch{\textbf{switch}}
\algnewcommand\algorithmiccase{\textbf{case}}
\algnewcommand\algorithmicassert{\texttt{assert}}
\algnewcommand\Assert[1]{\State \algorithmicassert(#1)}%
\algdef{SE}[SWITCH]{Switch}{EndSwitch}[1]{\algorithmicswitch\ #1\ \algorithmicdo}{\algorithmicend\ \algorithmicswitch}%
\algdef{SE}[CASE]{Case}{EndCase}[1]{\algorithmiccase\ #1\ :}{\algorithmicend\ \algorithmiccase}%
\algtext*{EndSwitch}%
\algtext*{EndCase}%
\algtext*{EndFor}%
\algnewcommand\Let{\State\textbf{let} }
\algnewcommand\Break{\State\textbf{break} }
\algnewcommand\InputAlgorithmic{\Statex\textbf{Procedure }} 

\algnewcommand\OutputAlgorithmic{\Statex\textbf{Output: }}
\algblockdefx[BlockName]{ForTo1}{EndForTo1}%
[3][i]{\textbf{for} #1 $\gets$ #2 \textbf{to} #3 \textbf{do}}%

\lstdefinelanguage{MyLang}{%
	keywords = { delete, do, each, else, export, finally, for, foreach, function,
		if, in, let, of, return, void, while, with, yield, elements, read, write,
		insert, remove, add, AddItem, DeleteItem, Push, Pop, Enroll, Tweet, Timeline,
		NewsFeed, begin, end, break, throw},
	morecomment = [l]{//},
	morecomment = [s]{/*}{*/},
	morestring  = [b]',
	morestring  = [b]",
	sensitive   = true,
}

\lstdefinelanguage{Java10}{
	language      = Java,
	morekeywords  ={ var },
}

\setcopyright{rightsretained}
\acmPrice{}
\acmDOI{10.1145/3591243}
\acmYear{2023}
\copyrightyear{2023}
\acmSubmissionID{pldi23main-p143-p}
\acmJournal{PACMPL}
\acmVolume{7}
\acmNumber{PLDI}
\acmArticle{129}
\acmMonth{6}
\received{2022-11-10}
\received[accepted]{2023-03-31}

\begin{document}
	
	\title[Dynamic Partial Order Reduction for Checking Correctness against Transaction Isolation Levels]{Dynamic Partial Order Reduction for Checking Correctness against Transaction Isolation Levels}       


	
	\author{Ahmed Bouajjani}
	\affiliation{
		\institution{Université Paris Cité, CNRS, IRIF}
		\country{France}                  
	}
	\email{abou@irif.fr}          

	\author{Constantin Enea}
	\affiliation{
		\institution{LIX, Ecole Polytechnique, CNRS and Institut Polytechnique de Paris}
		\country{France}                    
	}
	\email{cenea@lix.polytechnique.fr}          

	\author{Enrique Rom\'an-Calvo}
	\affiliation{
		\institution{Université Paris Cité, CNRS, IRIF} 
		\country{France}                    
	}
	\email{calvo@irif.fr}          
	
	\begin{abstract}

Modern applications, such as social networking systems and e-commerce platforms are centered around using large-scale databases for storing and retrieving data. Accesses to the database are typically enclosed in transactions that allow computations on shared data to be isolated from other concurrent computations and resilient to failures. Modern databases trade isolation for performance. The weaker the isolation level is, the more behaviors a database is allowed to exhibit and it is up to the developer to ensure that their application can tolerate those behaviors. 

In this work, we propose stateless model checking algorithms for studying correctness of such applications that rely on dynamic partial order reduction. These algorithms work for a number of widely-used weak isolation levels, including Read Committed, Causal Consistency, Snapshot Isolation and Serializability. We show that they are complete, sound and optimal, and run with polynomial memory consumption in all cases. We report on an implementation of these algorithms in the context of Java Pathfinder applied to a number of challenging applications drawn from the literature of distributed systems and databases.  
	\end{abstract}


	\maketitle
	
\section{Introduction}
\label{sec:intro}

Data storage is no longer about writing data to a single
disk with a single point of access. Modern applications require not just data
reliability, but also high-throughput concurrent accesses. 
Applications concerning supply chains, banking, etc. use traditional relational databases
for storing and processing data, whereas applications such as social networking
software and e-commerce platforms 
use cloud-based storage systems (such as Azure Cosmos DB \cite{cosmosdb}, Amazon DynamoDB
\cite{decandia2007dynamo}, Facebook TAO \cite{facebook-tao}, etc.). 

 

Providing high-throughput processing, unfortunately, comes at an unavoidable cost of weakening 
the consistency guarantees offered to users:
Concurrently-connected clients may end up observing different versions of the same data. 
These ``anomalies'' can be prevented by using a strong \textit{isolation level} 
such as \textit{Serializability} \cite{DBLP:journals/jacm/Papadimitriou79b}, which essentially offers a single version of the data to all clients at any point in time. 
However, serializability requires expensive synchronization and incurs a high performance cost. As a
consequence, most storage systems use weaker isolation levels, such as 
{\it Causal Consistency}~\cite{DBLP:journals/cacm/Lamport78,DBLP:conf/sosp/LloydFKA11,antidote},
{\it Snapshot Isolation}~\cite{DBLP:conf/sigmod/BerensonBGMOO95}, {\it Read
Committed}~\cite{DBLP:conf/sigmod/BerensonBGMOO95}, etc. for better performance.
In a recent survey of
database administrators \cite{survey}, 86\% of the participants responded that
most or all of the transactions in their databases execute at Read Committed level.

A weaker isolation level allows for more possible behaviors than stronger
isolation levels. It is up to the developers then to ensure that their
application can tolerate this larger set of behaviors. Unfortunately, weak
isolation levels are hard to understand or reason about
\cite{DBLP:conf/popl/BrutschyD0V17,adya-thesis} and resulting application bugs
can cause loss of business \cite{acidrain}.

\vspace{.5mm}
\noindent
\textbf{Model Checking Database-Backed Applications.}
This paper addresses the problem of \textit{model checking} code for correctness against a given isolation level. 
\emph{Model checking}~\cite{DBLP:conf/popl/ClarkeES83,DBLP:conf/programm/QueilleS82} explores the state space of a given program in a systematic manner and it provides high coverage of program behavior. However, it faces the infamous state explosion problem, i.e., the number of executions grows exponentially in the number of concurrent clients. 


\emph{Partial order reduction} (POR)~\cite{DBLP:journals/sttt/ClarkeGMP99,DBLP:books/sp/Godefroid96,DBLP:conf/cav/Peled93,DBLP:conf/apn/Valmari89} is an approach that limits the number of explored executions without sacrificing coverage. POR relies on an equivalence relation between executions where e.g., two executions are equivalent if one can be obtained from the other by swapping consecutive independent (non-conflicting) execution steps. It guarantees that at least one execution from each equivalence class is explored. \emph{Optimal} POR techniques explore exactly one execution from each equivalence class. Beyond this classic notion of optimality, POR techniques may aim for optimality by avoiding visiting states from which \nver{the exploration is blocked}. 
{\em Dynamic} partial order reduction (DPOR)~\cite{DBLP:conf/popl/FlanaganG05} has been introduced to explore the execution space (and tracking the equivalence relation between executions) on-the-fly without relying on a-priori static analyses. This is typically coupled with \emph{stateless} model checking (SMC)~\cite{DBLP:conf/popl/Godefroid97} which explores executions of a program without storing visited states, thereby, avoiding excessive memory consumption. 

There is a large body of work on (D)POR techniques that address their soundness when checking a certain class of specifications for a certain class of programs, as well as their completeness and their theoretical optimality (see Section~\ref{sec:related}). Most often these works consider shared memory concurrent programs executing under a strongly consistent memory model. 

In the last few years, some works have studied DPOR in the case of shared memory programs running under weak memory models such as TSO or Release-Acquire, e.g.~\cite{DBLP:conf/cav/AbdullaAJL16,DBLP:journals/acta/AbdullaAAJLS17,DBLP:journals/pacmpl/AbdullaAJN18,DBLP:conf/pldi/Kokologiannakis19}. While these algorithms are sound and complete, they have exponential space complexity when they are optimal. 
More recently, \citet{DBLP:journals/pacmpl/Kokologiannakis22} defined a DPOR algorithm that has a polynomial space complexity, in addition of being sound, complete and optimal. This algorithm can be applied for a range of shared memory models. 

While the works mentioned above concern shared memory programs, we are not aware of any published work addressing the case of database transactional programs running under weak isolation levels. In this paper, we address this case and propose new stateless model checking algorithms relying on DPOR techniques for database-backed applications.  \nver{We assume that all the transactions in an application execute under the \emph{same} isolation level, which happens quite frequently in practice (as mentioned above, most database applications are run on the default isolation level of the database).}
Our work generalizes the approach introduced by \cite{DBLP:journals/pacmpl/Kokologiannakis22}. However, this generalization to the transactional case, covering the most relevant isolation levels, is not a straightforward adaptation of \cite{DBLP:journals/pacmpl/Kokologiannakis22}. Ensuring optimality while preserving the other properties, e.g., completeness and polynomial memory complexity, is very challenging.
Next, we explain the main steps and features of our work.  

\vspace{.5mm}
\noindent
\textbf{Formalizing Isolation Levels.}
Our algorithms rely on the axiomatic definitions of isolation
levels introduced by \citet{DBLP:journals/pacmpl/BiswasE19}. These
definitions use logical constraints called \emph{axioms} to characterize the
set of executions of a database (e.g., key-value store) that conform to a particular isolation
level (extensible to SQL queries~\citep{DBLP:journals/pacmpl/BiswasKVEL21}). These constraints refer to a specific set of
relations between events/transactions in an execution that describe control-flow
or data-flow dependencies: a program order $\po$ between events in the same
transaction, a session order $\so$ between transactions in the same session\footnote{A
session is a sequential interface to the storage system. It corresponds to what
is also called a \textit{connection}.}, and a write-read $\wro$ (read-from) relation that
associates each read event with a transaction that writes the value returned by
the read. These relations along with the events in an
execution are called a \emph{history}. 
A history describes only the
interaction with the database, omitting application-side events (e.g., computing
values written to the database). 

\vspace{.5mm}
\noindent
\textbf{Execution Equivalence.} DPOR algorithms are parametrized by an equivalence relation on executions, most often, Mazurkiewicz equivalence~\cite{DBLP:conf/ac/Mazurkiewicz86}. In this work, we consider a weaker equivalence relation, also known as \emph{read-from equivalence}~\cite{DBLP:journals/pacmpl/ChalupaCPSV18,DBLP:journals/pacmpl/AbdullaAJN18,DBLP:journals/pacmpl/AbdullaAJLNS19,DBLP:conf/pldi/Kokologiannakis19,DBLP:conf/asplos/Kokologiannakis20,DBLP:journals/pacmpl/Kokologiannakis22}, which considers that two executions are equivalent when their histories are precisely the same (they contain the same set of events, and the relations $\po$, $\so$, and $\wro$ are the same). 
In general, reads-from equivalence is coarser than Mazurkiewicz equivalence, and its equivalence classes can be exponentially-smaller than Mazurkiewicz traces in certain cases~\cite{DBLP:journals/pacmpl/ChalupaCPSV18}.

\vspace{.5mm}
\noindent
\textbf{SMC Algorithms.}
Our SMC algorithms enumerate executions of a given program under a given isolation level $I$. They are \emph{sound}, i.e., enumerate only \emph{feasible} executions (admitted by the program under $I$), \emph{complete}, i.e., they output a representative of each read-from equivalence class, and \emph{optimal}, i.e., they output \emph{exactly one} complete execution from each read-from equivalence class.
For isolation levels weaker than and including Causal Consistency, they satisfy a notion of \emph{strong optimality} which says that additionally, the enumeration avoids states from which the execution is ``blocked'', i.e., it cannot be extended to a complete execution of the program. 
For Snapshot Isolation and Serializability, we show that \emph{there exists} no algorithm in the same class (to be discussed below) that can ensure such a strong notion of optimality. All the algorithms that we propose are polynomial space, as opposed to many DPOR algorithms introduced in the literature.

As a starting point, we define a generic class of SMC algorithms, called \emph{swapping based}, generalizing the approach adopted by \citep{DBLP:conf/pldi/Kokologiannakis19,DBLP:journals/pacmpl/Kokologiannakis22}, which enumerate histories of program executions. These algorithms focus on the interaction with the database assuming that 
the other steps in a transaction concern local variables visible only within the scope of the enclosing session.  Executions are extended according to a generic scheduler function $\genericNext$ and every read event produces several exploration branches, one for every write executed in the past that it can read from. Events in an execution can be swapped to produce new exploration ``roots'' that lead to different histories. Swapping events is required for completeness, to enumerate histories where a read $r$ reads from a write $w$ that is scheduled by $\genericNext$ after $r$. To ensure soundness, we restrict the definition of swapping so that it produces a history that is feasible by construction (extending an execution which is possibly infeasible may violate soundness). Such an algorithm is optimal w.r.t. the read-from equivalence when it enumerates each history exactly once. 

We define a concrete algorithm in this class that in particular, satisfies the stronger notion of optimality mentioned above for every isolation level $I$ which is \emph{prefix-closed} and \emph{causally-extensible}, e.g.,  {\it Read Committed} and \textit{Causal Consistency}. Prefix-closure means that every prefix of a history that satisfies $I$, i.e., a subset of transactions and all their predecessors in the causal relation, i.e., $(\so\cup\wro)^+$, is also consistent with $I$, and causal extensibility means that any pending transaction in a history that satisfies $I$ can be extended with one more event to still satisfy $I$, and if this is a read event, then, it can read-from a transaction that precedes it in the causal relation. To ensure strong optimality, this algorithm uses a carefully chosen condition for restricting the application of event swaps, which makes the proof of completeness in particular, quite non-trivial.

We show that isolation levels such as Snapshot Isolation and Serializability are \emph{not} causally-extensible and that there exists no swapping based SMC algorithm which is sound, complete, and strongly optimal at the same time (independent of memory consumption bounds). This impossibility proof uses a program to show that any $\genericNext$ scheduler and any restriction on swaps would violate either completeness or strong optimality. However, we define an extension of the previous algorithm which satisfies the weaker notion of optimality, while preserving soundness, completeness, and polynomial space complexity. This algorithm will simply enumerate executions according to a weaker prefix-closed and causally-extensible isolation level, and filter executions according to the stronger isolation levels Snapshot Isolation and Serializability at the end, before outputting.

We implemented these algorithms in the Java Pathfinder (JPF) model checker~\cite{DBLP:conf/issta/VisserPK04}, and evaluated them on a number of challenging database-backed applications drawn from the literature of distributed systems and databases.  

Our contributions and outline are summarized as follows:
\vspace{-1mm}
\begin{itemize}
	\item[\S~\ref{sec:props}] identifies a class of isolation levels called prefix-closed and causally-extensible that admit efficient SMC.
	\item[\S~\ref{sec:algs}] defines a generic class of swapping based SMC algorithms based on DPOR which are parametrized by a given isolation level.
	\item[\S~\ref{sec:CC-algorithm}] defines a swapping based SMC algorithm which is sound, complete, strongly-optimal, and polynomial space, for any isolation level that is prefix-closed and causally-extensible.
	\item[\S~\ref{sec:ser}] shows that there exists no swapping based algorithm for Snapshot Isolation and Serializability, which is sound, complete, and strongly-optimal at the same time, and proposes a swapping based algorithm which satisfies ``plain'' optimality. 
	\item[\S~\ref{sec:exp}] reports on an implementation and evaluation of these algorithms. 
\vspace{-1mm}
\end{itemize}

Section~\ref{sec:prelims} recalls the formalization of isolation levels of Biswas and Enea~\cite{DBLP:journals/pacmpl/BiswasE19,DBLP:journals/pacmpl/BiswasKVEL21}, while Sections~\ref{sec:related} and~\ref{sec:conc} conclude with a discussion of related work and concluding remarks. \nver{Additional formalization, proofs, and experimental data can be found in the technical report~\cite{bouajjani2023dynamic}.}

\section{Transactional Programs}\label{sec:prelims}

\subsection{Program Syntax}
\begin{figure}
\vspace{-4mm}
\small
\begin{align*}
\key\in \Vars\quad \xvar\in\LVars
\end{align*}
\begin{minipage}[t]{.4\textwidth}
\vspace{-4mm}
\begin{align*}
\mathsf{Prog} &  \eqdef  \mathsf{Sess} \ \mid\  \mathsf{Sess}\,||\,\mathsf{Prog} \\
\mathsf{Sess} & \eqdef  \mathsf{Trans} \ \mid\  \mathsf{Trans}; \mathsf{Sess} \\
\mathsf{Trans} & \eqdef  \ibegin; \mathsf{Body}; \icommit\\
\end{align*}
\end{minipage}
\begin{minipage}[t]{.4\textwidth}
\vspace{-4mm}
\begin{align*}
\mathsf{Body} & \eqdef  \mathsf{Instr} \ \mid\  \mathsf{Instr}; \mathsf{Body} \\
\mathsf{Instr} & \eqdef  \mathsf{InstrDB} \ \mid\  a := e \mid\ \iif{\phi(\vec{a})}{\mathsf{Instr}} \\
\mathsf{InstrDB} & \eqdef \xvar := \iread(\key)  \ \mid\  \iwrite(\key,\xvar) \ \mid \ \iabort
\end{align*}
\end{minipage}
\vspace{-7mm}
\caption{Program syntax. The set of global variables is denoted by $\Vars$ while $\LVars$ denotes the set of local variables.
We use $\phi$ to denote Boolean expressions over local variables, and $e$ to denote expressions over local variables interpreted as values. We use $\vec{\cdot}$ to denote vectors of elements.}
\label{fig:syntax}
\vspace{-8mm}
\end{figure}

Figure~\ref{fig:syntax} lists the definition of a simple programming language
that we use to represent applications running on top of a database. A program is a set of \emph{sessions} running in parallel, each
session being composed of a sequence of \emph{transactions}. Each transaction is
delimited by $\ibegin$ and either $\icommit$ or $\iabort$ instructions, 
 and its body contains instructions that access the
database and manipulate a set $\LVars$ of local variables. 
We use symbols $a$, $b$, etc. to denote elements of $\LVars$.

For simplicity, we abstract the database state as a valuation to a set $\Vars$ of \emph{global} variables\footnote{In the context of a relational database, global variables correspond to fields/rows of a table while in the context of a key-value store, they correspond to keys.}, ranged over using $x$, $y$, etc. The instructions accessing the database correspond to reading the value of a global variable and storing it into a local variable $a$ ($a := \iread(x)$) , writing the value of a local variable $a$ to a global variable $x$ ($\iwrite(x,a)$), or an assignment to a local variable $a$ ($a := e$). The set of values of global or local variables is denoted by $\Vals$. Assignments to local variables use expressions $e$ over local variables, which are interpreted as values and whose syntax is left unspecified. Each of these instructions can be guarded by a Boolean condition $\phi(\vec{a})$ over a set of local variables $\vec{a}$ (their syntax is not important). Our results assume bounded programs, as usual in SMC algorithms, and therefore, we omit other constructs like $\mathtt{while}$ loops. 
SQL statements (SELECT, JOIN, UPDATE) manipulating relational tables can be compiled to reads or writes of variables representing rows in a table (see for instance,~\cite{DBLP:journals/pacmpl/RahmaniNDJ19,DBLP:journals/pacmpl/BiswasKVEL21}). 


\subsection{Isolation Levels}

We present the axiomatic framework introduced
by \citet{DBLP:journals/pacmpl/BiswasE19} for defining isolation levels. 
Isolation levels are defined as logical constraints, called \emph{axioms}, over \emph{histories}, which are an abstract representation of the interaction between a program and the database in an execution. 

\vspace{-2mm}
\subsubsection{Histories}

Programs interact with a database by issuing transactions formed of $\ibegin$, $\icommit$, $\iabort$, $\textsf{read}$ and $\textsf{write}$ instructions. The effect of executing one such instruction is represented using an \emph{event} $\langle e, \mathit{type} \rangle$ where $e$ is an \textit{identifier} and $\mathit{type} $ is a \textit{type}. There are five types of events: $\ebegin$, $\ecommit$, $\eabort$, $\erd{x}$ for reading the global variable $x$, and $\ewrt{x,v}$ for writing value $v$ to $x$. $\Events$ denotes the set of events.
For a read/write event $e$, we use $\mathit{var}(e)$ to denote the variable $x$. 

A \emph{transaction log} $\tup{t,E, \po_t}$ is an identifier $t$ and a finite set of events $E$ along with a strict total order $\po_t$ on $E$, called \emph{program order} (representing the order between instructions in the body of a transaction).
The minimal element of $\po_t$ is a $\ebegin$ event. A transaction log without neither a $\ecommit$ nor an $\eabort$ event is called \emph{pending}. Otherwise, it is called \emph{complete}. A complete transaction log with a $\ecommit$ event is called \textit{committed} and \textit{aborted} otherwise. If a $\ecommit$ or an $\eabort$ event occurs, then it is maximal in $\po_t$; $\ecommit$ and $\eabort$ cannot occur in the same log. The set $E$ of events in a transaction log $t$ is denoted by $\events{t}$. Note that a transaction is aborted because it executed an $\iabort$ instruction. Histories do not include transactions aborted by the database because their effect should not be visible to other transactions and the abort is not under the control of the program. For simplicity, we may use the term \emph{transaction} instead of transaction log.

Isolation levels differ in the values returned by read events which are not preceded by a write on the same variable in the same transaction. \nver{We assume in the following that every transaction in a program is executed under the same isolation level.} For every isolation level that we are aware of, 
%
if a read of a global variable $x$ is preceded by a write to $x$ in $\po_t$, then it should return the value written by the last write to $x$ before the read (w.r.t. $\po_t$). 


The set of $\erd{x}$ events in a transaction log $t$ that are \textit{not} preceded by a write to $x$ in $\po_t$, for some $x$, is denoted by $\readOp{t}$. 
Also, if $t$ does \textit{not} contain an $\eabort$ event, the set of $\ewrt{x,\_}$ events in $t$ that are \textit{not} followed by other writes to $x$ in $\po_t$, for some $x$, is denoted by $\writeOp{t}$. If a transaction contains multiple writes to the same variable, then only the last one (w.r.t. $\po_t$) can be visible to other transactions (w.r.t. any isolation level that we are aware of). If $t$ contains an abort event, then we define $\writeOp{t}$ to be the empty set. This is because the effect of aborted transactions (its set of writes) should not be visible to other transactions.
The extension to sets of transaction logs is defined as usual. 
Also, we say that a transaction log $t$ \emph{writes} $x$, denoted by $\writeVar{t}{x}$, when $\writeOp{t}$ contains some $\ewrt{x,\_}$ event. 

A \emph{history} contains a set of transaction logs (with distinct identifiers)
ordered by a (partial) \emph{session order} $\so$ that represents the order
between transactions in the same session. 
It also includes a
\emph{write-read} relation (also called read-from) that defines read values by associating each read to a transaction that wrote that value. Read events do \emph{not} contain a value, and their return value is defined as the value written by the transaction associated by the write-read relation.
\nver{Let $T$ be a set of transaction logs.}
For a write-read relation \oldver{$\wro\subseteq T\times \readOp{T}$}\nver{$\wro\subseteq \writeOp{T}\times \readOp{T}$} and variable $\key$, $\wro_\key$ is the restriction of $\wro$ to reads of $\key$, \oldver{$\wro_\key=\wro\cap (T\times \{ e\ |\ e \mbox{ is a }\erd{\key} \mbox{ event}\})$}\nver{$\wro_\key=\wro\cap (\writeOp{T}\times \{ e\ |\ e \mbox{ is a }\erd{\key} \mbox{ event}\})$}. We extend the relations $\wro$ and $\wro_\key$ to pairs of transactions by $\tup{t_1,t_2}\in \wro$, resp., $\tup{t_1,t_2}\in \wro_\key$, iff there exists \nver{a $\wrt{x}{\_}$ event $w$ in $t_1$ and} a $\rd{\key}$ event $r$ in $t_2$  s.t. \oldver{$\tup{t_1,e}\in \wro$, resp., $\tup{t_1,e}\in \wro_\key$}\nver{$\tup{w,r}\in \wro$, resp., $\tup{w, r}\in \wro_\key$.} \nver{Analogously, $\wro$ and $\wro_\key$ can be extended to tuples formed of a transaction (containing a write) and a read event.} We say that the transaction log $t_1$ is \emph{read} by the transaction log $t_2$ when $\tup{t_1,t_2}\in \wro$. 

\vspace{-1mm}
\begin{definition}
A \emph{history} $\tup{T, \so, \wro}$ is a set of transaction logs $T$ along with a strict partial \emph{session order} $\so$, and a 
\emph{write-read} relation \oldver{$\wro\subseteq T\times \readOp{T}$}\nver{$\wro\subseteq \writeOp{T}\times \readOp{T}$} such that
\vspace{-1mm}
\begin{itemize}
	\item the inverse of $\wro$ is a total function, 
	\item \oldver{if $(t,e)\in \wro$, where $e$ is a $\erd{x}$ event, then $\writeOp{t}$ contains a $\ewrt{x,\_}$ event,}\nver{if $(w,r)\in \wro$, then $w$ and $r$ are a write and respectively, a read, of the same variable,} and
	\item $\so\cup\wro$ is acyclic (here we use the extension of $\wro$ to pairs of transactions).
\end{itemize}
\vspace{-2mm}
\end{definition}

Every history includes a distinguished transaction writing the initial values of all global variables. This transaction precedes all the other transactions in $\so$. We use $\hist$, $\hist_1$, $\hist_2$, $\ldots$ to range over histories. 

The set of transaction logs $T$ in a history $\hist=\tup{T, \so, \wro}$ is denoted by $\tlogs{\hist}$, and $\events{\hist}$ is the union of $\events{t}$ for $t\in T$. For a history $\hist$ and an event $e$ in $\hist$, $\trans{h}{e}$ is the transaction $t$ in $\hist$ that contains $e$. Also, $\writeOp{\hist}=\bigcup_{t\in \tlogs{h}}\writeOp{t}$ and $\readOp{\hist}=\bigcup_{t\in \tlogs{h}}\readOp{t}$.

We extend $\so$\oldver{ and $\wro$} to pairs of events by $(e_1,e_2)\in \so$ if $(\trans{h}{e_1},\trans{h}{e_2})\in\so$. Also, $\po=\bigcup_{t\in T} \po_t$.

%
%
%

\vspace{-1mm}
\subsubsection{Axiomatic Framework}

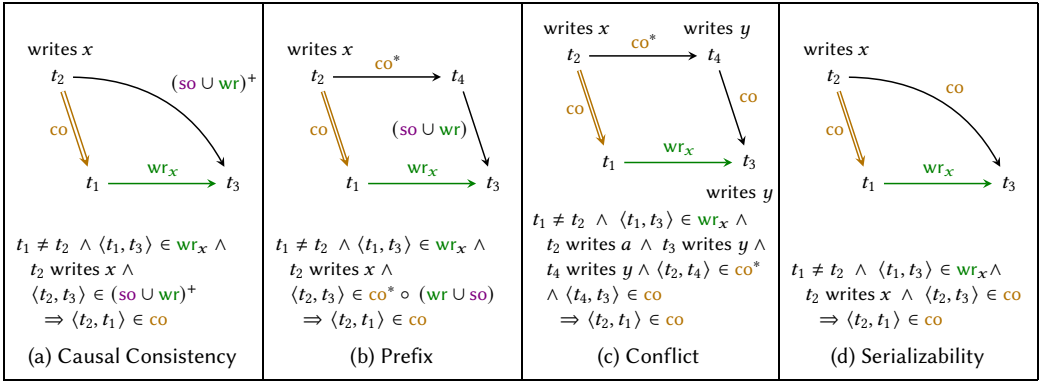
\begin{figure}[t]
	\resizebox{\textwidth}{!}{
		\footnotesize
		\begin{tabular}{|c|c|c|c|}
			\hline & & & \\ [-2mm]
			\begin{subfigure}[b]{.24\textwidth}
				\centering
				\begin{tikzpicture}[->,>=stealth,shorten >=1pt,auto,node distance=4cm,
					semithick, transform shape]
					\node[transaction state, text=black] at (0,0)       (t_1)           {$t_1$};
					\node[transaction state] at (2,0)       (t_3)           {$t_3$};
					\node[transaction state, text=black,label={above:\textcolor{black}{$\writeVar{ }{\key}$}}] at (-.5,1.5) (t_2) {$t_2$};
					\path (t_1) edge[wrColor] node {$\wro_x$} (t_3);
					\path (t_2) edge[bend left] node {$(\so \cup \wro)^+$} (t_3);
					\path (t_2) edge[left,double equal sign distance,coColor] node {$\co$} (t_1);
				\end{tikzpicture}
				\parbox{\textwidth}{
					$ $ \\

					$t_1\neq t_2\ \land \tup{t_1,t_3}\in \wro_x \ \land $
					
					\hspace{2mm}$\writeVar{t_2}{x}\ \land$ 

					\hspace{2mm}$\tup{t_2,t_3}\in(\so \cup \wro)^+$
					
					\hspace{4mm}$\implies \tup{t_2,t_1}\in\co$
				}
				
				\caption{$\mathsf{Causal\ Consistency}$}
				\label{cc_def}
			\end{subfigure}
			
			&     
			\begin{subfigure}[b]{.24\textwidth}
				\centering
				\begin{tikzpicture}[->,>=stealth,shorten >=1pt,auto,node distance=4cm,
				 semithick, transform shape]
				\node[transaction state, text=black] at (0,0)       (t_1)           {$t_1$};
				\node[transaction state] at (2,0)       (t_3)           {$t_3$};
				\node[transaction state, text=black,label={above:\textcolor{black}{$\writeVar{ }{x}$}}] at (-0.5,1.5) (t_2) {$t_2$};
				\node[transaction state] at (1.5,1.5) (t_4) {$t_4$};
				\path (t_1) edge[wrColor] node {$\wro_x$} (t_3);
				\path (t_2) edge[] node {$\co^*$} (t_4);
				\path (t_4) edge[left] node {$(\so \cup \wro)$} (t_3);
				\path (t_2) edge[left,double, coColor] node {$\co$} (t_1);
				\end{tikzpicture}
				\parbox{\textwidth}{
				
				$ $ \\

				$t_1\neq t_2\ \land \tup{t_1,t_3}\in \wro_x \ \land $
				
				\hspace{2mm}$\writeVar{t_2}{x}\ \land$ 
				
				\hspace{2mm}$\tup{t_2,t_3}\in\co^*\circ\,(\wro\cup\so)$
				
				\hspace{4mm}$\implies \tup{t_2,t_1}\in\co$
				}
				
				\caption{$\mathsf{Prefix}$}
				\label{pre_def}
				\end{subfigure}
			&
			\hspace{-1mm}

			\begin{subfigure}[b]{.24\textwidth}
			    \centering
			    \begin{tikzpicture}[->,>=stealth,shorten >=1pt,auto,node distance=4cm,
			      semithick, transform shape]
			     \node[transaction state, text=black] at (0,0)       (t_1)           {$t_1$};
			     \node[transaction state, label={[xshift=-5]below:$\writeVar{ }{\yvar}$}] at (2,0)       (t_3)           {$t_3$};
			     \node[transaction state, text=black,label={above:$\writeVar{}{x}$}] at (-.5,1.5) (t_2) {$t_2$};
			     \node[transaction state, label={[yshift=-2]above:{$\writeVar{}{\yvar}$}}] at (1.5,1.5) (t_4) {$t_4$};
			     \path (t_1) edge[wrColor] node {$\wro_x$} (t_3);
			     \path (t_2) edge node {$\co^*$} (t_4);
			     \path (t_4) edge node {$\co$} (t_3);
			     \path (t_2) edge[left,double, coColor] node {$\co$} (t_1);
			    \end{tikzpicture}
			    \parbox{\textwidth}{

			     $t_1\neq t_2\ \land \ \tup{t_1,t_3}\in \wro_x \ \land$
			     
			     \hspace{2mm}$ \writeVar{t_2}{\xvar}\ \land \ \writeVar{t_3}{\yvar}\ \land$ 
			     
			     \hspace{2mm}$ \writeVar{t_4}{\yvar} \land  \tup{t_2,t_4}\in\co^*$ 
			     
			     \hspace{2mm}$\land \ \tup{t_4,t_3}\in\co$
			     
			     \hspace{4mm}$\implies \tup{t_2,t_1}\in\co$
			    }
			    
			    \caption{$\mathsf{Conflict}$}
			    \label{confl_def}
			   \end{subfigure}
			
			   & 
    
			\begin{subfigure}[b]{.24\textwidth}
				\centering
				\begin{tikzpicture}[->,>=stealth,shorten >=1pt,auto,node distance=4cm,
					semithick, transform shape]
					\node[transaction state, text=black] at (0,0)       (t_1)           {$t_1$};
					\node[transaction state] at (2,0)       (t_3)           {$t_3$};
					\node[transaction state, text=black, label={above:\textcolor{black}{$\writeVar{ }{x}$}}] at (-.5,1.5) (t_2) {$t_2$};
					\path (t_1) edge[wrColor] node {$\wro_x$} (t_3);
					\path (t_2) edge[bend left] node {$\co$} (t_3);
					\path (t_2) edge[left,double, coColor] node {$\co$} (t_1);
				\end{tikzpicture}
				\parbox{\textwidth}{
					$ $ \\
					$ $ \\

					$t_1\neq t_2\ \land \ \tup{t_1,t_3}\in \wro_x \land$
					
					\hspace{2mm}$ \writeVar{t_2}{x}\ \land \ \tup{t_2,t_3}\in\co$ 
										
					\hspace{4mm}$\implies \tup{t_2,t_1}\in\co$
				}
				
				\caption{$\mathsf{Serializability}$}
				\label{ser_def}
			\end{subfigure}
			
			\\ \hline			  
		\end{tabular}
	}
	\vspace{-2mm}
	\caption{Axioms defining isolations levels (all logical variables representing transactions, e.g., $t_1$, are universally quantified). The reflexive and transitive, resp., transitive, closure of a relation $rel$ is denoted by $rel^*$, resp., $rel^+$. Also, $\circ$ denotes the composition of two relations, i.e., $rel_1 \circ rel_2 = \{\tup{a, b} | \exists c. \tup{a, c} \in rel_1 \land \tup{c, b} \in rel_2\}$.}
	\label{fig:consistency_defs}
	\vspace{-4mm}
\end{figure}

A history satisfies a certain isolation level if there is a strict total order $\co$ on its transactions, called \emph{commit order}, which extends the write-read relation and the session order, and which satisfies certain properties. These properties, called \emph{axioms}, relate the commit order with the $\so$ and $\wro$ relations in a history and are
defined as 
first-order formulas of the form:
\label{eq:axioms}
\vspace{-1mm}
\begin{align}
  & \forall \key,\ \forall \tr_1\neq \tr_2,\ \forall \tr_3.\ \nonumber\\
  & \hspace{3mm}  \tup{\tr_1,t_3}\in \wro_\key \land \writeVar{\tr_2}{\key} \land \phi(\tr_2,\tr_3) \implies \tup{\tr_2,\tr_1}\in\co \label{eq:axiom}
\end{align}
\vspace{-5mm}

where $\phi$ is a property relating $\tr_2$ and $\tau$ (i.e., the read or the
transaction reading from $\tr_1$) that varies from one axiom to another.\footnote{These formulas are interpreted on tuples $\tup{\hist,\co}$ of a history $\hist$ and a commit order $\co$ on the transactions in $\hist$ as usual.} Note that an aborted transaction $t$ cannot take the role of $\tr_1$ nor $\tr_2$ in equation \ref{eq:axiom} as the set $\writeOp{t}$ is empty. Intuitively, this axiom schema states the following: in order for $\tau$ to read specifically $t_1$'s write on $k$, it must be the case that every $t_2$ that also writes $k$ and satisfies $\phi(t_2,\tau)$ was committed before $t_1$. 
The property $\phi$ relates $\tr_2$ and $\tau$ using the relations in a history and the commit order. 
Figure~\ref{fig:consistency_defs} shows two axioms which correspond to their homonymous isolation levels: \textit{Causal Consistency} ($\CC$) and \textit{Serializability} ($\SER$). The conjunction of the other two axioms Conflict and Prefix defines \textit{Snapshot Isolation} ($\SI$). \textit{Read Atomic} ($\RA$) is a weakening of $\CC$ where $(\so \cup \wro)^+$ is replaced with $\so \cup \wro$. \textit{Read Committed} ($\RC$) is defined similarly\appendixver{(see Appendix~\ref{sec:def:RA-RC})}.
Note that $\SER$ is stronger than $\SI$ (i.e., every history satisfying $\SER$ satisfies $\SI$ as well), $\SI$ is stronger than $\CC$, $\CC$ is stronger than $\RA$, and $\RA$ is stronger than $\RC$.

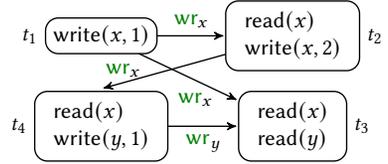
\begin{wrapfigure}{r}{0.4\textwidth}
\vspace{-5mm}
   \centering
\resizebox{0.38\textwidth}{!}{
\begin{tikzpicture}[->,>=stealth',shorten >=1pt,auto,node distance=3cm,
 semithick, transform shape]
 \node[draw, rounded corners=2mm,outer sep=0, label={[font=\small]180:$t_1$}] (t1) at (0, 1.4) {$\ewrt{x,1}$};
\node[draw, rounded corners=2mm,outer sep=0,label={[font=\small]0:$t_2$}] (t2) at (3, 1.4) {\begin{tabular}{l} $\erd{x}$ \\ $\ewrt{x,2}$ \end{tabular}};
\node[draw, rounded corners=2mm,outer sep=0,label={[font=\small]0:$t_3$}] (t3) at (3, 0) {\begin{tabular}{l} $\erd{x}$ \\ $\erd{y}$ \end{tabular}};
\node[draw, rounded corners=2mm,outer sep=0,label={[font=\small]180:$t_4$}] (t4) at (0, 0) {\begin{tabular}{l} $\erd{x}$ \\ $\ewrt{y,1}$\end{tabular}};

\path (t1) edge[above] node[yshift=0,xshift=0] {$\wro_x$} (t2);

\path (t1) edge[below] node[yshift=-3,xshift=3] {$\wro_x$} (t3);

\path (t2) edge[above] node[yshift=-6,xshift=-18] {$\wro_x$} (0,0.58);

\path (t4) edge[below] node[yshift=0,xshift=0] {$\wro_y$} (t3);
\end{tikzpicture}  
}
\vspace{-2mm}
  \caption{Causal Consistency violation. Boxes group events from the same transaction.}
  \label{counter_example:1}
\vspace{-4mm}
\end{wrapfigure}
For instance, the axiom defining Causal Consistency~\cite{DBLP:journals/cacm/Lamport78} states that for any transaction $\tr_1$ writing a variable $x$ that is read in a transaction $\tr_3$, the set of $(\wro\cup \so)^+$ predecessors of $\tr_3$ writing $x$ must precede $\tr_1$ in commit order ($(\wro\cup \so)^+$ is usually called the \emph{causal} order). A violation of this axiom can be found in Figure~\ref{counter_example:1}: the transaction $\tr_2$ writing 2 to $x$ is a $(\wro\cup \so)^+$ predecessor of the transaction $\tr_3$ reading 1 from $x$ because the transaction $\tr_4$, writing 1 to $y$, reads $x$ from $\tr_2$ and $\tr_3$ reads $y$ from $\tr_4$. This implies that $\tr_2$ should precede in commit order the transaction $\tr_1$ writing 1 to $x$, which is inconsistent with the write-read relation ($\tr_2$ reads from $\tr_1$).

The $\mathsf{Serializability}$ axiom requires that for any transaction $\tr_1$ writing to a variable $x$ that is read in a transaction $\tr_3$, the set of $\co$ predecessors of $\tr_3$ writing $x$ must precede $\tr_1$ in commit order. This ensures that each transaction observes the effects of all the $\co$ predecessors. 


\begin{definition}
For an isolation level $I$ defined by a set of axioms $X$, a history
$\hist=\tup{T, \so, \wro}$ \emph{satisfies} $I$ iff there is a strict total
order $\co$ s.t. $\wro\cup\so\subseteq \co$ and $\tup{h,\co}$ satisfies $X$. 
 \label{axiom-criterion}
\end{definition}

A history that satisfies an isolation level $I$ is called $I$-consistent. For two isolation levels $I_1$ and $I_2$, $I_1$ is \emph{weaker than} $I_2$ when every $I_1$-consistent history is also $I_2$-consistent.

\subsection{Program Semantics}\label{ssec:semantics}

We define a small-step operational semantics for transactional programs, which is parametrized by an isolation level $I$. The semantics keeps a history of previously executed database accesses in order to maintain consistency with $I$. 

For readability, we define a program as a partial function $\prog:\mathsf{SessId}\rightharpoonup \mathsf{Sess}$ that associates session identifiers in $\mathsf{SessId}$ with concrete code as defined in Figure~\ref{fig:syntax} (i.e., sequences of transactions). Similarly, the session order $\so$ in a history is defined as a partial function $\so:\mathsf{SessId}\rightharpoonup \mathsf{Tlogs}^*$ that associates session identifiers with sequences of transaction logs. Two transaction logs are ordered by $\so$ if one occurs before the other in some sequence $\so(j)$ with 
$j\in \mathsf{SessId}$.

The operational semantics is defined as a transition relation $\Rightarrow_I$ between \emph{configurations}, which are defined as tuples containing the following:
\begin{itemize}
	\item history $\hist$ storing the events generated by database accesses executed in the past, 
	\item a valuation map $\vec{\gamma}$ that records local variable values in the current transaction of each session ($\vec{\gamma}$ associates identifiers of sessions with valuations of local variables),
	\item a map $\vec{B}$ that stores the code of each live transaction (mapping session identifiers to code), 
	\item sessions/transactions $\prog$ that remain to be executed from the original program.
\end{itemize}

The relation $\Rightarrow_I$ is defined using a set of rules as expected\appendixver{(their precise definition can be found in Appendix~\ref{app:rules})}. 
Starting a new transaction in a session $j$ is enabled as long as this session has no live transactions ($\vec{\mathsf{B}}(j) = \epsilon$) and results in adding a transaction log with a single $\ebegin$ event to the history and scheduling the body of the transaction (adding it to $\vec{\mathsf{B}}(j)$). Local steps, i.e., checking a Boolean condition or computation with local variables, use the local variable valuations and advance the code as expected. Read instructions of some global variable $x$ can have two possible behaviors: (1) if the read follows a write on $x$ in the same transaction, then it returns the value written by the last write on $x$ in that transaction, and (2) otherwise, the read reads from another transaction $\tr'$ which is chosen non-deterministically as long as extending the current history with the write-read dependency associated to this choice leads to a history that still satisfies $I$. \nver{Depending on the isolation level, there may not exist a transaction $\tr'$ the read can read from\appendixver{(see Appendix~\ref{app:impossibility} for a concrete example)}.} For other instructions, e.g., $\icommit$ and $\iabort$, the history is simply extended with the corresponding events while ending the transaction execution in the case of $\iabort$.


An \emph{initial} configuration for program $\prog$ contains the program $\prog$, a history $\hist=\tup{\{\tr_0\},\emptyset,\emptyset}$ where $\tr_0$ is a transaction log containing writes that write the initial value for all variables, and empty current transaction code ($\mathsf{B}=\epsilon$). 
An execution of a program $\prog$ under an isolation level $I$ is a sequence of configurations $c_0 c_1\ldots c_n$ where $c_0$ is an initial configuration for $\prog$, and $c_m\Rightarrow_I c_{m+1}$, for every $0\leq m < n$. We say that $c_n$ is \emph{$I$-reachable} from $c_0$.
The history of such an execution is the history $\hist$ in the last configuration $c_n$. 
A configuration is called \emph{final} if it contains the empty program ($\prog=\emptyset$).
Let $\histOf[I]{\prog}$ denote the set of all histories of an execution of $\prog$ under $I$ that ends in a final configuration.

\section{Prefix-Closed and Causally-Extensible Isolation Levels}\label{sec:props}

We define two properties of isolation levels, prefix-closure and causal extensibility, which enable efficient DPOR algorithms (as shown in Section~\ref{sec:CC-algorithm}).


\subsection{Prefix Closure}

For a relation $R\subseteq A\times A$, the restriction of $R$ to $A'\times A'$, denoted by $R\downarrow A'\times A'$, is defined by $\{(a,b): (a,b)\in R, a,b\in A'\}$. Also, a set $A'$ is called $R$-downward closed when it contains $a\in A$ every time it contains some $b\in A$ with $(a,b)\in R$.

\begin{figure}[t]
%
	\centering
	\begin{subfigure}[b]{.30\textwidth}
		\resizebox{\textwidth}{!}{
			\begin{tikzpicture}[->,>=stealth',shorten >=1pt,auto,node distance=3cm,
				semithick, transform shape]
				\node[draw, rounded corners=2mm,outer sep=0] (t1) at (-3.25, 0) {\begin{tabular}{l} $\init$ \end{tabular}};
				\node[draw, rounded corners=2mm,outer sep=0] (t2) at (-3.25, -1.6) {\begin{tabular}{l} 
						$\rd{x}$ \\ $\rd{y}$
				\end{tabular}};
				\node[draw, rounded corners=2mm,outer sep=0] (t3) at (0, 0) {\begin{tabular}{l} 
						$\wrt{x}{2}$
				\end{tabular}};	
				\node[draw, rounded corners=2mm,outer sep=0] (t4) at (0, -1.6)
				{\begin{tabular}{l} 
						$\rd{x}$
				\end{tabular}};			
				
				\path (t1.south west) -- (t1.south) coordinate[pos=0.67] (t1x);
				\path (t2.north west) -- (t2.north) coordinate[pos=0.67] (t2x);
				\path (t3.north west) -- (t3.west) coordinate[pos=0.67] (t3x);
				
				\path (t3.west) edge [below] node[right] {$\wro_x$} (t2.north east);
				\path (t1.south) edge [above] node[left] {$\so \cap \wro_y$} (t2.north);
				\path (t1.east) edge [above] node[above] {$\so$} (t3.west);
				\path (t3.south) edge node[right] {$\wro_x$} (t4.north);
				\path (t1.south east) edge node[below] {$\so$} (t4.north west);
			\end{tikzpicture}
			
		}
		\caption{A history.}
		\label{fig:prefix:a}
	\end{subfigure}
	\hspace{.15cm}
	\centering
	\begin{subfigure}[b]{.30\textwidth}
		\resizebox{\textwidth}{!}{
			\begin{tikzpicture}[->,>=stealth',shorten >=1pt,auto,node distance=3cm,
				semithick, transform shape]
				\node[draw, rounded corners=2mm,outer sep=0] (t1) at (-3.25, 0) {\begin{tabular}{l} $\init$ \end{tabular}};
				\node[draw, rounded corners=2mm,outer sep=0] (t2) at (-3.25, -1.6) {\begin{tabular}{l} 
						$\rd{x}$ \\ $\rd{y}$
				\end{tabular}};
				\node[draw, rounded corners=2mm,outer sep=0] (t3) at (0, 0) {\begin{tabular}{l} 
						$\wrt{x}{2}$
				\end{tabular}};			
				
				\path (t1.south west) -- (t1.south) coordinate[pos=0.67] (t1x);
				\path (t2.north west) -- (t2.north) coordinate[pos=0.67] (t2x);
				\path (t3.north west) -- (t3.west) coordinate[pos=0.67] (t3x);
				
				\path (t3.west) edge [below] node[right] {$\wro_x$} (t2.north east);
				\path (t1.south) edge [above] node[left] {$\so \cap \wro_y$} (t2.north);
				\path (t1.east) edge [above] node[above] {$\so$} (t3.west);

			\end{tikzpicture}
		}
		\caption{A prefix.}
		\label{fig:prefix:b}
	\end{subfigure}
	\hspace{.15cm}
	\centering
	\begin{subfigure}[b]{.30\textwidth}
		\resizebox{\textwidth}{!}{
			\begin{tikzpicture}[->,>=stealth',shorten >=1pt,auto,node distance=3cm,
				semithick, transform shape]
				\node[draw, rounded corners=2mm,outer sep=0] (t1) at (-3.25, 0) {\begin{tabular}{l} $\init$ \end{tabular}};
				\node[draw, rounded corners=2mm,outer sep=0, ] (t2) at (-3.25, -1.4) {\begin{tabular}{l} 
						$\rd{x}$ \\ $\rd{y}$
				\end{tabular}};
				\node[draw, rounded corners=2mm,outer sep=0] (t4) at (0, -1.4)
				{\begin{tabular}{l} 
						$\rd{x}$
				\end{tabular}};			
				
				\path (t1.south west) -- (t1.south) coordinate[pos=0.67] (t1x);
				\path (t2.north west) -- (t2.north) coordinate[pos=0.67] (t2x);
				\path (t3.north west) -- (t3.west) coordinate[pos=0.67] (t3x);
				
				\path (t1.south east) edge node[below] {$\so$} (t4.north west);
				\path (t1.south) edge [above] node[left] {$\so \cap \wro_y$} (t2.north);
			\end{tikzpicture}
			
		}
		\caption{Not a prefix.}
		\label{fig:prefix:c}
	\end{subfigure}
	\vspace{-2mm}
	\caption{Explaining the notion of prefix of a history. $\init$ denotes the transaction log writing initial values. Boxes group events from the same transaction.}
	\vspace{-4mm}
\end{figure}
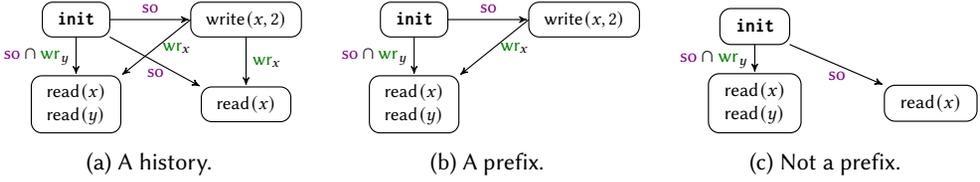

A \emph{prefix} of a transaction log $\tup{t,E, \po_t}$ is a transaction log $\tup{t,E', \po_t \downarrow E'\times E'}$ such that $E'$ is $\po_t$-downward closed. 
A \emph{prefix} of a history $\hist=\tup{T, \so, \wro}$ is a history $\hist'=\tup{T',\so\downarrow T'\times T',\wro\downarrow T'\times T'}$ such that every transaction log in $T'$ is a prefix of a different transaction log in $T$ but carrying the same id, $\events{\hist'}\subseteq\events{\hist}$, and $\events{\hist'}$ is $(\po\cup \so \cup \wro)^*$-downward closed.
For example, the history pictured in Fig.~\ref{fig:prefix:b} is a prefix of the one in Fig.~\ref{fig:prefix:a} while the history in Fig.~\ref{fig:prefix:c} is not. The transactions on the bottom of Fig.~\ref{fig:prefix:c} have a $\wro$ predecessor in Fig.~\ref{fig:prefix:a} which is not included.

\vspace{-1.5mm}
\begin{definition}
An isolation level $I$ is called \emph{prefix-closed} when every prefix of an $I$-consistent history is also $I$-consistent.
\vspace{-1.5mm}
\end{definition}


Every isolation level $I$ discussed above is prefix-closed because if a history $\hist$ is $I$-consistent with a commit order $\co$, then the restriction of $\co$ to the transactions that occur in a prefix $\hist'$ of $\hist$ satisfies the corresponding axiom(s) when interpreted over $\hist'$. 

\vspace{-2mm}
\begin{theorem}
\textit{Read Committed}, \textit{Read Atomic}, \textit{Causal Consistency}, \textit{Snapshot Isolation}, and \textit{Serializability} are prefix closed.
\vspace{-1.5mm}
\end{theorem}

\subsection{Causal Extensibility}\label{ssec:causal_ext}

\begin{figure}[t]
	\begin{minipage}[b]{0.66\textwidth}
	
		\centering
		\begin{subfigure}[b]{.46\textwidth}
			\resizebox{\textwidth}{!}{
				\begin{tikzpicture}[->,>=stealth',shorten >=1pt,auto,node distance=3cm,
					semithick, transform shape]
					\node[draw, rounded corners=2mm,outer sep=0] (t1) at (-3.25, 0) {\begin{tabular}{l} $\init$ \end{tabular}};
					\node[draw, rounded corners=2mm,outer sep=0] (t2) at (-3.25, -1.6) {\begin{tabular}{l} 
							$\rd{x}$ \\ \textbf{\textcolor{blue}{$\rd{y}$}} 
					\end{tabular}};
					\node[draw, rounded corners=2mm,outer sep=0] (t3) at (0, 0) {\begin{tabular}{l} 
							$\wrt{x}{2}$
					\end{tabular}};			
					
					\path (t1.south west) -- (t1.south) coordinate[pos=0.67] (t1x);
					\path (t2.north west) -- (t2.north) coordinate[pos=0.67] (t2x);
					\path (t3.north west) -- (t3.west) coordinate[pos=0.67] (t3x);
					
					\path (t3.west) edge [below] node[right] {$\wro_x$} (t2.north east);
					\path (t1.south) edge [above] node[left] {$\so$} (t2.north);
					\path (t1.east) edge [above] node[above] {$\so$} (t3.west);
				\end{tikzpicture}  
				
			}
			\caption{Extensible history.}
			\label{fig:maxclosed:a}
		\end{subfigure}
		\hspace{.1cm}
		\centering
		\begin{subfigure}[b]{.5\textwidth}
			\resizebox{\textwidth}{!}{
				\begin{tikzpicture}[->,>=stealth',shorten >=1pt,auto,node distance=3cm,
					semithick, transform shape]
					\node[draw, rounded corners=2mm,outer sep=0] (t1) at (-3.25, 0) {\begin{tabular}{l} $\init$ \end{tabular}};
					\node[draw, rounded corners=2mm,outer sep=0] (t2) at (-3.25, -1.6) {\begin{tabular}{l} 
							$\rd{x}$ \\ $\rd{y}$
					\end{tabular}};
					\node[draw, rounded corners=2mm,outer sep=0] (t3) at (0, 0) {\begin{tabular}{l} 
							$\wrt{x}{2}$\\
							\textbf{\textcolor{blue}{$\wrt{y}{2}$}}
					\end{tabular}};			
					
					\path (t1.south west) -- (t1.south) coordinate[pos=0.67] (t1x);
					\path (t2.north west) -- (t2.north) coordinate[pos=0.67] (t2x);
					\path (t3.north west) -- (t3.west) coordinate[pos=0.67] (t3x);
					
					\path (t3.west) edge [below] node[right] {$\wro_x$} (t2.north east);
					\path (t1.south) edge [above] node[left] {$\so \cap \wro_y$} (t2.north);
					\path (t1.east) edge [above] node[above] {$\so$} (t3.west);
				\end{tikzpicture}
				
			}
			\caption{Non-extensible history.}
			\label{fig:maxclosed:b}
		\end{subfigure}
		\centering
		\vspace{-2mm}
		\caption{Explaining causal extensibility. $\init$ denotes the transaction log writing initial values. Boxes group events from the same transaction.
		}
		\vspace{-2mm}

		\label{fig:maxclosed}

	\end{minipage}
	\hfill
	\begin{minipage}[b]{0.32\textwidth}
		\resizebox{\textwidth}{!}{
			\begin{tikzpicture}[->,>=stealth',shorten >=1pt,auto,node distance=3cm,
				semithick, transform shape]
				\node[draw, rounded corners=2mm,outer sep=0] (t1) at (-1.5, -0.25) {\begin{tabular}{l} $\init$ \end{tabular}};
				\node[draw, rounded corners=2mm,outer sep=0] (t2) at (-3, -2) {\begin{tabular}{l} 
						$\wrt{z}{1}$ \\ $\rd{x}$ \\ $\wrt{y}{1}$ 
				\end{tabular}};
				\node[draw, rounded corners=2mm,outer sep=0] (t3) at (0, -2) {\begin{tabular}{l} 
						$\wrt{z}{2}$ \\ $\rd{y}$ \\ 	\textbf{\textcolor{blue}{$\wrt{x}{2}$}}
				\end{tabular}};		
				
				\path (t1.south west) -- (t1.south) coordinate[pos=0.67] (t1sw);
				\path (t1.south east) -- (t1.south) coordinate[pos=0.67] (t1se);
				\path (t2.north east) -- (t2.north) coordinate[pos=0.67] (t2x);
				\path (t3.north west) -- (t3.north) coordinate[pos=0.67] (t3x);
				
				\path (t1sw) edge [left] node {$\wro_x$} (t2x);
				\path (t1se) edge [right] node {$\wro_y$} (t3x);
			\end{tikzpicture}  
		}

	\vspace{-2mm}

	\caption{A counter-example to causal extensibility for $\SI$ and $\SER$. 
	The $\so$-edges from $\init$ to the other transactions are omitted for legibility.}
	\vspace{-2mm}

	\label{fig:non-causally-extensible}

	\end{minipage}
\vspace{-2mm}
\end{figure}

We start with an example to explain causal extensibility. Let us consider the histories $h_1$ and $h_2$ in Figures~\ref{fig:maxclosed:a} and \ref{fig:maxclosed:b}, respectively, \emph{without} the events $\erd{y}$ and $\ewrt{y,2}$ written in blue bold font. These histories satisfy Read Atomic. 
The history $h_1$ can be extended by adding the event $\rd{y}$ and the $\wro$ dependency $\wro(\init,\rd{y})$ while still satisfying Read Atomic.
On the other hand, the history $h_2$ \emph{can not} be extended with the event $\wrt{y}{2}$ while still satisfying Read Atomic. Intuitively, if the reading transaction on the bottom reads $x$ from the transaction on the right, then it should read $y$ from the same transaction because this is more ``recent'' than $\init$ w.r.t. session order. The essential difference between these two extensions is that the first concerns a transaction which is maximal in $(\so\cup\wro)^+$ while the second no. The extension of $\hist_2$ concerns the transaction on the right in Figure~\ref{fig:maxclosed:b} which is a $\wro$ predecessor of the reading transaction. Causal extensibility will require that at least the $(\so\cup\wro)^+$ maximal (pending) transactions can always be extended with any event while still preserving consistency. The restriction to $(\so\cup\wro)^+$ maximal transactions is intuitively related to the fact that transactions should not read from non-committed (pending) transactions, e.g., the reading transaction in $\hist_2$ should not read from the still pending transaction that writes $x$ and later $y$.

Formally, let  $\hist=\tup{T, \so, \wro}$ be a history. A transaction $t$ is called $(\so \cup \wro)^+$-maximal in $h$ if $h$ does not contain any transaction $t'$ such that $(t,t')\in (\so \cup \wro)^+$. We define a \emph{causal extension} of a pending transaction $t$ in $h$ with an event $e$ as \nver{a history $h'$} such that:
\vspace{-1mm}
\begin{itemize}
\item $e$ is added to $t$ as a maximal element of $\po_t$,
\item if $e$ is a read event and $t$ \emph{does not} contain a write to $\mathit{var}(e)$, then $\wro$ is extended with some tuple $(t',e)$ such that $(t', t) \in  (\so \cup \wro)^+$ \nver{in $h$} \nver{(if $e$ is a read event and $t$ \emph{does} contain a write to $\mathit{var}(e)$, then the value returned by $e$ is the value written by the latest write on $\mathit{var}(e)$ before $e$ in $t$; the definition of the return value in this case is unique and does not involve $\wro$ dependencies),}
\item the other elements of $\hist$ remain unchanged \nver{in $h'$}.
\vspace{-1mm}
\end{itemize}

\begin{figure}[t]
%
	\begin{subfigure}[t]{.28\textwidth}
		\resizebox{\textwidth}{!}{
			\begin{tikzpicture}[->,>=stealth',shorten >=1pt,auto,node distance=3cm,
				semithick, transform shape]
				\node[draw, rounded corners=2mm,outer sep=0] (t0) at (-1.5, 0) {\begin{tabular}{l} $\init$ \end{tabular}};
				\node[draw, rounded corners=2mm,outer sep=0, label={[font=\small]100:$t_1$}] (t1) at (-3, -1.6) {\begin{tabular}{l} 
					$\wrt{x}{1}$ \\ $\wrt{y}{1}$
				\end{tabular}};
				\node[draw, rounded corners=2mm,outer sep=0, label={[font=\small]50:$t_2$}] (t2) at (-3, -3.2) {\begin{tabular}{l} 
					$\wrt{x}{2}$
				\end{tabular}};
				\node[draw, rounded corners=2mm,outer sep=0, label={[font=\small]50:$t_3$}] (t3) at (0, -1.6) {\begin{tabular}{l} 
						$\wrt{x}{3}$
				\end{tabular}};	
				\node[draw, rounded corners=2mm,outer sep=0, label={[font=\small]50:$t_4$}] (t4) at (0, -3.2)
				{\begin{tabular}{l} 
						$\rd{y}$ \\ $\cdots$
				\end{tabular}};			
				
				\path (t1.south west) -- (t1.south) coordinate[pos=0.67] (t1x);
				\path (t2.north west) -- (t2.north) coordinate[pos=0.67] (t2x);
				\path (t3.north west) -- (t3.west) coordinate[pos=0.67] (t3x);
				
				\path (t3.south) edge node[right] {$\so$} (t4.north);
				\path (t1.south) edge node[left] {$\so$} (t2.north);
				\path (t1.south east) edge node[above] {$\wro_y$} (t4.north west);
				\path (t0.south) edge node[above] {$\so$} (t1.north);
				\path (t0.south) edge node[above] {$\so$} (t3.north);
			\end{tikzpicture}
			
		}
		\vspace{-5.5mm}
		\caption{History $\hist$.}
		\label{fig:extension:a}
	\end{subfigure}
	\hspace{.15cm}
	\centering
	\begin{subfigure}[t]{.28\textwidth}
		\resizebox{\textwidth}{!}{
			\begin{tikzpicture}[->,>=stealth',shorten >=1pt,auto,node distance=3cm,
				semithick, transform shape]
				\node[draw, rounded corners=2mm,outer sep=0] (t0) at (-1.5, 0) {\begin{tabular}{l} $\init$ \end{tabular}};
				\node[draw, rounded corners=2mm,outer sep=0, label={[font=\small]100:$t_1$}] (t1) at (-3, -1.6) {\begin{tabular}{l} 
						$\wrt{x}{1}$ \\ $\wrt{y}{1}$
				\end{tabular}};
				\node[draw, rounded corners=2mm,outer sep=0, label={[font=\small]50:$t_2$}] (t2) at (-3, -3.2) {\begin{tabular}{l} 
						$\wrt{x}{2}$
				\end{tabular}};
				\node[draw, rounded corners=2mm,outer sep=0, label={[font=\small]50:$t_3$}] (t3) at (0, -1.6) {\begin{tabular}{l} 
						$\wrt{x}{3}$
				\end{tabular}};	
				\node[draw, rounded corners=2mm,outer sep=0, label={[font=\small]50:$t_4$}] (t4) at (0, -3.2)
				{\begin{tabular}{l} 
						$\rd{y}$ \\ \textbf{\textcolor{blue}{$\rd{x}$}}
				\end{tabular}};			
				
				\path (t1.south west) -- (t1.south) coordinate[pos=0.67] (t1x);
				\path (t2.north west) -- (t2.north) coordinate[pos=0.67] (t2x);
				\path (t3.north west) -- (t3.west) coordinate[pos=0.67] (t3x);
				
				\path (t3.south) edge node[right] {$\so$} (t4.north);
				\path (t1.south) edge node[left] {$\so$} (t2.north);
				\path (t1.south east) edge node[above, xshift = 4mm,yshift=-1mm] {$\wro_y$, $\wro_x$} (t4.north west);
				\path (t0.south) edge node[above] {$\so$} (t1.north);
				\path (t0.south) edge node[above] {$\so$} (t3.north);
			\end{tikzpicture}
		}
		\vspace{-5.5mm}
		\caption{$t_4$ reads $x$ and $y$ from $t_1$.}
		\label{fig:extension:b}
	\end{subfigure}
	\hspace{.15cm}
	\centering
	\begin{subfigure}[t]{.28\textwidth}
		\resizebox{1.1\textwidth}{!}{
			\begin{tikzpicture}[->,>=stealth',shorten >=1pt,auto,node distance=3cm,
				semithick, transform shape]
				\node[draw, rounded corners=2mm,outer sep=0] (t0) at (-1.5, 0) {\begin{tabular}{l} $\init$ \end{tabular}};
				\node[draw, rounded corners=2mm,outer sep=0, label={[font=\small]100:$t_1$}] (t1) at (-3, -1.6) {\begin{tabular}{l} 
						$\wrt{x}{1}$ \\ $\wrt{y}{1}$
				\end{tabular}};
				\node[draw, rounded corners=2mm,outer sep=0, label={[font=\small]50:$t_2$}] (t2) at (-3, -3.2) {\begin{tabular}{l} 
						$\wrt{x}{2}$
				\end{tabular}};
				\node[draw, rounded corners=2mm,outer sep=0, label={[font=\small]50:$t_3$}] (t3) at (0, -1.6) {\begin{tabular}{l} 
						$\wrt{x}{3}$
				\end{tabular}};	
				\node[draw, rounded corners=2mm,outer sep=0, label={[font=\small]10:$t_4$}] (t4) at (0, -3.2)
				{\begin{tabular}{l} 
						$\rd{y}$ \\ \textbf{\textcolor{blue}{$\rd{x}$}}
				\end{tabular}};			
				
				\path (t1.south west) -- (t1.south) coordinate[pos=0.67] (t1x);
				\path (t2.north west) -- (t2.north) coordinate[pos=0.67] (t2x);
				\path (t3.north west) -- (t3.west) coordinate[pos=0.67] (t3x);
				
				\path (t1.south east) edge node[above] {$\wro_y$} (t4.north west);
				\path (t3.south) edge node[right] {$\so \cap \wro_x$} (t4.north);
				\path (t1.south) edge node[left] {$\so$} (t2.north);
				\path (t0.south) edge node[above] {$\so$} (t1.north);
				\path (t0.south) edge node[above] {$\so$} (t3.north);
			\end{tikzpicture}
			
		}
		\vspace{-5.5mm}
		\caption{$t_4$ reads $x$ from $t_3$, $y$ from $t_1$.}
		\label{fig:extension:c}
	\end{subfigure}
	\vspace{-4mm}
	\caption{Two causal extensions of the history $h$ on the left with the $\erd{x}$ event written in blue.}
	\vspace{-3.5mm}
\end{figure}
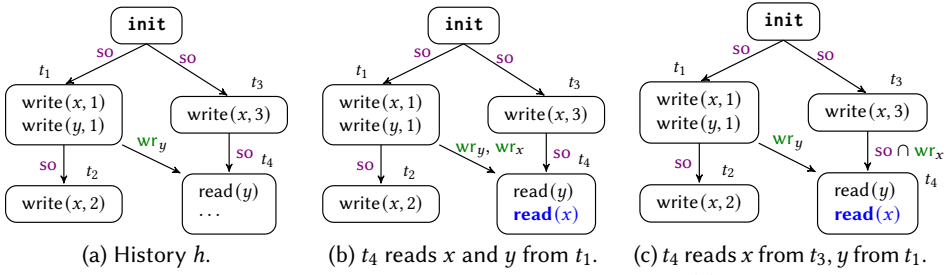

For example, Figure~\ref{fig:extension:b} and~\ref{fig:extension:c} present two causal extensions with a $\erd{x}$ event of the transaction $t_4$ in the history $\hist$ in Figure~\ref{fig:extension:a}. The new read event reads from transaction $t_1$ or $t_3$ which were already related by $(\so \cup \wro)^+$ to $t_4$.
An extension of $h$ where the new read event reads from $t_2$ is \emph{not} a causal extension because $(t_2, t_4) \not\in (\so \cup \wro)^+$.

\vspace{-1.5mm}
\begin{definition}
\label{def:causally-extensible}
An isolation level $I$ is called \emph{causally-extensible} if for every $I$-consistent history $h$, every $(\so \cup \wro)^+$-maximal pending transaction $t$ in $h$, and every event $e$, there exists a causal extension $\hist'$ of $t$ with $e$ that is $I$-consistent.
\vspace{-1.5mm}
\end{definition}

\appendixver{The proof of the following theorem can be found in Appendix~\ref{app:causalExtensible}.}

\vspace{-1.5mm}
\begin{restatable}{theorem}{causalExtensibleModels}
\label{theorem:causalExtensibleModels-CC-RA-RC}
Causal Consistency, Read Atomic, and Read Committed are causally-extensible.
\vspace{-1.5mm}
\end{restatable}


Snapshot Isolation and Serializability are \emph{not} causally extensible. 
Figure~\ref{fig:non-causally-extensible} presents a counter-example to causal extensibility: the causal extension of the history $\hist$ that does \emph{not} contain the $\ewrt{x,2}$ written in blue bold font with this event does not satisfy neither Snapshot Isolation nor Serializability although $\hist$ does. Note that the causal extension with a write event is unique. \nver{(Note that both $h$ and this causal extension satisfy Causal Consistency and therefore, as expected, this counter-example does not apply to isolation levels weaker than Causal Consistency.)}


\section{Swapping-Based Model Checking Algorithms}\label{sec:algs}



We define a class of stateless model checking algorithms for enumerating executions of a given transactional program, that we call \emph{swapping-based algorithms}. Section~\ref{sec:CC-algorithm} will describe a concrete instance that applies to isolation levels that are prefix-closed and causally extensible.


These algorithms are defined by the recursive function \textsc{explore} listed in Algorithm~\ref{algorithm:algo-class}. The function \textsc{explore} receives as input a program $\prog$, an \emph{ordered history} $\hist_<$, which is a pair $(\hist,<)$ {\parfillskip0pt\par}

\begin{wrapfigure}{l}{0.45\textwidth}
\vspace{-6mm}
\begin{minipage}{0.45\textwidth}
\begin{algorithm}[H]
\caption{\textsc{explore} algorithm}
\begin{algorithmic}[1]
\Function{\textsc{explore}}{$\prog$, $\hist_<$, $\locals$}
\State $j, e, \gamma \gets  \genericNext(\prog, \hist_<, \locals)$

\State $\locals' \gets \locals[e\mapsto \gamma]$
\If{$e = \bot$ and $\genericEvaluate(h)$}
\State \textbf{output} $h$, $\locals'$
\ElsIf{$\mathit{type}(e) = \iread$}

\ForAll{$t \in \genericValidWrites(h,e)$} 
\State $h'_< \gets h_< \oplus_j e \oplus \wro(t,e)$

\State $\textsc{explore}(\prog, h'_< , \locals')$
\State $\textsc{exploreSwaps}(\prog, h'_< , \locals')$
\EndFor
\Else 
\State $h'_< \gets h_< \oplus_j e$
\State $\textsc{explore}(\prog, h'_<,\locals')$
\State $\textsc{exploreSwaps}(\prog, h'_< , \locals')$
\label{algorithm:algo-class:exploreStd}
\EndIf
\EndFunction

\end{algorithmic}
\label{algorithm:algo-class}
\end{algorithm}
\end{minipage}
\end{wrapfigure}

\noindent of a history and a total order $<$ on all the events in $\hist$, and a mapping $\locals$ that associates each event $e$ in $\hist$ with the valuation of local variables in the transaction of $e$ ($\trans{h}{e}$) just before executing $e$. 
For an ordered history $(\hist,<)$ with $\hist=\tup{T, \so, \wro}$, we assume that $<$ is consistent with $\po$, $\so$, and $\wro$, i.e., $e_1 < e_2$ if $(\trans{h}{e_1},\trans{h}{e_2})\in (\so\cup\wro)^+$ or $(e_1,e_2)\in\po$. Initially, the ordered history and the mapping $\locals$ are empty. 


The function \textsc{explore} starts by calling $\genericNext$ to obtain an event representing the next database access in some pending transaction of $\prog$, or a  $\ebegin$/$\ecommit$/$\eabort$ event for starting or ending a transaction. This event is associated to some session $j$. For example, a typical implementation of $\genericNext$ would choose one of the pending transactions (in some session $j$), execute all local instructions until the next database instruction in that transaction (applying the transition rules \textsc{if-true}, \textsc{if-false}, and \textsc{local}\appendixver{from Figure~\ref{fig:op:sem:baseline} in Appendix~\ref{app:rules}}) and return the event $e$ corresponding to that database instruction and the current local state $\gamma$. $\genericNext$ may also return $\bot$ if the program finished.
If $\genericNext$ returns $\bot$, then the function $\genericEvaluate$ can be used to filter executions that satisfy the intended isolation level before outputting the current history and local states \nver{(the use of $\genericEvaluate$ will become relevant in Section~\ref{sec:ser})}.


Otherwise, the event $e$ is added to the ordered history $h_<$. If $e$ is a read event, then $\genericValidWrites$ computes a set of write events $w$ in the current history that are valid for $e$, i.e., adding the event $e$ along with the $\wro$ dependency $(w,e)$ leads to a history that still satisfies the intended isolation level.
%
Concerning notations, let $\hist$ be a history where $\so$ is represented as a function $\so:\mathsf{SessId}\rightharpoonup \mathsf{Tlogs}^*$ (as in \S~\ref{ssec:semantics}). 
For event $e$, $\hist\oplus_j e$ is the history obtained from $\hist$ by adding $e$ to the last transaction in $\so(j)$ as the last event in $\po$ (i.e.,  if $\so(j)=\sigma; \tup{t,E,\po_t}$, then the session order $\so'$ of $\hist\oplus_j e$ is defined by $\so'(k)=\so(k)$ for all $k\neq j$ and $\so(j) =\sigma; \tup{t,E\cup\{e\},\po_t\cup \{(e',e): e'\in E\}}$). This is extended to ordered histories: $(\hist,<) \oplus_j e$ is defined as $(\hist \oplus_j e, < \cdot\ e)$ where $<\cdot\ e$ means that $e$ is added as the last element of $<$. Also, $\hist \oplus_j (e,\ebegin)$ is a history where $\tup{\tr,\{\tup{e,\ebegin}\},\emptyset}$ with $\tr$ a fresh id is appended to $\so(j)$, and 
$\hist\oplus\wro(t,e)$ is defined by adding $(t,e)$ to the write-read of $\hist$.

\begin{wrapfigure}{lt}{0.47\textwidth}
\vspace{-8mm}
\begin{minipage}{0.47\textwidth}
\begin{algorithm}[H]
\caption{\textsc{exploreSwaps} }
\begin{algorithmic}[1]
\Function{\textsc{exploreSwaps}}{$\prog$, $\hist_<$, $\locals$}
\State $l \gets \genericCompute(h_<)$		
\ForAll{$(\alpha, \beta) \in l$}
\If{$\genericProtocol(h_<, \alpha, \beta,\locals)$}
\State $\textsc{explore}(\prog,\genericSwap(h_<, \alpha, \beta, \locals))$	

\EndIf
\EndFor
\EndFunction

\end{algorithmic}
\label{algorithm:explore-swaps}
\end{algorithm}
\end{minipage}
\vspace{-4mm}
\end{wrapfigure}

Once an event is added to the current history, the algorithm may explore other histories obtained by re-ordering events in the current one. Such re-orderings are required for completeness. New read events can only read from writes executed in the past which limits the set of explored histories to the scheduling imposed by $\genericNext$. Without re-orderings, writes scheduled later by $\genericNext$ cannot be read by read events executed in the past, although this may be permitted by the isolation level. 

The function $\textsc{exploreSwaps}$ calls $\genericCompute$ to compute pairs of sequences of events $\alpha, \beta$ that should be re-ordered; $\alpha$ and $\beta$ are \emph{contiguous and disjoint} subsequences of the total order $<$, and $\alpha$ should end before $\beta$ (since $\beta$ will be re-ordered before $\alpha$). Typically, $\alpha$ would contain a read event $r$ and $\beta$ a write event $w$ such that re-ordering the two enables $r$ to read from $w$. Ensuring soundness and avoiding redundancy, i.e., exploring the same history multiple times, may require restricting the application of such re-orderings. This is modeled by the Boolean condition called $\genericProtocol$. If this condition holds, the new explored histories are computed by the function $\genericSwap$. This function returns local states as well, which are necessary for continuing the exploration. 
We assume that $\genericSwap(\hist_<, \alpha, \beta, \locals)$ returns pairs $(\hist'_{<'},\locals')$ such that 
\begin{enumerate}
	\item $\hist'$ contains \nver{at least} the events in $\alpha$ and $\beta$,
	\item $\hist'$ without the events in $\alpha$ is a prefix of $\hist$, and
	\item if a read $r$ in $\alpha$ reads from different writes in $\hist$ and $\hist'$ (the $\wro$ relations of $\hist$ and $\hist'$ associate different transactions to $r$),
	 then $r$ is the last event in its transaction (w.r.t. $\po$).
\end{enumerate}

The first condition makes the re-ordering ``meaningful'' while the last two conditions ensure that the history $\hist'$ is feasible by construction, i.e., it can be obtained using the operational semantics defined in Section~\ref{ssec:semantics}. Feasibility of $\hist'$ is ensured by keeping prefixes of transaction logs from $\hist$ and all their $\wro$ dependencies except possibly for read events in $\alpha$ (second condition). In particular, for events in $\beta$, it implies that $\hist'$ contains all their $(\po\cup\so\cup\wro)^*$ predecessors. Also, the change of a read-from dependency is restricted to the last read in a transaction (third condition) because changing the value returned by a read may disable later events in the same transaction\footnote{\nver{Different $\wro$ dependencies for previous reads can be explored  in other steps of the algorithm.}}.

A concrete implementation of \textsc{explore} is called:
\vspace{-1mm}
\begin{itemize}
	\item \emph{$I$-sound} if it outputs only histories in $\histOf[I]{\prog}$ for every program $\prog$,
	\item \emph{$I$-complete} if it outputs every history in $\histOf[I]{\prog}$ for every program $\prog$,
	\item \emph{optimal} if it does not output the same history twice,
	\item \emph{strongly optimal} if it is optimal and never engages in fruitless explorations, i.e., \textsc{explore} is never called (recursively) on a history $\hist$ that does not satisfy $I$, and every call to \textsc{explore} results in an output or another recursive call to \textsc{explore}.
\vspace{-1mm}
\end{itemize}

\section{Swapping-based model checking for Prefix-Closed and Causally-Extensible Isolation Levels}
\label{sec:CC-algorithm}


We define a concrete implementation of $\textsc{explore}$, denoted as $\textsc{explore-ce}$, that is $I$-sound, $I$-complete, and strongly optimal for any isolation level $I$ that is prefix-closed and causally-extensible. The isolation level $I$ is a parameter of $\textsc{explore-ce}$. The space complexity of $\textsc{explore-ce}$ is polynomial in the size of the program. An important invariant of this implementation is that it explores histories with \emph{at most one} pending transaction and this transaction is maximal in session order. This invariant is used to avoid fruitless explorations: since $I$ is assumed to be causally-extensible, there always exists an extension of the current history with one more event that continues to satisfy $I$. \nver{Moreover, this invariant is sufficient to guarantee completeness in the sense defined above of exploring all histories of ``full'' program executions (that end in a final configuration).}

Section~\ref{ssec:extensions} describes the implementations of $\genericNext$ and $\genericValidWrites$ used to extend a given execution, Section~\ref{subsection:SwappingHistories} describes the functions $\genericCompute$ and $\genericSwap$ used to compute re-ordered executions, and Section~\ref{ssec:optimality} describes the $\genericProtocol$ restriction on re-ordering. We assume that the function $\genericEvaluate$ is defined as simply $\genericEvaluate(\hist) ::= true$ (no filter before outputting). Section~\ref{ssec:corr} discusses correctness arguments\appendixver{ (see App.~\ref{sec:proofs-algorithm} for proofs)}.

\subsection{Extending Histories According to An Oracle Order}\label{ssec:extensions}


The function $\genericNext$ generates events representing database accesses to extend an execution, according to an \emph{arbitrary but fixed} order between the transactions in the program called \callout{oracle order}. 
We assume that the oracle order, denoted by $<_{\ora}$, is consistent with the order between transactions in the same session of the program. The extension of $<_{\ora}$ to events is defined as expected. \nver{For example, assuming that each session has an id, an oracle order can be defined by an order on session ids along with the session order $\so$: transactions from sessions with smaller ids are considered first and the order between transactions in the same session follows $\so$.}


$\genericNext$ returns a new event of the transaction that is not already completed and that is \emph{minimal} according to $<_{\ora}$. In more detail, if $j,e,\gamma$ is the output of $\genericNext(\prog, \hist_<, \locals)$, then either:
\vspace{-1mm}
\begin{itemize}
	\item the last transaction log $t$ of session $j$ (w.r.t. $\so$) in $\hist$ is pending, and $t$ is the smallest among pending transaction logs in $\hist$ w.r.t. $<_{\ora}$ 
	\item $\hist$ contains no pending transaction logs and the next transaction of sessions $j$ is the smallest among not yet started transactions in the program w.r.t. $<_{\ora}$.
	\vspace{-1mm}
\end{itemize}

This implementation of $\genericNext$ is deterministic and it prioritizes the completion of pending transactions. The latter is useful to maintain the invariant that any history explored by the algorithm has at most one pending transaction. Preserving this invariant requires that the histories given as input to $\genericNext$ also have at most one pending transaction. This is discussed further when explaining the process of re-ordering events in Section~\ref{subsection:SwappingHistories}.

\begin{figure}[t]
	\centering
	\begin{subfigure}[b]{.32\textwidth}
\begin{adjustbox}{max width=\textwidth}
{
\begin{tabular}{c||c}
	\begin{lstlisting}[xleftmargin=5mm,basicstyle=\ttfamily\scriptsize,escapeinside={(*}{*)}, tabsize=1]
begin;
a = read((*x*));
if(a == 3)
	write((*$y$*),1);
commit
begin;
b = read((*x*));
c = read((*y*));
commit
\end{lstlisting} &
\begin{lstlisting}[xleftmargin=5mm,basicstyle=\ttfamily\scriptsize,escapeinside={(*}{*)}, tabsize=1]
begin;
d = read((*x*));
write((*$x$*),3);
commit
\end{lstlisting} 
\end{tabular}}
\end{adjustbox}

		\caption{Program \nver{(2 sessions)}.}
		\label{fig:oracle-order:prog}
	\end{subfigure}
\centering
\begin{subfigure}[b]{.31\textwidth}
	\resizebox{\textwidth}{!}{
		\begin{tikzpicture}[->,>=stealth',shorten >=1pt,auto,node distance=3cm,
			semithick, transform shape]
			\node[draw, rounded corners=2mm,outer sep=0] (t1) at (-3.5, 0) {\begin{tabular}{l} 
					\init
			\end{tabular}};
			
			\node[draw, rounded corners=2mm,outer sep=0, label={[font=\small]50:$t_1$}] (t2) at (-3.5, -2) {
				\begin{tabular}{l}
					$\rd{x}$\\
					{\pgfsetfillopacity{0.3}$\wrt{y}{1}$}						
			\end{tabular}};
			
			\node[draw, rounded corners=2mm,outer sep=0, label={[font=\small]50:$t_3$}] (t3) at (0, 0) {
				\begin{tabular}{l}
					$\rd{x}$ \\
					$\wrt{x}{3}$
			\end{tabular}};
			
			\node[draw, rounded corners=2mm,outer sep=0, opacity=0.3, label={[font=\small]50:$t_2$}] (t4) at (0, -2) {
				\begin{tabular}{l} 
					 $\rd{x}$ \\
					 $\rd{y}$
			\end{tabular}};
			
			\path (t1) edge[left] node[yshift=0,xshift=0] {$\ora$} (t2);
			\path (t4) edge[left] node[yshift=0,xshift=0] {$\ora$} (t3);
			\path (t1) edge[above] node[yshift=0,xshift=0] {$\wro$} (t3);
						
			\path (t2.north east) -- (t2.east) coordinate[pos=0.67] (t2or);
			\path (t2.south east) -- (t2.east) coordinate[pos=0.67] (t2so);
			\path (t4.south west) -- (t4.west) coordinate[pos=0.67] (t4so);
			\path (t4.north west) -- (t4.west) coordinate[pos=0.67] (t4or);

			\path (t2or) edge[above] node[yshift=0,xshift=0] {$\ora$} (t4or);
			\path [opacity=0.3](t2so) edge[below] node[yshift=0,xshift=0] {$\so$} (t4so);

			\path (t3.south west) edge[left] node[yshift=0,xshift=0] {$\wro$} (t2.north east);
			
		\end{tikzpicture}  
		
	}
	\caption{An incomplete history.}
	\label{fig:oracle_order:a}
\end{subfigure}
	\hspace{.5cm}
	\centering
	\begin{subfigure}[b]{.31\textwidth}
		\resizebox{\textwidth}{!}{
			\begin{tikzpicture}[->,>=stealth',shorten >=1pt,auto,node distance=3cm,
				semithick, transform shape]
				\node[draw, rounded corners=2mm,outer sep=0] (t1) at (-3.5, 0) {\begin{tabular}{l} 
						\init
				\end{tabular}};
				
				\node[draw, rounded corners=2mm,outer sep=0, label={[font=\small]50:$t_1$}] (t2) at (-3.5, -2) {
					\begin{tabular}{l}
						$\rd{x}$\\
						$\wrt{y}{1}$					
				\end{tabular}};
				
				\node[draw, rounded corners=2mm,outer sep=0, label={[font=\small]50:$t_3$}] (t3) at (0, 0) {
					\begin{tabular}{l}
						$\rd{x}$ \\
						$\wrt{x}{3}$
				\end{tabular}};
				
				\node[draw, rounded corners=2mm,outer sep=0, opacity=0.3, label={[font=\small]50:$t_2$}] (t4) at (0, -2) {
					\begin{tabular}{l} 
						$\rd{x}$ \\
						$\rd{y}$
				\end{tabular}};
				
				\path (t1) edge[left] node[yshift=0,xshift=0] {$\ora$} (t2);
				\path (t4) edge[left] node[yshift=0,xshift=0] {$\ora$} (t3);
				\path (t1) edge[above] node[yshift=0,xshift=0] {$\wro$} (t3);
				
				\path (t2.north east) -- (t2.east) coordinate[pos=0.67] (t2or);
				\path (t2.south east) -- (t2.east) coordinate[pos=0.67] (t2so);
				\path (t4.south west) -- (t4.west) coordinate[pos=0.67] (t4so);
				\path (t4.north west) -- (t4.west) coordinate[pos=0.67] (t4or);

				\path (t2or) edge[above] node[yshift=0,xshift=0] {$\ora$} (t4or);
				\path [opacity=0.3](t2so) edge[below] node[yshift=0,xshift=0] {$\so$} (t4so);

				\path (t3.south west) edge[left] node[yshift=0,xshift=0] {$\wro$} (t2.north east);
				
			\end{tikzpicture}  
			
		}
		\caption{An extension.} 
		\label{fig:oracle_order:b}
	\end{subfigure}
	\vspace{-2mm}
	\caption{A program with two sessions (a), a history $h$ (b), and an extension of $h$ with an event returned by $\genericNext$ (c). 
	The $\so$-edges from $\init$ to the other transactions are omitted for legibility. 
	We use edges labeled by $\ora$ to represent the oracle order $<_{\ora}$.
	Events in gray are not yet added to the history.}
	\label{fig:oracle_order}
	\vspace{-4mm}	
\end{figure}
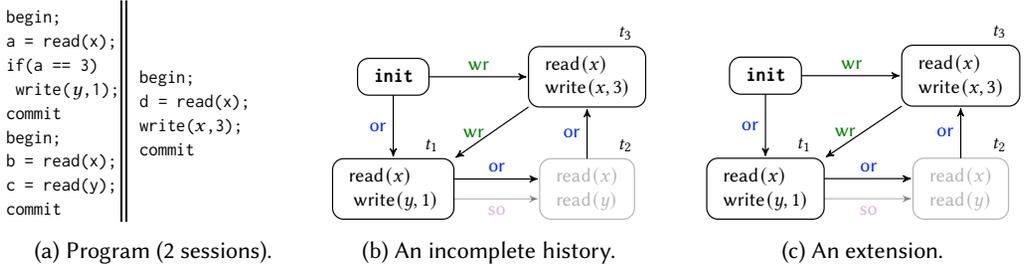

For example, consider the program in Figure~\ref{fig:oracle-order:prog}, an oracle order which orders \oldver{all}\nver{the two} transactions \oldver{on}\nver{in} the left \nver{session} before the \nver{transaction in the} right \oldver{one}\nver{session}, and the history $h$ in Figure~\ref{fig:oracle_order:a}. Since the local state of the pending transaction on the left stores $3$ to the local variable $a$ (as a result of the previous $\erd{x}$ event) and the Boolean condition in \texttt{if} holds, $\genericNext$ will return the event $\wrt{y}{1}$ when called with $\hist$.

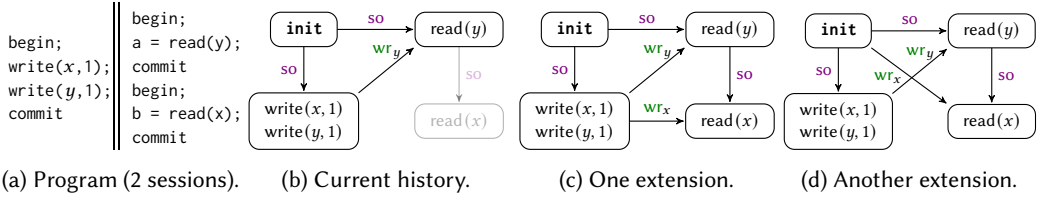
\begin{figure}[t]

\begin{subfigure}[b]{.24\textwidth}
\begin{adjustbox}{max width=\textwidth}
\begin{tabular}{c||c}
\begin{lstlisting}[xleftmargin=5mm,basicstyle=\ttfamily\scriptsize,escapeinside={(*}{*)}, tabsize=1]
begin;
write((*$x$*),1);
write((*$y$*),1);
commit
\end{lstlisting} & 
\begin{lstlisting}[xleftmargin=5mm,basicstyle=\ttfamily\scriptsize,escapeinside={(*}{*)}, tabsize=1]
begin;
a = read((*y*));
commit
begin;
b = read((*x*));
commit
\end{lstlisting} 
\end{tabular} 
\end{adjustbox}
\caption{Program \nver{(2 sessions)}.}
\label{fig:add_read:prog}
\end{subfigure}
\hspace{-2mm}
	\begin{subfigure}[b]{.253\textwidth}
		\resizebox{\textwidth}{!}{
			\begin{tikzpicture}[->,>=stealth',shorten >=1pt,auto,node distance=3cm,
				semithick, transform shape]
				\node[draw, rounded corners=2mm,outer sep=0] (t1) at (-3, 0) {\begin{tabular}{l} $\init$\end{tabular}};
				\node[draw, rounded corners=2mm,outer sep=0] (t2) at (-3, -1.75) {\begin{tabular}{l} $\wrt{x}{1}$ \\ $\wrt{y}{1}$ \end{tabular}};
				\node[draw, rounded corners=2mm,outer sep=0] (t3) at (0, 0) {\begin{tabular}{l} $\rd{y}$\end{tabular}};
				\node[draw, rounded corners=2mm, outer sep=0, opacity=0.3] (t4) at (0, -1.75) {
					\begin{tabular}{l} 
					$\rd{x}$
				\end{tabular}};
				
				\path (t1.south east) -- (t1.east) coordinate[pos=0.67] (t1y);
				\path (t1.north east) -- (t1.east) coordinate[pos=0.67] (t1x);
				\path (t3.south west) -- (t3.west) coordinate[pos=0.67] (t3y);
				\path (t3.north west) -- (t3.west) coordinate[pos=0.67] (t3x);
				\path (t1.east) edge[above] node[yshift=0,xshift=0] {$\so$} (t3.west);
				\path (t1.south) edge[above] node[left] {$\so$} (t2.north);
				\path (t2.north east) edge[above] node [yshift=2,xshift=0] {$\wro_y$} (t3.south west);
				\path[opacity=0.3] (t3.south) edge[below] node [right]{$\so$} (t4.north); 
			\end{tikzpicture}  
			
		}
		\caption{Current history.}
		\label{fig:add_read:a}
	\end{subfigure}
	\centering
	\begin{subfigure}[b]{.253\textwidth}
		\resizebox{\textwidth}{!}{
			\begin{tikzpicture}[->,>=stealth',shorten >=1pt,auto,node distance=3cm,
				semithick, transform shape]
				\node[draw, rounded corners=2mm,outer sep=0] (t1) at (-3, 0) {\begin{tabular}{l} $\init$\end{tabular}};
				\node[draw, rounded corners=2mm,outer sep=0] (t2) at (-3, -1.75) {\begin{tabular}{l} $\wrt{x}{1}$ \\ $\wrt{y}{1}$ \end{tabular}};
				\node[draw, rounded corners=2mm,outer sep=0] (t3) at (0, 0) {\begin{tabular}{l} $\rd{y}$\end{tabular}};
				\node[draw, rounded corners=2mm, outer sep=0] (t4) at (0, -1.75) {
					\begin{tabular}{l} 
						$\rd{x}$
				\end{tabular}};
				
				\path (t1.south east) -- (t1.east) coordinate[pos=0.67] (t1y);
				\path (t1.north east) -- (t1.east) coordinate[pos=0.67] (t1x);
				\path (t3.south west) -- (t3.west) coordinate[pos=0.67] (t3y);
				\path (t3.north west) -- (t3.west) coordinate[pos=0.67] (t3x);
				\path (t1.east) edge[above] node[yshift=0,xshift=0] {$\so$} (t3.west);
				\path (t1.south) edge[above] node[left] {$\so$} (t2.north);
				\path (t2.north east) edge[above] node [yshift=2,xshift=0] {$\wro_y$} (t3.south west);
				\path (t3.south) edge[below] node [right]{$\so$} (t4.north);
				\path (t2.east) edge[below] node [above]{$\wro_x$} (t4.west); 
			\end{tikzpicture}  
			
		}
		\caption{One extension.}
		\label{fig:add_read:b}
	\end{subfigure}
	\centering
	\begin{subfigure}[b]{.25\textwidth}
		\resizebox{\textwidth}{!}{
			\begin{tikzpicture}[->,>=stealth',shorten >=1pt,auto,node distance=3cm,
				semithick, transform shape]
				\node[draw, rounded corners=2mm,outer sep=0] (t1) at (-3, 0) {\begin{tabular}{l} $\init$\end{tabular}};
				\node[draw, rounded corners=2mm,outer sep=0] (t2) at (-3, -1.75) {\begin{tabular}{l} $\wrt{x}{1}$ \\ $\wrt{y}{1}$ \end{tabular}};
				\node[draw, rounded corners=2mm,outer sep=0] (t3) at (0, 0) {\begin{tabular}{l} $\rd{y}$\end{tabular}};
				\node[draw, rounded corners=2mm, outer sep=0] (t4) at (0, -1.75) {
					\begin{tabular}{l} 
						$\rd{x}$
				\end{tabular}};
				
				\path (t1.south east) -- (t1.east) coordinate[pos=0.67] (t1y);
				\path (t1.north east) -- (t1.east) coordinate[pos=0.67] (t1x);
				\path (t3.south west) -- (t3.west) coordinate[pos=0.67] (t3y);
				\path (t3.north west) -- (t3.west) coordinate[pos=0.67] (t3x);
				\path (t1.east) edge[above] node[yshift=0,xshift=0] {$\so$} (t3.west);
				\path (t1.south) edge[above] node[left] {$\so$} (t2.north);
				\path (t2.north east) edge[above] node [yshift=2,xshift=0] {$\wro_y$} (t3.south west);
				\path (t3.south) edge[below] node [right]{$\so$} (t4.north);
				\path (t1.south east) edge[below] node [yshift=6,xshift=-12]{$\wro_x$} (t4.north west); 
			\end{tikzpicture}  
			
		}
		\caption{Another extension.}
		\label{fig:add_read:c}
	\end{subfigure}\hspace{.5cm}
	\vspace{-6mm}
	\caption{Extensions of a history by adding a $\iread$ event. Events in gray are not yet added to the history.}
	\label{fig:add_read}
	\vspace{-3mm}
\end{figure}


According to Algorithm~\ref{algorithm:algo-class}, if the event returned by $\genericNext$ is not a read event, then it is simply added to the current history as the maximal element of the order $<$ (cf. the definition of $\oplus_j$ on ordered histories). If it is a read event, then adding this event may result in multiple histories depending on the chosen $\wro$ dependency. 
%
For example, in Figure \ref{fig:add_read}, extending the history in Figure \ref{fig:add_read:a} with the $\erd{x}$ event could result in two different histories, pictured in Figure \ref{fig:add_read:b} and \ref{fig:add_read:c}, depending on the write with whom this read event is associated by $\wro$. However, under $\CC$, the latter history is inconsistent.
The function $\genericValidWrites$ limits the choices to those that preserve consistency with the intended isolation level $I$, i.e.,
\begin{align*}
\genericValidWrites(h, e) \coloneqq  \{t \ \in \transC{h} \ | \ h \oplus_j e \oplus \wro(t,e) \mbox{ satisfies }I\}
\end{align*}
where $\transC{h}$ is the set of committed transactions in $h$.

\subsection{Re-Ordering Events in Histories}
\label{subsection:SwappingHistories} 

After extending the current history with one more event, $\textsc{explore}$ may be called recursively on other histories obtained by re-ordering events in the current one (and dropping some other events).

\begin{figure}[t]
\vspace{-1mm}

	\centering
	\begin{subfigure}[b]{.25\textwidth}
\begin{adjustbox}{max width=\textwidth}
\begin{tabular}{c||c}
\begin{lstlisting}[xleftmargin=5mm,basicstyle=\ttfamily\scriptsize,escapeinside={(*}{*)}, tabsize=1]
begin;
a = read((*x*));
b = read((*y*));
commit
\end{lstlisting} &
\begin{lstlisting}[xleftmargin=5mm,basicstyle=\ttfamily\scriptsize,escapeinside={(*}{*)}, tabsize=1]
begin;
write((*$x$*),2);
write((*$y$*),2);
commit
\end{lstlisting} 
\end{tabular} 
\end{adjustbox}
\vspace{1.2cm}
		\caption{Program \nver{(2 sessions)}.}
		\label{fig:dead_branch:prog}
	\end{subfigure}
\centering
	\begin{subfigure}[b]{.125\textwidth}
		\resizebox{\textwidth}{!}{
			\begin{tikzpicture}[->,>=stealth',shorten >=1pt,auto,node distance=3cm,
				semithick, transform shape]
				
				\node[draw, rounded corners=2mm,outer sep=0] (t1) at (-3.25, 0) {\begin{tabular}{l} $\init$ \end{tabular}};
				\node[draw, rounded corners=2mm,outer sep=0, inner sep=0.75mm, label={[font=\small]50:$t_1$}] (t2) at (-3.25, -1.6) {\begin{tabular}{l} 
						$\rd{x}$ \\ $\rd{y}$
				\end{tabular}};
				\node[draw, rounded corners=2mm,outer sep=0, inner sep=0.75mm, label={[font=\small]50:$t_2$}] (t3) at (-3.25, -3.2) {\begin{tabular}{l} 
						$\wrt{x}{2}$\pgfsetfillopacity{0.3}\\
						$\wrt{y}{2}$
				\end{tabular}};			
				
				\path (t1.south west) -- (t1.south) coordinate[pos=0.67] (t1x);
				\path (t1.south east) -- (t1.south) coordinate[pos=0.67] (t1y);
				\path (t2.north west) -- (t2.north) coordinate[pos=0.67] (t2x);
				\path (t2.north east) -- (t2.north) coordinate[pos=0.67] (t2y);
				
				\path (t1.south) edge [left] node {$\wro$} (t2.north);

			\end{tikzpicture}  
			
		}
		\caption{Current.}
		\label{fig:dead_branch:a}
	\end{subfigure}
	\hspace{.175cm}
	\centering
	\begin{subfigure}[b]{.125\textwidth}
		\resizebox{\textwidth}{!}{
			\begin{tikzpicture}[->,>=stealth',shorten >=1pt,auto,node distance=3cm,
				semithick, transform shape]
				\node[draw, rounded corners=2mm,outer sep=0] (t1) at (-3.25, 0) {\begin{tabular}{l} $\init$ \end{tabular}};
				\node[draw, rounded corners=2mm,outer sep=0, inner sep=0.75mm, label={[font=\small]50:$t_2$}] (t2) at (-3.25, -3.2) {\begin{tabular}{l} 
						$\rd{x}$ \pgfsetfillopacity{0.3}\\ $\rd{y}$
				\end{tabular}};
				\node[draw, rounded corners=2mm,outer sep=0, inner sep=0.75mm, label={[font=\small]50:$t_1$}] (t3) at (-3.25, -1.6) {\begin{tabular}{l} 
						$\wrt{x}{2}$\pgfsetfillopacity{0.3}\\
						$\wrt{y}{2}$
				\end{tabular}};			
				
				\path (t1.south west) -- (t1.south) coordinate[pos=0.67] (t1x);
				\path (t2.north west) -- (t2.north) coordinate[pos=0.67] (t2x);
				\path (t3.north west) -- (t3.west) coordinate[pos=0.67] (t3x);
				
				\path (t3.south) edge [left] node {$\wro_x$} (t2.north); 
			\end{tikzpicture}  
			
		}
		\caption{Reorder.}
		\label{fig:dead_branch:b}
	\end{subfigure}
	\hspace{.175cm}
	\centering
		\begin{subfigure}[b]{.177\textwidth}
		\resizebox{\textwidth}{!}{
			\begin{tikzpicture}[->,>=stealth',shorten >=1pt,auto,node distance=3cm,
				semithick, transform shape]
				\node[draw, rounded corners=2mm,outer sep=0] (t1) at (-3.25, 0) {\begin{tabular}{l} $\init$ \end{tabular}};
				\node[draw, rounded corners=2mm,outer sep=0, inner sep=0.75mm, label={[font=\small]50:$t_2$}] (t2) at (-3.25, -3.2) {\begin{tabular}{l} 
						$\rd{x}$ \\ $\rd{y}$
				\end{tabular}};
				\node[draw, rounded corners=2mm,outer sep=0, inner sep=0.75mm, label={[font=\small]50:$t_1$}] (t3) at (-3.25, -1.6) {\begin{tabular}{l} 
						$\wrt{x}{2}$\pgfsetfillopacity{0.3}\\
						$\wrt{y}{2}$
				\end{tabular}};			
				
				\path (t1.south west) -- (t1.south) coordinate[pos=0.67] (t1x);
				\path (t2.north west) -- (t2.north) coordinate[pos=0.67] (t2x);
				\path (t3.north west) -- (t3.west) coordinate[pos=0.67] (t3x);
				
				\path (t3.south) edge [left] node {$\wro_x$} (t2.north);
				\path (t1.south west) edge [bend right=40] node[left] {$\wro_y$} (t2.north west);
			\end{tikzpicture}
		}
		\caption{Extended.}
		\label{fig:dead_branch:c}
	\end{subfigure}
	\hspace{.175cm}
	\centering
	\begin{subfigure}[b]{.177\textwidth}
		\resizebox{\textwidth}{!}{
			\begin{tikzpicture}[->,>=stealth',shorten >=1pt,auto,node distance=3cm,
				semithick, transform shape]
				\node[draw, rounded corners=2mm,outer sep=0] (t1) at (-3.25, 0) {\begin{tabular}{l} $\init$ \end{tabular}};
				\node[draw, rounded corners=2mm,outer sep=0, inner sep=0.75mm, label={[font=\small]50:$t_2$}] (t2) at (-3.25, -3.2) {\begin{tabular}{l} 
						$\rd{x}$ \\ $\rd{y}$
				\end{tabular}};
				\node[draw, rounded corners=2mm,outer sep=0, inner sep=0.75mm, label={[font=\small]50:$t_1$}] (t3) at (-3.25, -1.6) {\begin{tabular}{l} 
						$\wrt{x}{2}$\\
						$\wrt{y}{2}$
				\end{tabular}};			
				
				\path (t1.south west) -- (t1.south) coordinate[pos=0.67] (t1x);
				\path (t2.north west) -- (t2.north) coordinate[pos=0.67] (t2x);
				\path (t3.north west) -- (t3.west) coordinate[pos=0.67] (t3x);
				
				\path (t3.south) edge [left] node {$\wro_x$} (t2.north);
				\path (t1.south west) edge [bend right=40] node[left] {$\wro_y$} (t2.north west);
			\end{tikzpicture}
			
		}
		\caption{Inconsistent.}
		\label{fig:dead_branch:d}
	\end{subfigure}
	\vspace{-2.5mm}
	\caption{Example of inconsistency after swapping two events. All $\so$-edges from $\init$ to the other transactions are omitted for legibility. The history order $<$ is represented by the top to bottom order in each figure.
	Events in gray are not yet added to the history. }
	\label{fig:dead_branch}
	\vspace{-4mm}
\end{figure}
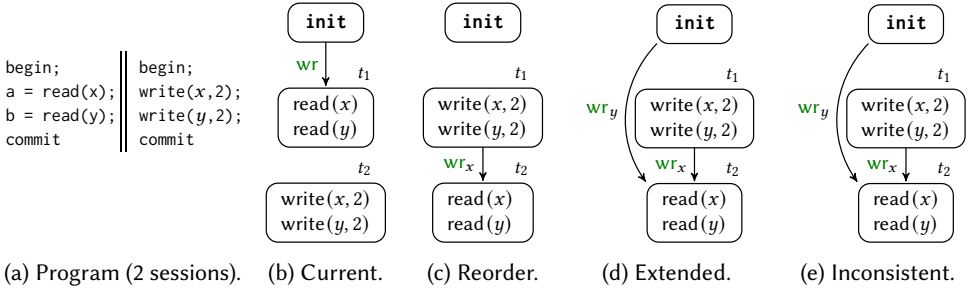

Re-ordering events must preserve the invariant of producing histories with at most one pending transaction.
To explain the use of this invariant in avoiding fruitless explorations, let us consider the program in Figure~\ref{fig:dead_branch:prog} assuming an exploration under Read Committed. The oracle order gives priority to the transaction on the left.
Assume that the current history reached by the exploration is the one pictured in Figure~\ref{fig:dead_branch:a} (the last added event is $\wrt{x}{2}$). Swapping $\wrt{x}{2}$ with $\rd{x}$ would result in the history pictured in Figure~\ref{fig:dead_branch:b}. To ensure that this swap produces a new history which was not explored in the past, the $\wro_x$ dependency of $\rd{x}$ is changed towards the $\wrt{x}{2}$ transaction (we detail this later).
By the definition of $\nextEvent$ (and the oracle order), this history shall be extended with $\rd{y}$, and this read event will be associated by $\wro_y$ to the only available $\iwrite(y,\_)$ event from $\init$. This is pictured in 
Figure~\ref{fig:dead_branch:c}. The next exploration step will extend the history with $\ewrt{y,2}$ (the only extension possible) which however, results in a history that does \emph{not} satisfy Read Committed, thereby, the recursive exploration branch being blocked.
The core issue is related to the history in Figure \ref{fig:dead_branch:c} which has a pending transaction that is \emph{not} $(\so \cup \wro)^+$-maximal. Being able to extend such a transaction while maintaining consistency is not guaranteed by Read Committed (and any other isolation level we consider). Nevertheless, causal extensibility guarantees the existence of an extension for pending transactions that are $(\so \cup \wro)^+$-maximal. We enforce this requirement by restricting the explored histories to have at most one pending transaction. This pending transaction will necessarily be $(\so \cup \wro)^+$-maximal.

%


To enforce histories with at most one pending transaction, the function $\genericCompute$, which identifies events to reorder, has a non-empty return value only when the last added event is $\ecommit$ (the end of a transaction)\footnote{\nver{Aborted transactions have no visible effect on the state of the database so swapping an aborted transaction cannot produce a new meaningful history.}}. Therefore, in such a case, it returns pairs \oldver{ that contain some read $r$ and the last completed transaction log $t$}\nver{of some transaction log prefix ending in a read $r$ and the last completed transaction log $t$}, such that the transaction log containing $r$ and $t$ are \emph{not} causally dependent (i.e., related by $(\so \cup \wro)^*$) \nver{(the transaction log prefix ending in $r$ and $t$ play the role of the subsequences $\alpha$ and respectively, $\beta$ in the description of $\genericCompute$ from Section~\ref{sec:algs}). To simplify the notation, we will assume that $\genericCompute$ returns pairs $(r,t)$.}
\vspace{-0.5mm}
\begin{align*}
& \genericCompute(h_<)  \coloneqq  \{(r,t) \in \Events \times T \ | \ r \in \readOp{T} \land \writeVar{t}{\variable{r}}\land \trans{h}{r} < t\\ 
& \hspace{2.7cm} \land\ (\trans{h}{r},t)\not\in (\so \cup \wro)^*\land \mbox{$t$ is complete and it includes the last event in $<$}\} \\
\end{align*}
\vspace{-9mm}
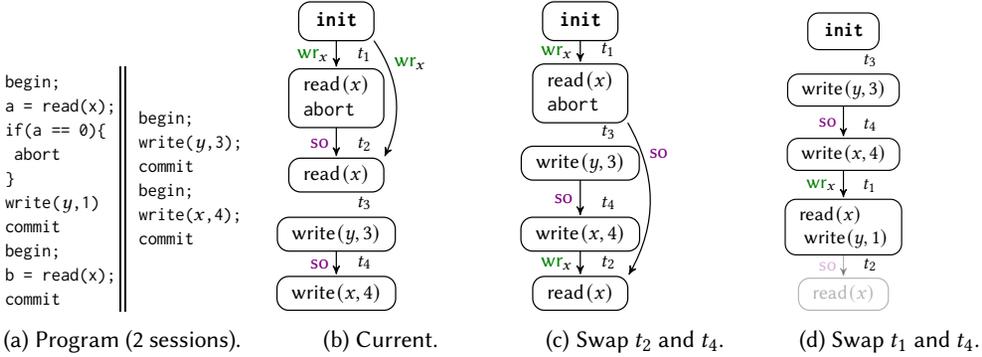
\begin{figure}[t]
\centering
\begin{subfigure}[b]{.25\textwidth}
\begin{adjustbox}{max width=\textwidth}
\begin{tabular}{c||c}
\begin{lstlisting}[xleftmargin=5mm,basicstyle=\ttfamily\scriptsize,escapeinside={(*}{*)}, tabsize=1]
begin;
a = read((*x*));
if(a == 0){
	abort
}
write((*$y$*),1)
commit
begin;
b = read((*x*));
commit		
\end{lstlisting} &
\begin{lstlisting}[xleftmargin=5mm,basicstyle=\ttfamily\scriptsize,escapeinside={(*}{*)}, tabsize=1]
begin;
write((*$y$*),3);
commit
begin;
write((*$x$*),4);
commit
	\end{lstlisting} 
\end{tabular} 
\end{adjustbox}
	
	\caption{Program \nver{(2 sessions)}.}
	\label{fig:compute-reordering:prog}
\end{subfigure} 
\hspace{.15cm}
	\centering
	\begin{subfigure}[b]{.21\textwidth}
		\resizebox{.76\textwidth}{!}{
			\begin{tikzpicture}[->,>=stealth',shorten >=1pt,auto,node distance=3cm,
				semithick, transform shape]
				
				\node[draw, rounded corners=2mm,outer sep=0] (t0) at (0, 0) {\begin{tabular}{l} $\init$ \end{tabular}};
				\node[draw, rounded corners=2mm,outer sep=0, inner sep=0.75mm, label={[font=\small]65:$t_1$}] (t1) at (0, -1.3) {\begin{tabular}{l} 
					$\rd{x}$ \\ $\iabort$
				\end{tabular}};
				\node[draw, rounded corners=2mm,outer sep=0, inner sep=0.75mm, label={[font=\small]50:$t_2$}] (t2) at (0, -2.6) {\begin{tabular}{l} 
					$\rd{x}$
				\end{tabular}};
				\node[draw, rounded corners=2mm,outer sep=0, inner sep=0.75mm, label={[font=\small]50:$t_3$}] (t3) at (0, -3.6) {\begin{tabular}{l} 
					$\wrt{y}{3}$
				\end{tabular}};
				\node[draw, rounded corners=2mm,outer sep=0, inner sep=0.75mm, label={[font=\small]50:$t_4$}] (t4) at (0, -4.6) {\begin{tabular}{l} 
					$\wrt{x}{4}$
				\end{tabular}};			
				
				\path (t1.south west) -- (t1.south) coordinate[pos=0.67] (t1x);
				\path (t1.south east) -- (t1.south) coordinate[pos=0.67] (t1y);
				\path (t2.north west) -- (t2.north) coordinate[pos=0.67] (t2x);
				\path (t2.north east) -- (t2.north) coordinate[pos=0.67] (t2y);
				
				\path (t1.south) edge [left] node {$\so$} (t2.north);
				\path (t3.south) edge [left] node {$\so$} (t4.north);
				\path (t0.south) edge [left] node {$\wro_x$} (t1.north);\\
				\path (t0.south east) edge [bend left] node[right, pos=0.2] {$\wro_x$} (t2.north east);
				
			\end{tikzpicture}  
			
		}
		\caption{Current.}
		\label{fig:compute-reordering:a}
	\end{subfigure}
	\hspace{.15cm}
	\centering
	\begin{subfigure}[b]{.23\textwidth}

		\resizebox{.685\textwidth}{!}{
			\begin{tikzpicture}[->,>=stealth',shorten >=1pt,auto,node distance=3cm,
				semithick, transform shape]
				\node[draw, rounded corners=2mm,outer sep=0] (t0) at (0, 0) {\begin{tabular}{l} $\init$ \end{tabular}};
				\node[draw, rounded corners=2mm,outer sep=0, inner sep=0.75mm, label={[font=\small]65:$t_1$}] (t1) at (0, -1.22) {\begin{tabular}{l} 
						$\rd{x}$ \\ $\iabort$
				\end{tabular}};
				\node[draw, rounded corners=2mm,outer sep=0, inner sep=0.75mm, label={[font=\small]50:$t_2$}] (t2) at (0, -4.6) {\begin{tabular}{l} 
						$\rd{x}$
				\end{tabular}};
				\node[draw, rounded corners=2mm,outer sep=0, inner sep=0.75mm, label={[font=\small]50:$t_3$}] (t3) at (0, -2.4) {\begin{tabular}{l} 
						$\wrt{y}{3}$
				\end{tabular}};
				\node[draw, rounded corners=2mm,outer sep=0, inner sep=0.75mm, label={[font=\small]50:$t_4$}] (t4) at (0, -3.6) {\begin{tabular}{l} 
						$\wrt{x}{4}$
				\end{tabular}};			
				
				\path (t1.south west) -- (t1.south) coordinate[pos=0.67] (t1x);
				\path (t1.south east) -- (t1.south) coordinate[pos=0.67] (t1y);
				\path (t2.north west) -- (t2.north) coordinate[pos=0.67] (t2x);
				\path (t2.north east) -- (t2.north) coordinate[pos=0.67] (t2y);
				
				\path (t1.south east) edge [bend left] node [right, pos=0.2] {$\so$} (t2.north east);
				\path (t3.south) edge [left] node {$\so$} (t4.north);
				\path (t0.south) edge [left] node {$\wro_x$} (t1.north);\\
				\path (t4.south) edge [left] node{$\wro_x$} (t2.north);
			\end{tikzpicture}  

		}
		
				\caption{Swap $t_2$ and $t_4$.}
	
		\label{fig:compute-reordering:b}
	\end{subfigure}
	\hspace{.15cm}
	\centering
	\begin{subfigure}[b]{.21\textwidth}
		\resizebox{.58\textwidth}{!}{
				\begin{tikzpicture}[->,>=stealth',shorten >=1pt,auto,node distance=3cm,
				semithick, transform shape]
				\node[draw, rounded corners=2mm,outer sep=0] (t0) at (0, 0) {\begin{tabular}{l} $\init$ \end{tabular}};
				\node[draw, rounded corners=2mm,outer sep=0, inner sep=0.75mm, label={[font=\small]65:$t_1$}] (t1) at (0, -3.5) {\begin{tabular}{l} 
						$\rd{x}$ \\ \pgfsetfillopacity{0.3}$\wrt{y}{1}$ 
				\end{tabular}};
				\node[draw, rounded corners=2mm,outer sep=0, inner sep=0.75mm, opacity=0.3, label={[font=\small]50:$t_2$}] (t2) at (0, -4.7) {\begin{tabular}{l} 
						$\rd{x}$
				\end{tabular}};
				\node[draw, rounded corners=2mm,outer sep=0, inner sep=0.75mm, label={[font=\small]50:$t_3$}] (t3) at (0, -1.05) {\begin{tabular}{l} 
						$\wrt{y}{3}$
				\end{tabular}};
				\node[draw, rounded corners=2mm,outer sep=0, inner sep=0.75mm, label={[font=\small]50:$t_4$}] (t4) at (0, -2.2) {\begin{tabular}{l} 
						$\wrt{x}{4}$
				\end{tabular}};			
				
				\path (t1.south west) -- (t1.south) coordinate[pos=0.67] (t1x);
				\path (t1.south east) -- (t1.south) coordinate[pos=0.67] (t1y);
				\path (t2.north west) -- (t2.north) coordinate[pos=0.67] (t2x);
				\path (t2.north east) -- (t2.north) coordinate[pos=0.67] (t2y);
				
				\path[opacity=0.3] (t1.south) edge [left] node [left] {$\so$} (t2.north);
				\path (t3.south) edge [left] node {$\so$} (t4.north);
				\path (t4.south) edge [left] node {$\wro_x$} (t1.north);\\
			\end{tikzpicture}  
			
		}
		\caption{Swap $t_1$ and $t_4$.}
		\label{fig:compute-reordering:c}
	\end{subfigure}
	\vspace{-2mm}
	\caption{Re-ordering events. All $\so$-edges from $\init$ to other transactions are omitted for legibility. The history order $<$ is represented by the top to bottom order in each figure. Events in gray are deleted from the history.}
	\label{fig:compute-reordering}
	\vspace{-2mm}
\end{figure}

For example, for the program in Figure~\ref{fig:compute-reordering:prog} and history $h$ in Figure~\ref{fig:compute-reordering:a}, $\genericCompute(h)$ would return $(r_1, t_4)$ and $(r_2, t_4)$ where $r_1$ and $r_2$ are the $\rd{x}$ events in $t_1$ and $t_2$ respectively. 

For a pair $(r,t)$, the function $\genericSwap$ produces a new history $\hist'$ which contains all the events ordered before $r$ (w.r.t. $<$), the transaction $t$ and all its $(\so \cup \wro)^*$ predecessors, and the event $r$ reading from $t$. All the other events are removed. Note that the $\po$ predecessors of $r$ from the same transaction are ordered before $r$ by $<$ and they will be also included in $\hist'$. The history $\hist'$ without $r$ is a prefix of the input history $\hist$. By definition, the only pending transaction in $\hist'$ is the one containing the read $r$. The order relation is updated by moving the transaction containing the read $r$ to be the last; it remains unchanged for the rest of the events.
\begin{align*}
\genericSwap(h_<, r, t, \locals) \coloneqq \big( (h' = (h \setminus D) \oplus \wro(t,r), <'), \locals' \big), \text{where $\locals' = \locals \downarrow \events{h'}$} \\
 D = \{e |  r < e \land (\trans{h}{e},t)\not\in (\so \cup \wro)^*\} \mbox{ and }
<' = \big( < \downarrow (\events{h'}\setminus \events{\trans{h'}{r}}) \big) \cdot \trans{h'}{r} 
\end{align*}

Above, $h \setminus D$ is the prefix of $\hist$ obtained by deleting all the events in $D$ from its transaction logs; a transaction log is removed altogether if it becomes empty. Also, $\hist'' \oplus \wro(t,r)$ denotes an \emph{update} of the $\wro$ relation of $\hist''$ where any pair $(\_,r)$ is replaced by $(t,r)$. Finally, $<'' \cdot\ \trans{h'}{r}$ is an extension of the total order $<''$ obtained by appending the events in $\trans{h'}{r}$ according to program order.

Continuing with the example of Figure~\ref{fig:compute-reordering}, when swapping $r_1$ and $t_4$, all the events in transaction $t_2$ belong to $D$ and they will be removed. This is shown in Figure~\ref{fig:compute-reordering:c}. Note that transaction $t_1$ aborted in Figure~\ref{fig:compute-reordering:a} while it will commit in Figure~\ref{fig:compute-reordering:c} (because the value read from $x$ changed). When swapping $r_2$ and $t_4$, no event but the commit in $t_2$ will be deleted (Figure~\ref{fig:compute-reordering:b}).

%

\subsection{Ensuring Optimality}\label{ssec:optimality}

Simply extending histories according to $\genericNext$ and making recursive calls on re-ordered histories whenever they are $I$-consistent guarantees soundness and completeness, but it does not guarantee optimality. Intuitively, the source of redundancy is related to the fact that applying $\genericSwap$ on different histories may give the same result. 

\begin{figure}[t]
	\begin{subfigure}[b]{.25\textwidth}
\begin{adjustbox}{max width=.9\textwidth}
\begin{tabular}{c||c}
\begin{lstlisting}[xleftmargin=5mm,basicstyle=\ttfamily\footnotesize,escapeinside={(*}{*)}, tabsize=1]
begin;
write((*$x$*),2)
commit
\end{lstlisting} &
\begin{lstlisting}[xleftmargin=5mm,basicstyle=\ttfamily\footnotesize,escapeinside={(*}{*)}, tabsize=1]
begin;
a=read((*$x$*));
commit
\end{lstlisting} 
				
\\
\multicolumn{1}{c}{} & \multicolumn{1}{c}{}
\\
\begin{lstlisting}[xleftmargin=5mm,basicstyle=\ttfamily\footnotesize,escapeinside={(*}{*)}, tabsize=1]
begin;
b=read((*$x$*));
commit
\end{lstlisting} &
\begin{lstlisting}[xleftmargin=5mm,basicstyle=\ttfamily\footnotesize,escapeinside={(*}{*)}, tabsize=1]
begin;
write((*$x$*),4);
commit
\end{lstlisting}
\end{tabular} 
\end{adjustbox}
		
		\caption{Program \nver{(4 sessions)}}
		\label{fig:non-optimality:prog}
	\end{subfigure}
	\begin{subfigure}[b]{.17\textwidth}
		\resizebox{0.785\textwidth}{!}{
			\begin{tikzpicture}[->,>=stealth',shorten >=1pt,auto,node distance=3cm,
				semithick, transform shape]
				\node[draw, rounded corners=2mm,outer sep=0] (t1) at (-3, 0) {\begin{tabular}{l} $\init$ \end{tabular}};
				
				\node[draw, rounded corners=2mm,outer sep=0, inner sep=0.75mm, label={[font=\small]30:$t_1$}] (t2) at (-3, -1.0) {\begin{tabular}{l} $\wrt{x}{2}$ \end{tabular}};
				\node[draw, rounded corners=2mm,outer sep=0, inner sep=0.75mm, label={[font=\small]30:$t_2$}] (t3) at (-3, -2.0) {\begin{tabular}{l} $\rd{x}$ \end{tabular}};
				
				\node[draw, rounded corners=2mm,outer sep=0, inner sep=0.75mm, opacity=0.3, label={[font=\small]30:$t_3$}] (t4) at (-3, -3.0) {\begin{tabular}{l} $\rd{x}$ \end{tabular}};
				\node[draw, rounded corners=2mm,outer sep=0, inner sep=0.75mm, opacity=0.3, label={[font=\small]30:$t_4$}] (t5) at (-3, -4.0) {\begin{tabular}{l} $\wrt{x}{2}$ \end{tabular}};
				
				\path (t1.south west) edge[above, bend right=50] node[left, pos=0.1] {$\wro$} (t3.north west);
			\end{tikzpicture}  
			
		}
		\caption{Before reading}
		\label{fig:non_optimality:a}
	\end{subfigure}
	\centering	
	\begin{subfigure}[b]{.17\textwidth}
		\resizebox{0.925\textwidth}{!}{
			\begin{tikzpicture}[->,>=stealth',shorten >=1pt,auto,node distance=3cm,
				semithick, transform shape]
				\node[draw, rounded corners=2mm,outer sep=0] (t1) at (-3, 0) {\begin{tabular}{l} $\init$ \end{tabular}};
				
				\node[draw, rounded corners=2mm,outer sep=0, inner sep=0.75mm, label={[font=\small]40:$t_1$}] (t2) at (-3, -1.0) {\begin{tabular}{l} $\wrt{x}{2}$ \end{tabular}};
				\node[draw, rounded corners=2mm,outer sep=0, inner sep=0.75mm, label={[font=\small]30:$t_2$}] (t3) at (-3, -2.0) {\begin{tabular}{l} $\rd{x}$ \end{tabular}};
				
				\node[draw, rounded corners=2mm,outer sep=0, inner sep=0.75mm, label={[font=\small]35:$t_3$}] (t4) at (-3, -3.0) {\begin{tabular}{l} $\rd{x}$ \end{tabular}};
				
				\node[draw, rounded corners=2mm, label={[font=\small]30:$t_4$}] (t5) at (-3, -4.0) {\begin{tabular}{l} \textbf{\textcolor{blue}{$\wrt{x}{4}$}} \end{tabular}};
				
				\path (t1.south west) edge[above, bend right=50] node[left, pos=0.1] {$\wro$} (t3.north west);
				\path (t1.south east) edge[below, bend left=50] node[right, pos=0.1] {$\wro$} (t4.north east);
			\end{tikzpicture}  
			
		}
		\caption{$t_3$ reads $\init$}
		\label{fig:non_optimality:b}
	\end{subfigure}
	\centering
	\begin{subfigure}[b]{.17\textwidth}
		\resizebox{1.\textwidth}{!}{
			\begin{tikzpicture}[->,>=stealth',shorten >=1pt,auto,node distance=3cm,
				semithick, transform shape]
				\node[draw, rounded corners=2mm,outer sep=0] (t1) at (-3, 0) {\begin{tabular}{l} $\init$ \end{tabular}};
				
				\node[draw, rounded corners=2mm,outer sep=0, inner sep=0.75mm, label={[font=\small]30:$t_1$}] (t2) at (-3, -1.0) {\begin{tabular}{l} $\wrt{x}{2}$ \end{tabular}};
				\node[draw, rounded corners=2mm,outer sep=0, inner sep=0.75mm, label={[font=\small]30:$t_2$}] (t3) at (-3, -2.0) {\begin{tabular}{l} $\rd{x}$ \end{tabular}};
				
				\node[draw, rounded corners=2mm,outer sep=0, inner sep=0.75mm, label={[font=\small]35:$t_3$}] (t4) at (-3, -3.0) {\begin{tabular}{l} $\rd{x}$ \end{tabular}};
				
				\node[draw, rounded corners=2mm, label={[font=\small]30:$t_4$}] (t5) at (-3, -4.0) {\begin{tabular}{l} \textbf{\textcolor{blue}{$\wrt{x}{4}$}} \end{tabular}};

				\path (t1.south west) edge[above, bend right=50] node[left, pos=0.1] {$\wro$} (t3.north west);
				\path (t2.south east) edge[below, bend left=50] node[right, pos=0.1] {$\wro$} (t4.north east);
			\end{tikzpicture}  
			
		}
		\caption{$t_3$ reads $t_1$}
		\label{fig:non_optimality:c}
	\end{subfigure}
	\centering
	\begin{subfigure}[b]{.18\textwidth}
		\resizebox{.615\textwidth}{!}{
			\begin{tikzpicture}[->,>=stealth',shorten >=1pt,auto,node distance=3cm,
				semithick, transform shape]
				\node[draw, rounded corners=2mm,outer sep=0] (t1) at (-3, 0) {\begin{tabular}{l} $\init$ \end{tabular}};
				\node[draw, rounded corners=2mm,outer sep=0, inner sep=0.75mm, label={[font=\small]30:$t_1$}] (t2) at (-3, -1.1) {\begin{tabular}{l} $\wrt{x}{2}$ \end{tabular}};
				\node[draw, rounded corners=2mm,outer sep=0, inner sep=0.75mm, label={[font=\small]30:$t_2$}] (t3) at (-3, -3.3) {\begin{tabular}{l} $\rd{x}$ \end{tabular}};
				
				\node[draw, rounded corners=2mm,outer sep=0, inner sep=0.75mm, opacity=0.3, label={[font=\small]30:$t_3$}] (t4) at (-3, -4.4) {\begin{tabular}{l} $\rd{x}$ \end{tabular}};
				
				\node[draw, rounded corners=2mm, label={[font=\small]30:$t_4$}] (t5) at (-3, -2.2) {\begin{tabular}{l} $\wrt{x}{4}$ \end{tabular}};				
				
				\path (t5) edge[left] node {$\wro$} (t3);
				
			\end{tikzpicture}  
			
		}
		\caption{After swap.}
		\label{fig:non_optimality:d}
	\end{subfigure}
	\vspace{-2mm}	
	\caption{Re-ordering events versus optimality. We assume an oracle order orders transaction from left to right, top to bottom in the program. All transaction logs are history-ordered top to bottom according to their position in the figure. Events in gray are not yet added to the history.}
	\label{fig:non_optimality}	
	\vspace{-5mm}
\end{figure}
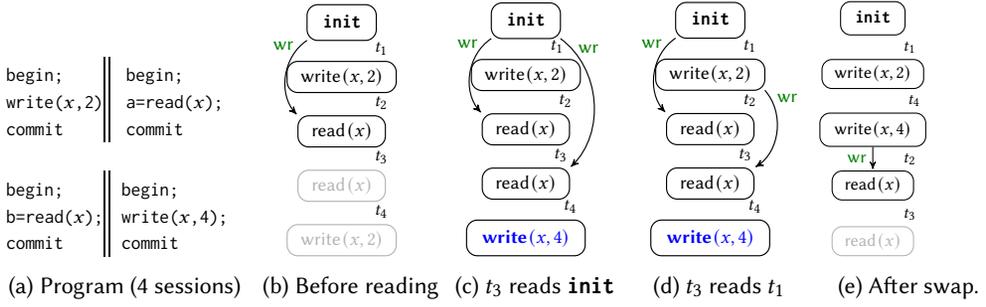

As a first example, consider the program in Figure~\ref{fig:non-optimality:prog} with 2 transactions that only read some variable $x$ and 2 transactions that only write to $x$, each transaction in a different session. Assume that $\textsc{explore}$ reaches the ordered history in Figure~\ref{fig:non_optimality:a} and $\genericNext$ is about to return the second reading transaction. $\textsc{explore}$ will be called recursively on the two histories in Figure~\ref{fig:non_optimality:b} and Figure~\ref{fig:non_optimality:c} that differ in the write that this last read is reading from (the initial write or the first write transaction). On both branches of the recursion, $\genericNext$ will extend the history with the last write transaction written in blue bold font. For both histories, swapping this last write with the first read on $x$ will result in the history in Figure~\ref{fig:non_optimality:d} (cf. the definition of $\genericCompute$ and $\genericSwap$). Thus, both branches of the recursion will continue extending the same history and optimality is violated. The source of non-optimality is related to $\wro$ dependencies that are \emph{removed} during the $\genericSwap$ computation. The histories in Figure~\ref{fig:non_optimality:b} and Figure~\ref{fig:non_optimality:c} differ in the $\wro$ dependency involving the last read, but this difference was discarded during the $\genericSwap$ computation. To avoid this behavior, $\genericSwap$ is enabled only on histories where the discarded $\wro$ dependencies relate to some ``fixed'' set of writes, i.e., latest\footnote{\nver{We use latest writes because they are uniquely defined. In principle, other ways of identifying some unique set of writes could be used.}} writes w.r.t. $<$ that \oldver{preserve consistency}\nver{guarantee consistency by causal extensibility} (see the definition of $\isMaximallyAdded{\_}{\_}$ below). 
\nver{By causal extensibility, a read $r$ can always read from a write which already belongs to its ``causal past'', i.e., predecessors in $(\so \cup \wro)^*$ excluding the $\wro$ dependency for $r$. For every discarded $\wro$ dependency, it is required that the read reads from the latest such write w.r.t. $<$.} In this example, re-ordering is enabled only when the second $\rd{x}$ reads from \oldver{$\ewrt{x,2}$}\nver{the initial write; $\ewrt{x,2}$ does not belong to its ``causal past'' (when the $\wro$ dependency of the read itself is excluded)}.

\begin{figure}[t]
	
\begin{subfigure}[b]{.23\textwidth}
\begin{adjustbox}{max width=\textwidth}
\begin{tabular}{c||c}
\begin{lstlisting}[xleftmargin=5mm,basicstyle=\ttfamily\scriptsize,escapeinside={(*}{*)}, tabsize=1]
begin;
a=read((*$x$*));
commit
\end{lstlisting} &
\begin{lstlisting}[xleftmargin=5mm,basicstyle=\ttfamily\scriptsize,escapeinside={(*}{*)}, tabsize=1]
begin;
b=read((*$y$*));
commit
\end{lstlisting} 
\\
\multicolumn{1}{c}{} & \multicolumn{1}{c}{}
\\
\begin{lstlisting}[xleftmargin=5mm,basicstyle=\ttfamily\scriptsize,escapeinside={(*}{*)}, tabsize=1]
begin;
write((*$y$*),3);
commit
\end{lstlisting} &
\begin{lstlisting}[xleftmargin=5mm,basicstyle=\ttfamily\scriptsize,escapeinside={(*}{*)}, tabsize=1]
begin;
write((*$x$*),4);
commit
\end{lstlisting}
\end{tabular} 
\end{adjustbox}

\caption{Program \nver{(4 sessions)}.}
\label{fig:redundant_Swap:prog}
\end{subfigure}
	\centering
	\begin{subfigure}[b]{.225\textwidth}
		\resizebox{.65\textwidth}{!}{
			\begin{tikzpicture}[->,>=stealth',shorten >=1pt,auto,node distance=3cm,
				semithick, transform shape]
				\node[draw, rounded corners=2mm,outer sep=0] (t0) at (-3, 0) {\begin{tabular}{l} $\init$ \end{tabular}};
				
				\node[draw, rounded corners=2mm,outer sep=0, inner sep=0.75mm, label={[font=\small]50:$t_1$}] (t1) at (-3, -1.075) {\begin{tabular}{l} $\rd{x}$\end{tabular}};
				\node[draw, rounded corners=2mm,outer sep=0, inner sep=0.75mm, label={[font=\small]50:$t_2$}] (t2) at (-3, -2.15) {\begin{tabular}{l} $\rd{y}$  \end{tabular}};
				
				\node[draw, rounded corners=2mm,outer sep=0, inner sep=0.75mm, label={[font=\small]50:$t_3$}] (t3) at (-3, -3.225) {\begin{tabular}{l} $\wrt{y}{3}$ \end{tabular}};
				\node[draw, rounded corners=2mm,outer sep=0, inner sep=0.75mm, label={[font=\small]50:$t_4$}] (t4) at (-3, -4.3) {\begin{tabular}{l} $\wrt{x}{4}$ \end{tabular}};
				
				\path (t0.south west) edge[above, bend right] node[left] {$\wro$} (t2.north west);
				\path (t0.south) edge node[left] {$\wro$} (t1.north);
			\end{tikzpicture}  
		}
		\caption{Current history.}
		\label{fig:redundant_swap:a}
	\end{subfigure}
	\hspace{.25cm}
	\centering
	\begin{subfigure}[b]{.17\textwidth}
		\resizebox{.68\textwidth}{!}{
			\begin{tikzpicture}[->,>=stealth',shorten >=1pt,auto,node distance=3cm,
				semithick, transform shape]
				\node[draw, rounded corners=2mm,outer sep=0] (t0) at (-3, 0) {\begin{tabular}{l} $\init$ \end{tabular}};
				
				\node[draw, rounded corners=2mm,outer sep=0, inner sep=0.75mm, label={[font=\small]50:$t_1$}] (t1) at (-3, -1.1) {\begin{tabular}{l} $\rd{x}$ \end{tabular}};
				\node[draw, rounded corners=2mm,outer sep=0, inner sep=0.75mm, label={[font=\small]50:$t_2$}] (t2) at (-3, -3.3) {\begin{tabular}{l} $\rd{y}$ \end{tabular}};
				
				\node[draw, rounded corners=2mm,outer sep=0, inner sep=0.75mm, label={[font=\small]50:$t_3$}] (t3) at (-3, -2.2) {\begin{tabular}{l} $\wrt{y}{3}$ \end{tabular}};
				\node[draw, rounded corners=2mm,outer sep=0, inner sep=0.75mm, label={[font=\small]50:$t_4$}] (t4) at (-3, -4.4) {\begin{tabular}{l} $\wrt{x}{4}$ \end{tabular}};
				
				\path (t3.south) edge node[left] {$\wro$} (t2.north);
				\path (t0.south) edge node[left] {$\wro$} (t1.north);
			\end{tikzpicture}  
		}
		
		\caption{Swap $t_2$ and $t_3$.}
		\label{fig:redundant_swap:b}
	\end{subfigure}
	\hspace{.25cm}
	\centering
	\begin{subfigure}[b]{.17\textwidth}
		\resizebox{.72\textwidth}{!}{
			\begin{tikzpicture}[->,>=stealth',shorten >=1pt,auto,node distance=3cm,
				semithick, transform shape]
				\node[draw, rounded corners=2mm,outer sep=0] (t0) at (-3, 0) {\begin{tabular}{l} $\init$ \end{tabular}};
				
				\node[draw, rounded corners=2mm,outer sep=0, inner sep=0.75mm, label={[font=\small]50:$t_1$}] (t1) at (-3, -2.2) {\begin{tabular}{l} $\rd{x}$ \end{tabular}};
				
				\node[draw, rounded corners=2mm,outer sep=0, inner sep=0.75mm, label={[font=\small]50:$t_4$}] (t4) at (-3, -1.1) {\begin{tabular}{l} $\wrt{x}{4}$ \end{tabular}};
				
				\path (t4.south) edge node[left] {$\wro$} (t1.north);
			\end{tikzpicture}  
		}
		
		\caption{Swap $t_1$ and $t_4$.}
		\label{fig:redundant_swap:c}
	\end{subfigure}
\vspace{-2mm}
	\caption{Re-ordering the same read on different branches of the recursion.}
	\label{fig:redundant_swap}
	\vspace{-5mm}
\end{figure}
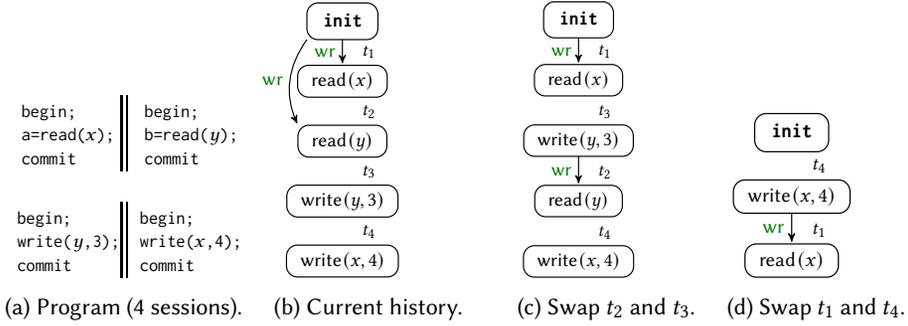

The restriction above is not sufficient, because the two histories for which $\genericSwap$ gives the same result may not be generated during the same recursive call (for different $\wro$ choices when adding a read).
For example, consider the program in Figure~\ref{fig:redundant_Swap:prog} that has four sessions each containing a single transaction. $\textsc{explore}$ may compute the history $h$ pictured in Figure~\ref{fig:redundant_swap:a}. Before adding transaction $t_4$, $\textsc{explore}$ can re-order $t_3$ and $t_2$ and then extend with $t_4$ and arrive at the history $h_1$ in Figure~\ref{fig:redundant_swap:b}. Also, after adding $t_4$, it can re-order $t_1$ and $t_4$ and arrive at the history $h_2$ in Figure~\ref{fig:redundant_swap:c}. However, swapping the same $t_1$ and $t_4$ in $h_1$ leads to the same history $h_2$, thereby, having two recursive branches that end up with the same input and violate optimality. Swapping $t_1$ and $t_4$ in $h_1$ should not be enabled because the $\erd{y}$ to be removed by $\genericSwap$ has been swapped in the past. Removing it makes it possible that this recursive branch explores that $\wro$ choice for $\erd{y}$ again.



The $\genericProtocol$ condition restricting re-orderings requires that the re-ordered history be $I$- consistent and that every read deleted by $\genericSwap$ or the re-ordered read $r$ (whose $\wro$ dependency is modified) reads from a latest valid write, \nver{cf. the example in Figure~\ref{fig:non_optimality},} and it is not already swapped, \nver{cf. the example in Figure~\ref{fig:redundant_swap}} (the set $D$ is defined as in $\genericSwap$):

\vspace{-5mm}
\begin{align*}
& \genericProtocol(h_<, r, t, \locals) \coloneqq \mbox{the history returned by }\genericSwap(h_<, r, t, \locals) \mbox{ satisfies }I\\
&\hspace{3.5cm}\land\ \forall r'\in \readOp{h}\cap (D\cup\{r\}).\ \lnot \swapped{h_<}{r'}\land \isMaximallyAdded{h_<}{r'}{t} 
\end{align*}
\vspace{-5mm}

A read $r$ reads from a \nver{causally} latest valid transaction, denoted as $\isMaximallyAdded{h_<}{r}$, if reading from any other later transaction $t'$ w.r.t. $<$ \nver{which is in the ``causal past'' of $\trans{h_<}{r}$} violates the isolation level $I$. Formally, assuming that $t_r$ is the transaction such that $(t_r,r)\in\wro$ in $h$,

\vspace*{-3.5mm}
\begin{equation*}
\isMaximallyAdded{h_<}{r}{t} \coloneqq 
t_r = \max_{<} \left\{\begin{array}{ccc}
	\writeVar{t'}{\mathit{var}(r)} \land (t', \trans{h_<}{r}) \in (\so \cup \wro)^* \text{ in } h' \\
	\land \ h' \oplus r\oplus \wro(t',r) \models I 
	\end{array}
	\right\} \\
\end{equation*}
where $h' = h \setminus \{e \ | \ r \leq e \land (\trans{h}{e}, t) \not\in (\so \cup \wro)^* \}$.

We say that a read $r$ is \emph{swapped} in $\hist_<$ when (1) $r$ reads from a transaction $t$ that is a successor in the oracle order $<_{\ora}$ (the transaction was added by $\genericNext$ after the read), which is now a predecessor\footnote{The $\textsc{explore}$ maintains the invariant that every read follows the transaction it reads from in the history order $<$.} in the history order $<$, (2) there is no transaction \nver{$t'$} that is before $r$ in both $<_{\ora}$ and $<$, and which is a $(\so \cup \wro)^+$ successor of $t$, and (3) $r$ is the first read in its transaction to read from $t$. Formally, assuming that $t$ is the transaction such that $(t,r)\in\wro$, 
%
\vspace{-0.5mm}
\begin{align*}
& \swapped{h_<}{r}  \coloneqq t < r \land t >_{\ora} r 
\land \forall t' \in h.\ t' <_{\ora} \trans{h}{r} \implies (r < t'  \lor (t,t')\not\in (\so \cup \wro)^+) \\
&\hspace{2.5cm} \land\ \forall r' \in \readOp{h}.\ (t,r')\in\wro \implies (r',r)\not\in \po  
\end{align*}

\nver{Condition (1) states a quite straightforward fact about swaps: $r$ could not have been involved in a swap if it reads from a predecessor in the oracle order which means that it was added by $\genericNext$ after the transaction it reads from. Conditions (2) and (3) are used to exclude spurious classifications as swapped reads. Concerning condition (2), suppose that in a history $h$ we swap a transaction $t$ with respect a
(previous) read event $r$. Later on, the algorithm may add a read $r'$ reading also from $t$. Condition (2) forbids $r'$ to be declared as swapped.
Indeed, taking $tr(h, r)$ as an instantiation of $t'$, $tr(h, r)$ is before $r'$ in both
$<_{or}$ and $<$ and it reads from the same transaction as $r'$, thereby, being a $(\so \cup \wro)^+$ successor of the transaction read by $r'$.
Condition (3) forbids that, after swapping $r$ and $t$ in $h$, later
read events from the same transaction as $r$ can be considered as swapped.}

\nver{
Showing that $I$-completeness holds despite discarding re-orderings is quite challenging. Intuitively, it can be shown that if some $\genericSwap$ is \emph{not} enabled in some history $h_<$ for some pair $(r,t)$ although the result would be $I$-consistent (i.e., $\genericProtocol(h_<, r, t, \locals)$ does not hold because some deleted read is swapped or does not read from a causally latest transaction), then the algorithm explores another history $h'$ which coincides with $h$ except for those deleted reads who are now reading from causally latest transactions. Then, $h'$ would satisfy $\genericProtocol(h_<, r, t, \locals)$, and moreover applying $\genericSwap$ on $h'$ for the pair $(r,t)$ would lead to the same result as applying $\genericSwap$ on $h$, thereby, ensuring completeness. 
}

\subsection{Correctness}\label{ssec:corr}

The following theorem states the correctness of the algorithm presented in this section:
\begin{restatable}{theorem}{algoCorrect}
\label{th:corr}
For any prefix-closed and causally extensible isolation level $I$, $\textsc{explore-ce}$ is $I$-sound, $I$-complete, strongly optimal, and polynomial space.
\end{restatable}

$I$-soundness is a consequence of the $\genericValidWrites$ and $\genericProtocol$ definitions which guarantee that all histories given to recursive calls are $I$-consistent, and of the $\genericSwap$ definition which ensures to only produce feasible histories (which can be obtained using the operational semantics defined in Section~\ref{ssec:semantics}). The fact that this algorithm never engages in fruitless explorations follows easily from causal-extensibility which ensures that any current history can be extended with any event returned by $\genericNext$. Polynomial space is also quite straightforward since the \textbf{for all} loops in Algorithm~\ref{algorithm:algo-class} have a linear number of iterations: the number of iterations of the loop in $\textsc{explore}$, resp., $\textsc{exploreSwaps}$, is bounded by the number of write, resp., read, events in the current history (which is smaller than the size of the program; recall that we assume bounded programs with no loops as usual in SMC algorithms). On the other hand, the proofs of $I$-completeness and optimality are quite complex. \oldver{For $I$-completeness, given a history $\hist\in \histOf[I]{\prog}$, we first show how to effectively compute the order $<$ associated to $\hist$ if $\textsc{explore}$ were to output this history. Then, we define an ``inverse'' of $\textsc{explore-sc}$ and an invariant on ordered histories reached by $\textsc{explore-sc}$, which enable constructing a chain of ordered histories going backwards from $(\hist,<)$ to the initial empty history. Then, we show that $\textsc{explore-sc}$ can proceed forwardly  in this chain and indeed output $\hist$.}


\nver{$I$-completeness means that for any given program $\prog$, the algorithm outputs every history $h$ in $\histOf[I]{\prog}$. The proof of $I$-completeness defines a sequence of histories produced by the algorithm starting with an empty history and ending in $h$, for every such history $h$. 
It consists of several steps:}
\begin{enumerate}
    \item \nver{Define a \emph{canonical} total order $<$ for every unordered partial history $h$, such that if the algorithm reaches $h_{<'}$, for some order $<'$, then $<$ and $<'$ coincide. This canonical order is useful in future proof steps as it allows to extend several definitions to arbitrary histories that are not necessarily reachable, such as $\genericProtocol$ or $\textsc{swapped}$.}
    \item \nver{Define the notion of \emph{$\ora$-respectfulness}, an invariant satisfied by every (partial) ordered history reached by the algorithm. Briefly, a history is $\ora$-respectful if it has only one pending transaction and for every two events $e, e'$ such that $e <_{\ora} e'$, either $e < e'$ or there is a swapped event $e''$ in between. 
 } 
    \item \nver{Define a deterministic function $\prev$ which takes as input a partial history (not necessarily reachable), such that if $h$ is reachable, then $\prev(h)$ returns the history computed by the algorithm just before $h$ (i.e., the previous history in the call stack). Prove that if a history $h$ is $\ora$-respectful, then $\prev(h)$ is also $\ora$-respectful.}
    
    \item \nver{Deduce that if $h$ is $\ora$-respectful, then there is a finite collection of $\ora$-respectful histories $H_h = \{h_i\}_{i=0}^{n}$ such that $h_n = h$, $h_0 = \emptyset$, and $h_i = \prev(h_{i+1})$ for each $i$. The $\ora$-respectfulness invariant and the causal-extensibility of the isolation level are key to being able to construct such a collection. In particular, they are used to prove that $h_i$ has at most the same number of swapped events as $h_{i+1}$ and in case of equality, $h_i$ contain exactly one event less than $h_{i+1}$, which implies that the collection is indeed finite.}
    
    \item \nver{Prove that if $h$ is $\ora$-respectful and $\prev(h)$ is reachable, then $h$ is also reachable. Conclude by induction that every history in $H_h$ is reachable, as $h_0$ is the initial state and $h_i = \prev(h_{i+1})$.}
\end{enumerate}

\nver{The proof of strong optimality relies on arguments employed for $I$-completeness. It can be shown that if the algorithm would reach a (partial) history $h$ twice, then for one of the two exploration branches, the history $h'$ computed just before $h$ would be different from $\prev(h)$, which contradicts the definition of $\prev(h)$.}



\nver{In terms of time complexity, the $\textsc{explore-ce}(I)$ algorithm achieves polynomial time between consecutive outputs for isolation
levels $I$ where checking $I$-consistency of a history is polynomial time, e.g., \RC, \RA, and \CC.} 
\section{Swapping-based model checking for Snapshot Isolation and Serializability}\label{sec:ser}

For $\textsc{explore-ce}$, the part of strong optimality concerning \emph{not} engaging in fruitless explorations was a direct consequence of causal extensibility (of the isolation level). However, isolation levels such as SI and SER are \emph{not} causally extensible (see Section~\ref{ssec:causal_ext}). Therefore, the question we investigate in this section is whether there exists another implementation of $\textsc{explore}$ that can ensure strong optimality along with $I$-soundness and $I$-completeness for $I$ being SI or SER. We answer this question in the negative, and as a result, propose an SMC algorithm that extends $\textsc{explore-ce}$ by just filtering histories before outputting to be consistent with SI or SER.

\vspace{-1mm}
\begin{restatable}{theorem}{impossibility}
\label{th:impos}
If $I$ is \textit{Snapshot Isolation} or \textit{Serializability}, there exists no $\textsc{explore}$ algorithm that is $I$-sound, $I$-complete, and strongly optimal.
\vspace{-1mm}
\end{restatable}

The proof of Theorem~\ref{th:impos} \appendixver{given in App.~\ref{proof:impossibility} }defines a program with two transactions and shows that 
any concrete instance of $\textsc{explore}$ in Alg.~\ref{algorithm:algo-class} \emph{cannot be both} $I$-complete and strongly optimal.

Given this negative result, we define an implementation of $\textsc{explore}$ for an isolation level $I\in\{SI,SER\}$ that ensures optimality instead of strong optimality, along with soundness, completeness, and polynomial space bound. Thus, let $\textsc{explore-ce}(I_0)$ be an instance of $\textsc{explore-ce}$ parametrized by $I_0\in \{\RC,\RA,\CC\}$. We define an implementation of $\textsc{explore}$ for $I$, denoted by $\textsc{explore-ce}^*(I_0,I)$, which is exactly $\textsc{explore-ce}(I_0)$ except that 
instead of $\genericEvaluate(\hist) ::= true$, it uses 
\vspace{-1mm}
\begin{equation*}
	\begin{array}{ccc}
		\genericEvaluate(\hist) & \coloneqq & \hist \mbox{ satisfies }I
	\end{array}
\vspace{-1mm}
\end{equation*}
$\textsc{explore-ce}^*(I_0,I)$ enumerates exactly the same histories as $\textsc{explore-ce}(I_0)$ except that it outputs only histories consistent with $I$. The following is a direct consequence of Theorem~\ref{th:corr}.

\vspace{-1mm}
\begin{corollary}
For any isolation levels $I_0$ and $I$ such that $I_0$ is prefix-closed and causally extensible, and $I_0$ is weaker than $I$, $\textsc{explore-ce}^*(I_0,I)$ is $I$-sound, $I$-complete, optimal, and polynomial space.
\end{corollary}




\section{Experimental evaluation}\label{sec:exp}

We evaluate an implementation of $\textsc{explore-ce}$ and $\textsc{explore-ce}^*$ in the context of the Java Pathfinder (JPF)~\cite{DBLP:conf/issta/VisserPK04} model
checker for Java concurrent programs. As benchmark, we use bounded-size client programs of a number of database-backed applications drawn from the literature. 
The experiments were performed on an Apple M1 with $8$ cores and $16$ GB of RAM.


\subsection{Implementation}

We implemented our algorithms as an extension of the \texttt{DFSearch} class in JPF. For performance reasons, we implemented an iterative version of these algorithms where roughly, inputs to recursive calls are maintained as a collection of histories instead of relying on the call stack. For checking consistency of a history with a given isolation level, we implemented the algorithms proposed by \citet{DBLP:journals/pacmpl/BiswasE19}. 


Our tool takes as input a Java program and isolation levels as parameters. We assume that the program uses a fixed API for interacting with the database, similar to a key-value store interface. This API consists of specific methods for starting/ending a transaction, and reading/writing a global variable. The fixed API is required for being able to maintain the database state separately from the JVM state (the state of the Java program) and update the current history in each database access. This relies on a mechanism for ``transferring'' values read from the database state to the JVM state. 

\vspace{-1mm}
\subsection{Benchmark}

We consider a set of benchmarks inspired by real-world applications and evaluate them under different types of client programs and isolation levels.

\textit{Shopping Cart~\cite{sivaramakrishnan2015declarative5}} allows users to add, get and remove items from their shopping cart and modify the quantities of the items present in the cart. 

\textit{Twitter~\cite{difallah2013oltp}} allows users to follow other users, publish tweets and get their followers, tweets and tweets published by other followers. 

\textit{Courseware~\cite{DBLP:conf/esop/NairP020}} manages the enrollment of students in courses in an institution. It allows to open, close and delete courses, enroll students and get all enrollments. One student can only enroll to a course if it is open and its capacity has not reached a fixed limit. 

\textit{Wikipedia~\cite{difallah2013oltp}} allows users to get the content of a page (registered or not), add or remove pages to their watching list and update pages. 

\textit{TPC-C~\cite{TPCC}} models an online shopping application with five types of transactions: reading the stock of a product, creating a new order, getting its status, paying it and delivering it. 

\nver{SQL tables are modeled using a ``set'' global variable whose content is the set of ids (primary keys) of the rows present in the table, and a set of global variables, one variable for each row in the table (the name of the variable is the primary key of that row). SQL statements such as INSERT and DELETE statements are modeled as writes on that ``set'' variable while SQL statements with a WHERE clause (SELECT, JOIN, UPDATE) are compiled to a read of the table's set variable followed by reads or writes of variables that represent rows in the table (similarly to~\cite{DBLP:journals/pacmpl/BiswasKVEL21}).} 

\vspace{-1mm}
\subsection{Experimental Results}

We designed three experiments where we compare the performance of a baseline model checking algorithm, $\textsc{explore-ce}$ and $\textsc{explore-ce}^*$ for different (combinations of) isolation levels, and we explore the scalability of $\textsc{explore-ce}$ when increasing the number of sessions and transactions per session, respectively.
\appendixver{All data is exhibited in Appendix~\ref{sec:appendix:experiments} in the supplementary material \cite{bouajjani2023dynamic}. }For each experiment we report running time, memory consumption, and the number of end states, i.e., histories of complete executions and in the case of $\textsc{explore-ce}^*$, before applying the $\genericEvaluate$ filter. As the number of end states for a program on a certain isolation level increases, the running time of our algorithms naturally increases as well. 




The first experiment  compares the performance of our algorithms for different combinations of isolation levels and a baseline model checking algorithm that performs no partial order reduction. We consider as benchmark five \nver{(independent)} client programs\footnote{For an application that defines a number of transactions, a client program consists of a number of sessions, each session containing a sequence of transactions defined by the application.} for each application described above ($25$ in total), each \nver{program} with 3 sessions and 3 transactions per session. Running time, memory consumption, and number of end states are reported in Fig.~\ref{fig:results1} as \nver{cactus plots~\cite{DBLP:conf/issac/BrainDG17}.} 

\begin{figure}[t]
	\centering
	\begin{subfigure}[t]{0.322\linewidth}
		\centering
		\includegraphics[width=\textwidth]{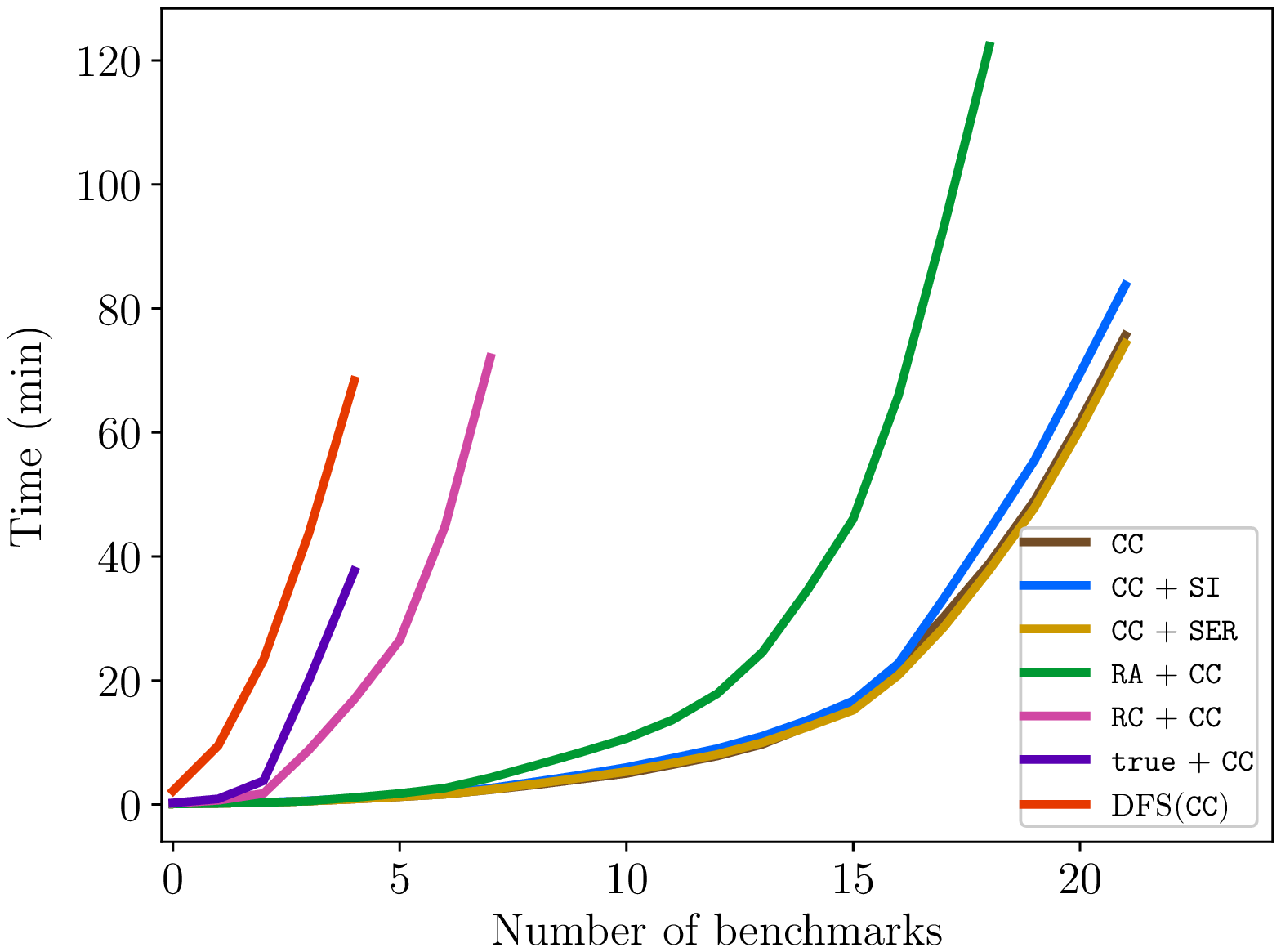}
		\vspace{-5mm}
		\caption{Running time.}
		\label{fig:results-ra-cc}
	\end{subfigure}
	\hfill
	\begin{subfigure}[t]{0.322\linewidth}
		\centering
		\includegraphics[width=\textwidth]{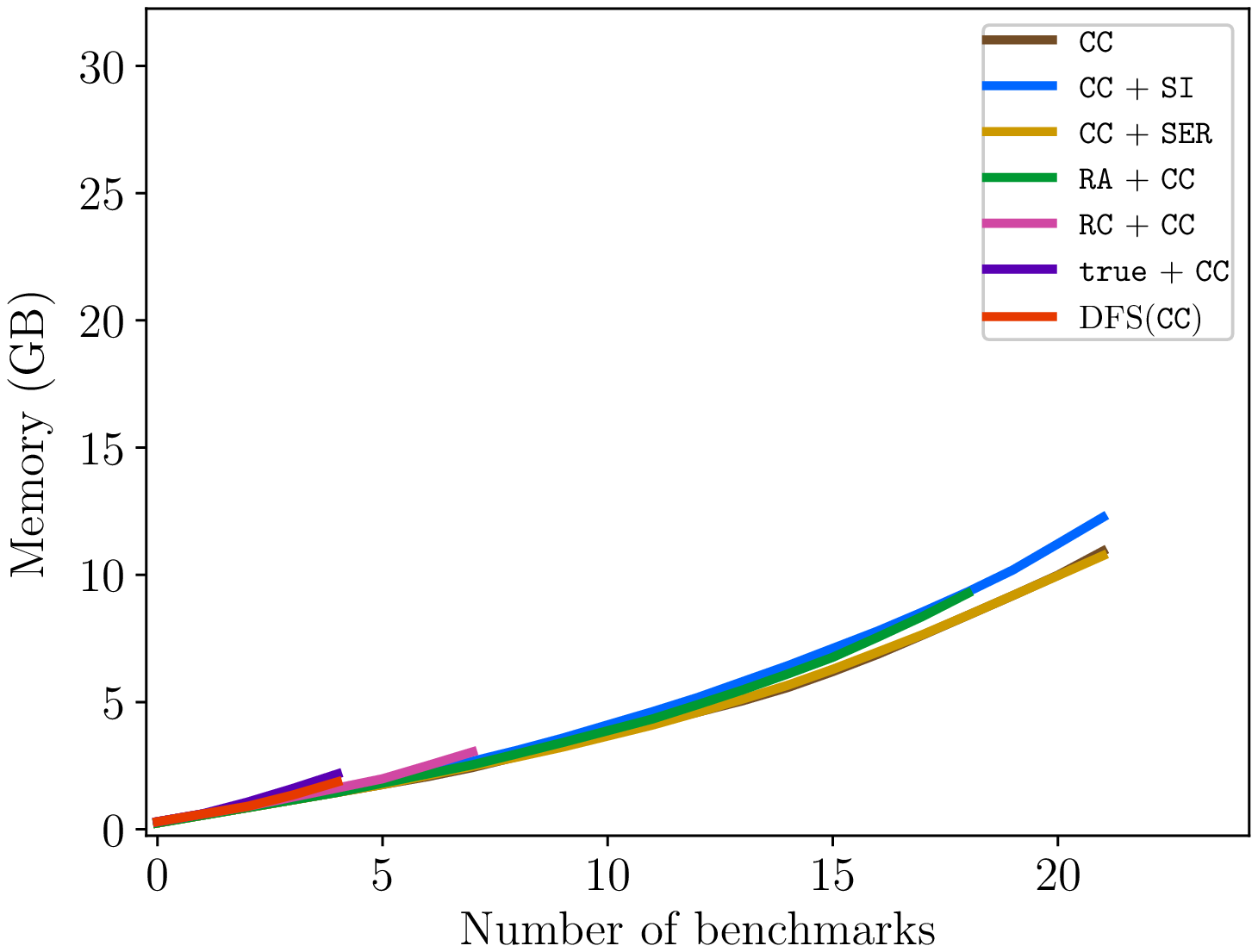}
		\vspace{-5mm}
		\caption{Memory consumption.}
		\label{fig:results-memory}
	\end{subfigure}
	\hfill
	\begin{subfigure}[t]{0.322\linewidth}
		\centering
		\includegraphics[width=\textwidth]{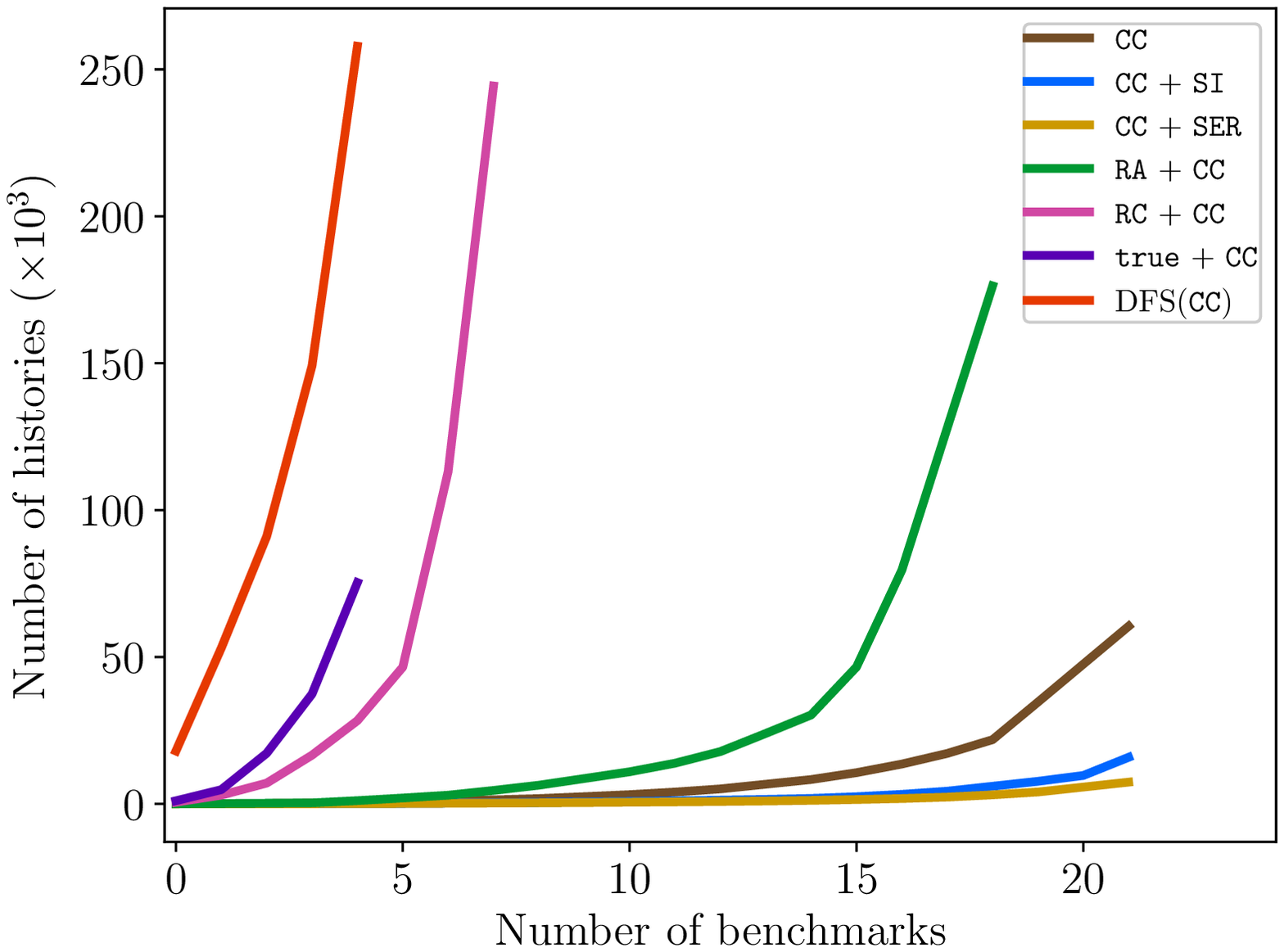}
		\vspace{-5mm}
		\caption{End states.}
		\label{fig:results-histories}
	\end{subfigure}
\vspace{-3mm}
	\caption[where every history is consistent]{Cactus plots comparing different algorithms in terms of time, memory, and end states. For readability, we use $\CC$ to denote $\textsc{explore-ce}$ under $\CC$, $I_1 + I_2$ stands for $\textsc{explore-ce}^*(I_1, I_2)$, and $\texttt{true}$ is the trivial isolation level where every history is consistent. \nver{Differences between $\CC$, $\CC + \SI$ and $\CC + \SER$ are very small and their graphics overlap.} Moreover, $\textsc{DFS}(\CC)$ denotes a standard DFS traversal of the semantics defined in Section~\ref{ssec:semantics}. These plots exclude benchmarks that timeout ($30$ mins): $3$ benchmarks for $\CC$, $\langle \SI, \CC \rangle$ and $\langle \SER, \CC \rangle$ and $6$, $17$, $20$ and $20$ benchmarks timeout for $\langle \RA, \CC \rangle$, $\langle \RC, \CC \rangle$, $\langle \texttt{true}, \CC \rangle$ and $\textsc{DFS}(\CC)$ respectively.}
	\label{fig:results1}
\vspace{-4mm}
\end{figure}


To justify the benefits of partial order reduction, we implement a baseline model checking algorithm $\textsc{DFS}(\CC)$ that performs a standard DFS traversal of the execution tree w.r.t. the formal semantics defined in Section~\ref{ssec:semantics} for $\CC$ (for fairness, we restrict interleavings so at most one transaction is pending at a time). This baseline algorithm may explore the same history multiple times since it includes no partial order reduction mechanism. In terms of time, $\textsc{DFS}(\CC)$ behaves poorly: it timeouts for $20$ out of the $25$ programs and it is less efficient even when it terminates. We consider a timeout of $30$ mins. In comparison the strongly optimal algorithm $\textsc{explore-ce}(\CC)$ (under $\CC$) finishes in \oldver{$17$ seconds in average.}\nver{in $3' 26''$ seconds in average (counting timeouts).} $\textsc{DFS}(\CC)$ is \nver{similiar to $\textsc{explore-ce}(\CC)$}\oldver{also worse} in terms of memory consumption. The memory consumption of $\textsc{DFS}(\CC)$ is \oldver{$441$MB}\nver{$381$MB} in average, compared to \oldver{$317$MB}\nver{$508$MB} for $\textsc{explore-ce}(\CC)$ (JPF forces a minimum consumption of $256$MB).

To show the benefits of \emph{strong} optimality, we compare $\textsc{explore-ce}(\CC)$ which is strongly optimal with ``plain'' optimal algorithms $\textsc{explore-ce}^*(I_0,\CC)$ for different levels $I_0$. As shown in Figure~\ref{fig:results1}(a), $\textsc{explore-ce}(\CC)$ is more efficient time-wise than every ``plain'' optimal algorithm, and the difference in performance grows as $I_0$ becomes weaker. In the limit, when $I_0$ is the trivial isolation level $\texttt{true}$ where every history is consistent, $\textsc{explore-ce}^*(\texttt{true},\CC)$ timeouts for \nver{$20$}\oldver{$12$} out of the $25$ programs. The average speedup (average of individual speedups) of $\textsc{explore-ce}(\CC)$ w.r.t. $\textsc{explore-ce}^*(\RA,\CC)$, $\textsc{explore-ce}^*(\RC,\CC)$ and $\textsc{explore-ce}^*(\texttt{true},\CC)$ is \oldver{$2$, $31$, and $54$}\nver{$3$, $18$ and $15$.} respectively (we exclude timeout cases when computing speedups). All algorithms consume around \oldver{300 MB}\nver{$500$MB} of memory in average.

%

For the $\SI$ and $\SER$ isolation levels that admit no strongly optimal $\textsc{explore}$ algorithm, we observe that the overhead of $\textsc{explore-ce}^*(\CC,\SI)$ or $\textsc{explore-ce}^*(\CC,\SER)$ relative to $\textsc{explore-ce}(\CC)$ is negligible (the corresponding lines in Figure~\ref{fig:results1} are essentially overlapping). This is due to the fact that the consistency checking algorithms of \citet{DBLP:journals/pacmpl/BiswasE19} are polynomial time when the number of sessions is fixed, which makes them fast at least on histories with few sessions.

In our second experiment, we investigate the scalability of $\textsc{explore-ce}$ when increasing the number of sessions. For each $i\in [1,5]$, we consider 5 \nver{(independent)} client programs for TPC-C and 5 for Wikipedia\oldver{\footnote{We focus on these two applications because they are of similar complexity.}}
(10 in total) with $i$ sessions, each session containing 3 transactions. We start with 10 programs with 5 sessions, and remove sessions one by one to obtain programs with fewer sessions. We take $\CC$ as isolation level. 
The plot in Figure~\ref{fig:results-threads} shows average running time and memory consumption for each number $i\in [1,5]$ of sessions. 
As expected, increasing the number of sessions is a bottleneck running time wise because the number of histories increases significantly. However, memory consumption does not grow with the same trend, cf. the polynomial space bound.


\begin{figure}[t]
	\vspace{-4mm}
	\centering
	\begin{subfigure}[b]{0.322\linewidth}
		\centering
		\includegraphics[width=\textwidth]{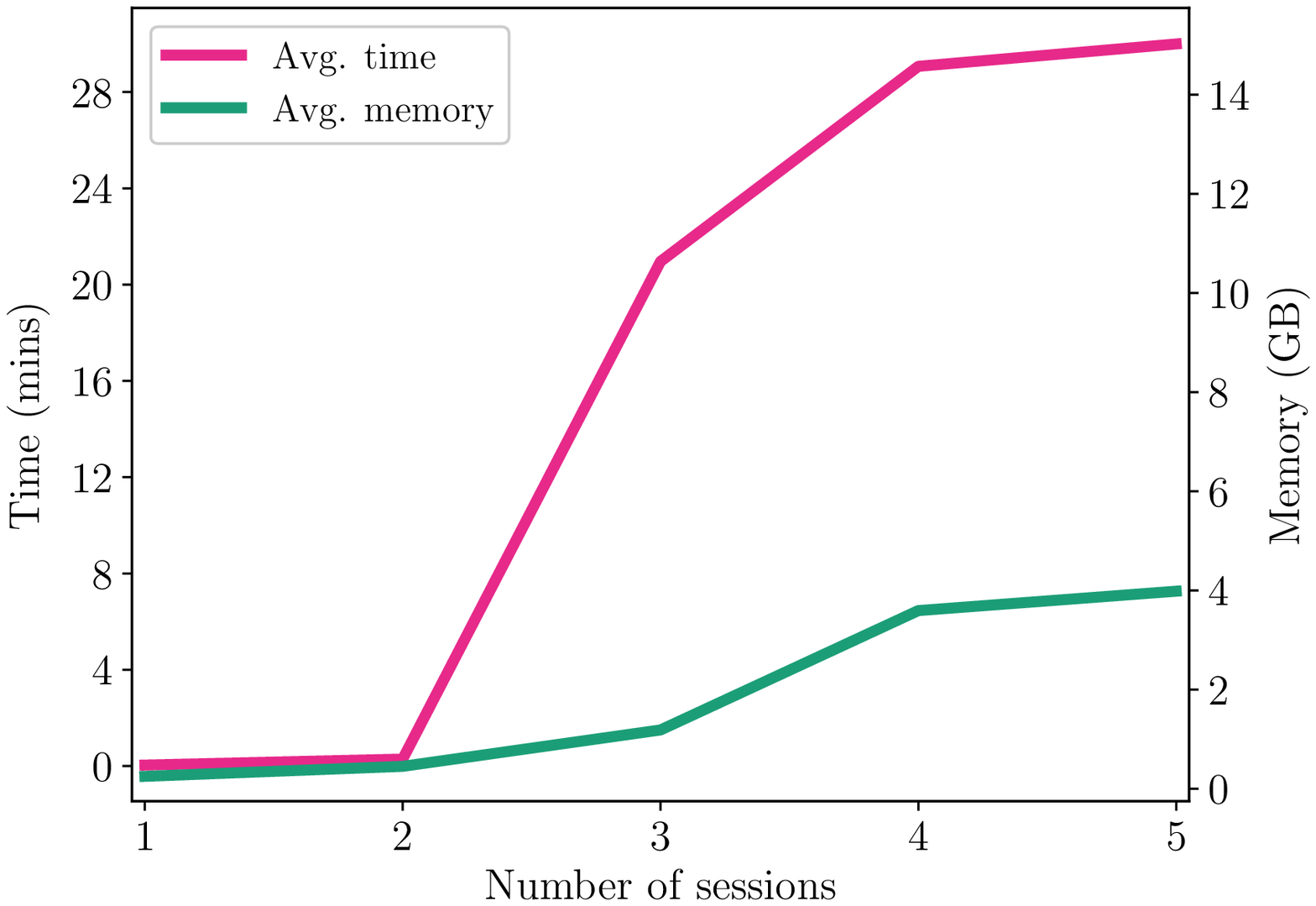}
		\vspace{-5mm}
		\caption{Increasing sessions.}
		\label{fig:results-threads}
	\end{subfigure}
	\hspace{1cm}
	\begin{subfigure}[b]{0.422\linewidth}
		\centering
		\includegraphics[width=0.76\textwidth]{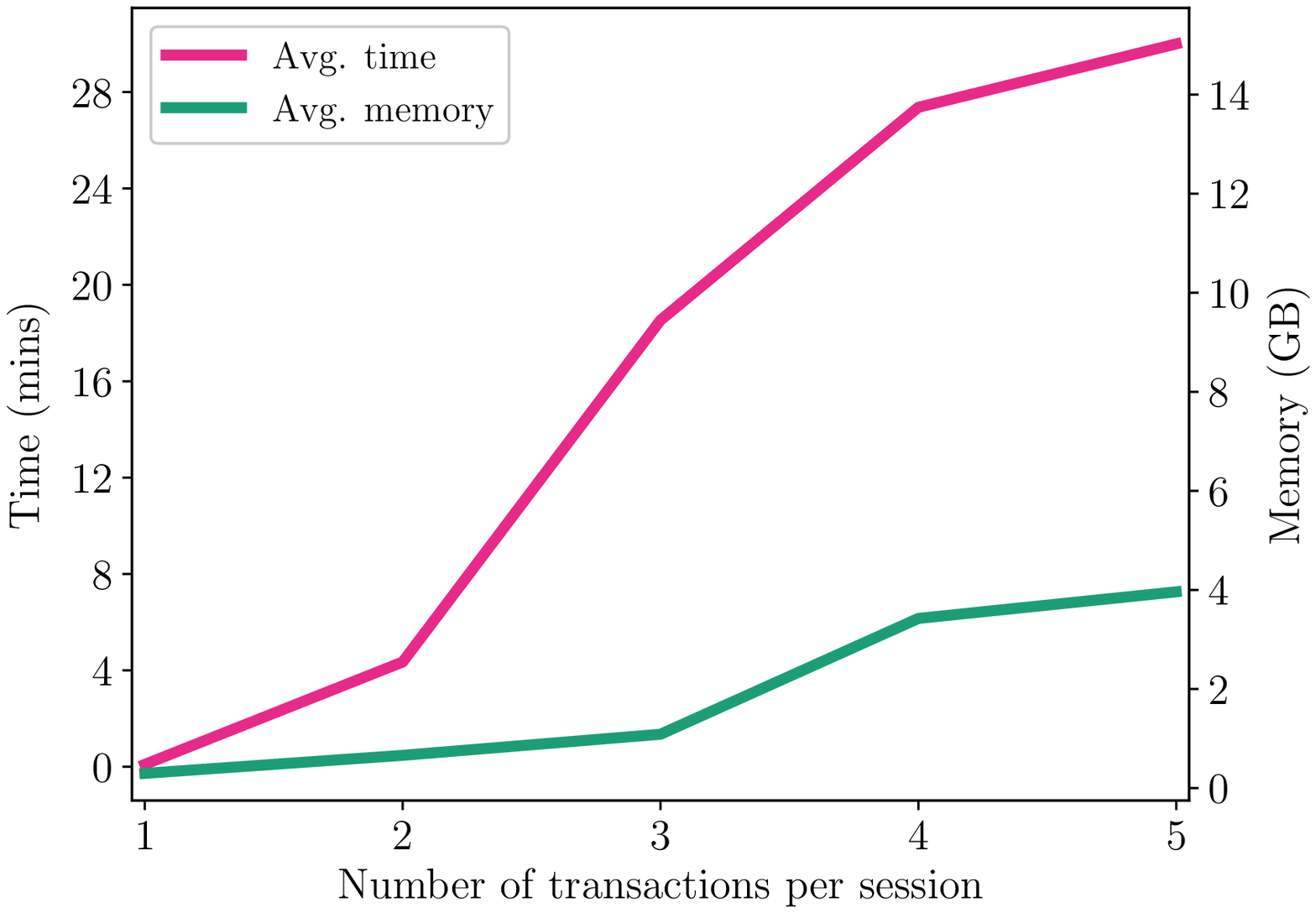}
		\vspace{-1mm}
		\caption{Increasing transactions per session.}
		\label{fig:results-transactions}
	\end{subfigure}
\vspace{-3.5mm}
	\caption{Evaluating the scalability of $\textsc{explore-ce}(\CC)$ for TPC-C and Wikipedia client programs when increasing their size. These plots include benchmarks that timeout ($30$ mins): $4$, $9$ and $10$ for $3$, $4$ and $5$ sessions respectively in Figure~\ref{fig:results-threads}, and $5$, $8$ and $10$ for $3$, $4$ and $5$ transactions per sessions respectively in Figure~\ref{fig:results-transactions}.}
	\label{fig:results2}
\vspace{-4mm}
\end{figure}

Finally, we evaluate the scalability of $\textsc{explore-ce}(\CC)$ when increasing the number of transactions per session. We consider 5 \nver{(independent)} TPC-C client programs and 5 \nver{(independent)} Wikipedia programs with $3$ sessions and $i$ transactions per session, for each $i\in [1,5]$. 
Figure~\ref{fig:results-transactions} shows average running time and memory consumption for each number $i\in [1,5]$ of transactions per session. Increasing the number of transactions per session is a bottleneck for the same reasons.


\section{Related Work}\label{sec:related}

\noindent
\textbf{Checking Correctness of Database-Backed Applications.}
One line of work is concerned with the logical formalization
of isolation levels 
\cite{ansi,DBLP:conf/icde/AdyaLO00,DBLP:conf/sigmod/BerensonBGMOO95,DBLP:conf/concur/Cerone0G15,DBLP:journals/pacmpl/BiswasE19}.
Our work relies on the axiomatic definitions of isolation levels introduced by \citet{DBLP:journals/pacmpl/BiswasE19}, which have also investigated
the problem of checking whether a given history satisfies a certain isolation
level. Our SMC algorithms rely on these algorithms to check consistency of a history with a given isolation level. 

Another line of work focuses on the problem of finding ``anomalies'': 
behaviors that are not possible under serializability. This is typically done
via a static analysis of the application code that builds a static dependency graph that
over-approximates the data dependencies in all possible
executions of the application~\cite{DBLP:journals/jacm/CeroneG18,DBLP:conf/concur/0002G16,DBLP:journals/tods/FeketeLOOS05,DBLP:conf/vldb/JorwekarFRS07,acidrain,isodiff}.
Anomalies with respect to a given isolation level then correspond to a
particular class of cycles in this graph. Static dependency graphs turn out to
be highly imprecise in representing feasible executions, leading to false
positives. Another source of false positives is that an anomaly might not be a
bug because the application may already be designed to handle the
non-serializable behavior \cite{DBLP:conf/pldi/BrutschyD0V18,isodiff}. 
Recent work has tried to address these issues by using more precise 
logical encodings of the application~\cite{DBLP:conf/popl/BrutschyD0V17,DBLP:conf/pldi/BrutschyD0V18}, or
by using user-guided heuristics~\cite{isodiff}. 
Another approach consists of modeling the application
logic and the isolation level in first-order logic and relying on SMT solvers to
search for anomalies~\cite{DBLP:journals/pacmpl/KakiESJ18,DBLP:conf/concur/NagarJ18,burcu-netys},
or defining specialized reductions to assertion
checking~\cite{DBLP:conf/concur/BeillahiBE19,DBLP:conf/cav/BeillahiBE19}.
Our approach, based on SMC, does not generate false positives because we systematically
enumerate only valid executions of a program which allows to check for user-defined assertions.

Several works have looked at the problem of reasoning about the correctness of
applications executing under weak isolation and introducing additional
synchronization when
necessary~\cite{DBLP:conf/eurosys/BalegasDFRPNS15,DBLP:conf/popl/GotsmanYFNS16,DBLP:conf/esop/NairP020,DBLP:conf/usenix/0001LCPRV14}.
These are based on static analysis or logical proof arguments. 
The issue of repairing applications is orthogonal to our work. 

MonkeyDB~\cite{DBLP:journals/pacmpl/BiswasKVEL21} is a mock storage system for testing storage-backed applications. 
While being able to scale to larger code, it has the inherent incompleteness of testing. As opposed to MonkeyDB, our algorithms perform a systematic and complete exploration of executions and can establish correctness at least in some bounded context, and they avoid redundancy, enumerating equivalent executions multiple times. Such guarantees are beyond the scope of MonkeyDB.

\vspace{.5mm}
\noindent
\textbf{Dynamic Partial Order Reduction.}
\citet{DBLP:journals/jacm/AbdullaAJS17} introduced the concept of \emph{source sets} which provided the first strongly optimal DPOR algorithm for Mazurkiewicz trace equivalence. Other works study DPOR techniques for coarser equivalence relations, e.g.,~\cite{DBLP:journals/pacmpl/AbdullaAJLNS19,DBLP:conf/cav/AgarwalCPPT21,DBLP:conf/tacas/AronisJLS18,DBLP:journals/pacmpl/ChalupaCPSV18,DBLP:journals/pacmpl/ChatterjeePT19}. In all cases, the space complexity is exponential when strong optimality is ensured. 


Other works focus on extending DPOR to weak memory models either by targeting a specific memory model~\cite{DBLP:conf/cav/AbdullaAJL16,DBLP:journals/acta/AbdullaAAJLS17,DBLP:journals/pacmpl/AbdullaAJN18,DBLP:conf/oopsla/NorrisD13} or by being parametric with respect to an axiomatically-defined memory model~\cite{DBLP:conf/pldi/Kokologiannakis19,DBLP:conf/asplos/Kokologiannakis20,DBLP:journals/pacmpl/Kokologiannakis22}. Some of these works can deal with the coarser reads-from equivalence, e.g.,~\cite{DBLP:journals/pacmpl/AbdullaAJN18,DBLP:conf/pldi/Kokologiannakis19,DBLP:conf/asplos/Kokologiannakis20,DBLP:journals/pacmpl/Kokologiannakis22}. Our algorithms build on the work of \citet{DBLP:journals/pacmpl/Kokologiannakis22} which for the first time, proposes a DPOR algorithm which is both strongly optimal and polynomial space. The definitions of database isolation levels are quite different with respect to weak memory models, which makes these previous works not extensible in a direct manner. These definitions include a semantics for \emph{transactions} which are collections of reads and writes, and this poses new difficult challenges. 
For instance, reasoning about the completeness and the (strong) optimality of existing DPOR algorithms for shared-memory is agnostic to the scheduler ($\genericNext$ function) while the strong optimality of our $\textsc{explore-ce}$ algorithm relies on the scheduler keeping at most one transaction pending at a time. \nver{In addition, unlike TruSt, \textsc{explore-ce} ensures that no swapped events can be swapped again and that the history order $<$ is an extension of $\so \cup \wro$.} This makes our completeness and optimality proofs radically different. 
\nver{Moreover, even for transactional programs with one access per transaction, where $\SER$ and $\SC$ are equivalent, TruSt under $\SC$ and $\textsc{explore-ce}^*(I_0,\SER)$ do not coincide, for any $I_0 \in \{\RC, \RA, \CC\}$. In this case, TruSt enumerates only $\SC$-consistent histories at the cost of solving an NP-complete problem at each step while the $\textsc{explore-ce}^*$ step cost is polynomial time at the price of not being strongly-optimal.}
Furthermore, we identify isolation levels ($\SI$ and $\SER$) for which it is impossible to ensure both strong optimality and polynomial space bounds with a swapping-based algorithm, a type of question that has not been investigated in previous work. 

\section{Conclusions}\label{sec:conc}

We presented efficient SMC algorithms based on DPOR for transactional programs running under standard isolation levels. These algorithms are instances of a generic schema, called swapping-based algorithms, which is parametrized by an isolation level. Our algorithms are sound and complete, and  polynomial space. Additionally, we identified a class of isolation levels, including $\RC$, $\RA$, and $\CC$, for which our algorithms are strongly optimal, and we showed that swapping-based algorithms cannot be strongly optimal for stronger levels $\SI$ and $\SER$ (but just optimal). For the isolation levels we considered, there is an intriguing coincidence between the existence of a strongly optimal swapping-based algorithm and the complexity of checking if a given history is consistent with that level. Indeed, checking consistency is polynomial time for $\RC$, $\RA$, and $\CC$, and NP-complete for $\SI$ and $\SER$. 
Investigating further the relationship between strong optimality and polynomial-time consistency checks is an interesting direction for future work. 
\section*{Acknowledgements}

We thank anonymous reviewers for their feedback, and Ayal Zaks for shepherding our paper. This work was partially supported by the project AdeCoDS of the French National Research Agency. 

\section*{Data availability statement}

The implementation is open-source and can be found in \cite{anonymous_2023_7824546}.


	
	
	\bibliography{main,dblp,acmart,misc}

	\newpage
	\appendix
	\counterwithin{figure}{section}
	\counterwithin{table}{section}

	\section{Axiomatic levels: Read Committed and Read Atomic.}
\label{sec:def:RA-RC}

\begin{figure}[H]
	\resizebox{\textwidth}{!}{
		\footnotesize
		\begin{tabular}{|c|c|}
			\hline &  \\
			
            \begin{subfigure}[t]{.4\textwidth}
                \centering
                \begin{tikzpicture}[->,>=stealth,shorten >=1pt,auto,node distance=1cm,
                    semithick, transform shape]
                    \node[transaction state, text=black] at (0,0)       	(t_1)           {$t_1$};
                    \node[transaction state, text=black, label={above:\textcolor{black}{$\writeVar{ }{x}$}}] at (-0.5,1.5) (t_2) {$t_2$};
                    \node[transaction state, text=black] at (2,0)       (o_1)           {$\alpha$};
                    \node[transaction state] at (1.5,1.5) (o_2) {$\beta$};
                    \path (t_1) edge[color=wrColor] node {$\wro_x$} (o_1);
                    \path (t_2) edge[color=wrColor] node {$\wro$} (o_2);
                    \path (o_2) edge[color=poColor] node {$\po$} (o_1);
                    \path (t_2) edge[left,double,color=coColor] node {$\co$} (t_1);
                \end{tikzpicture}
                \parbox{\textwidth}{
                    $\forall x,\ \forall t_1, t_2,\ \forall \alpha.\ t_1\neq t_2\ \land$
                    
                    \hspace{2mm}$\tup{t_1,\alpha}\in \wro_x \land \writeVar{t_2}{x}\ \land$ 
                    
                    \hspace{2mm}$\tup{t_2,\alpha}\in\wro\circ\po$
                    
                    \hspace{4mm}$\implies \tup{t_2,t_1}\in\co$
                }
                
                \caption{$\mathsf{Read\ Committed}$}
                \label{lock_rc_def}
            \end{subfigure}   
			
			&

			\begin{subfigure}[t]{.4\textwidth}
				\centering
				\begin{tikzpicture}[->,>=stealth,shorten >=1pt,auto,node distance=1cm,
					semithick, transform shape]
					\node[transaction state, text=black] at (0,0)       (t_1)           {$t_1$};
					\node[transaction state] at (2,0)       (t_3)           {$t_3$};
					\node[transaction state, text=black,label={above:\textcolor{black}{$\writeVar{ }{\xvar}$}}] at (-.5,1.5) (t_2) {$t_2$};
					\path (t_1) edge[wrColor] node {$\wro_x$} (t_3);
					\path (t_2) edge[bend left] node {$\so \cup \wro$} (t_3);
					\path (t_2) edge[left, double ,coColor] node {$\co$} (t_1);
				\end{tikzpicture}
				\parbox{\textwidth}{
					$\forall x,\ \forall t_1, t_2,\ \forall t_3.\ t_1\neq t_2\ \land$
					
					\hspace{2mm}$\tup{t_1,t_3}\in \wro_x \land \writeVar{t_2}{x}\ \land$ 
					
					\hspace{2mm}$\tup{t_2,t_3}\in\so \cup \wro$
					
					\hspace{4mm}$\implies \tup{t_2,t_1}\in\co$
				}
				
				\caption{$\mathsf{Read\ Atomic}$}
				\label{ra_def}
			\end{subfigure}

		\\\hline
		\end{tabular}
	}
	\caption{Axioms defining isolations levels. The reflexive and transitive, resp., transitive, closure of a relation $rel$ is denoted by $rel^*$, resp., $rel^+$. Also, $\circ$ denotes the composition of two relations, i.e., $rel_1 \circ rel_2 = \{\tup{a, b} | \exists c. \tup{a, c} \in rel_1 \land \tup{c, b} \in rel_2\}$.}
	\label{fig:consistency_defs_rc_ra}
\end{figure}

The axioms defined above in Figure~\ref{fig:consistency_defs_rc_ra} define the homonymous isolation levels \textit{Read Atomic} (also called Repeatable Read in the literature) and \textit{Read Committed}.

	\newpage
\section{Rules of the operational semantics (Section~\ref{ssec:semantics}).}
\label{app:rules}

\begin{figure} [H]
    \small
      \centering
      \begin{mathpar}
        \inferrule[spawn]{\tr \mbox{ fresh}\quad e \mbox{ fresh}\quad  \mathsf{P}(j) = \ibegin; \mathsf{Body}; \icommit; \mathsf{S} \quad \vec{\mathsf{B}}(j) = \epsilon}{
          \hist,\vec{\gamma},\vec{\mathsf{B}},\mathsf{P}
          \Rightarrow_I
          \hist \oplus_j \tup{\tr,\{\tup{e,\ebegin}\},\emptyset},\vec{\gamma}[j\mapsto \emptyset],\vec{\mathsf{B}}[j\mapsto \mathsf{Body}; \icommit],\mathsf{P}[j\mapsto \mathsf{S}]
        } 
    
        \inferrule[if-true]{\psi(\vec{x})[x\mapsto \vec{\gamma}(j)(x): x\in\vec{x}]\mbox{ true} \\
        \vec{\mathsf{B}}(j) = \iif{\psi(\vec{x})}{\mathsf{Instr}};\mathsf{B}
        }{
          \hist,\vec{\gamma},\vec{\mathsf{B}}, \mathsf{P}
          \Rightarrow_I
          \hist,\vec{\gamma},\vec{\mathsf{B}}[j\mapsto \mathsf{Instr};\mathsf{B}],\mathsf{P}
        } 
    
        \inferrule[if-false]{\psi(\vec{x})[x\mapsto \vec{\gamma}(j)(x): x\in\vec{x}]\mbox{ false} \\
        \vec{\mathsf{B}}(j) = \iif{\psi(\vec{x})}{\mathsf{Instr}};\mathsf{B}
        }{
          \hist,\vec{\gamma},\vec{\mathsf{B}}, \mathsf{P}
          \Rightarrow_I
          \hist,\vec{\gamma},\vec{\mathsf{B}}[j\mapsto \mathsf{B}],\mathsf{P}
        } 
    
        \inferrule[local]{v = \vec{\gamma}(j)(e) \\ \vec{\mathsf{B}}(j) = a := e;\mathsf{B}
        }{
          \hist,\vec{\gamma},\vec{\mathsf{B}}, \mathsf{P}
          \Rightarrow_I
          \hist,\vec{\gamma}[(j,\xvar)\mapsto \val],\vec{\mathsf{B}}[j\mapsto \mathsf{B}],\mathsf{P}
        } 
    
        \inferrule[write]{v = \vec{\gamma}(j)(x)\quad e\mbox{ fresh} \quad 
        \vec{\mathsf{B}}(j) = \iwrite(\key,\xvar);\mathsf{B}\quad \hist \oplus_j \tup{e,\ewrt{\key,\val}}\mbox{ satisfies $I$}
        }{
          \hist,\vec{\gamma},\vec{\mathsf{B}}, \mathsf{P}
          \Rightarrow_I
          \hist \oplus_j \tup{e,\ewrt{\key,\val}},\vec{\gamma},\vec{\mathsf{B}}[j\mapsto \mathsf{B}], \mathsf{P}
        } 
    
        \inferrule[read-local]{
        \writeOp{\mathit{last}(\hist,j)}\mbox{ contains a }\wrt{\key}{\val}\mbox{ event}\\
        e\mbox{ fresh } \\
        \vec{\mathsf{B}}(j) = \xvar := \iread(\key);\mathsf{B}
        }{
          \hist,\vec{\gamma},\vec{\mathsf{B}}, \mathsf{P}
          \Rightarrow_I
          \hist \oplus_j \tup{e,\erd{\key}},\vec{\gamma}[(j,\xvar)\mapsto \val],\vec{\mathsf{B}}[j\mapsto \mathsf{B}],\mathsf{P}
        } 
    
        \inferrule[read-extern]{
        \writeOp{\mathit{last}(\hist,j)}\mbox{ does not contain a }\wrt{\key}{\val}\mbox{ event} \\
        e\mbox{ fresh }\\
        \vec{\mathsf{B}}(j) = \xvar := \iread(\key);\mathsf{B} \\
        \hist=(T,\so,\wro) \\
        \tr =\mathit{last}(\hist,j) \\
        \wrt{\key}{\val}\in\writeOp{\tr'}\mbox{ with $\tr'\in \transC{\hist}$ and $\tr\neq \tr'$} \\
        \hist' = (\hist \oplus_j \tup{e,\erd{\key}}) \oplus \wro(\tr',e) \\
        \hist' \mbox{ satisfies }I }{
          \hist,\vec{\gamma},\vec{\mathsf{B}}, \mathsf{P}
          \Rightarrow_I
          \hist',\vec{\gamma}[(j,\xvar)\mapsto \val],\vec{\mathsf{B}}[j\mapsto \mathsf{B}],\mathsf{P}
        } 
    
        \inferrule[commit]{e\mbox{ fresh} \quad 
        \vec{\mathsf{B}}(j) = \icommit 
        }{
          \hist,\vec{\gamma},\vec{\mathsf{B}}, \mathsf{P}
          \Rightarrow_I
          \hist \oplus_j \tup{e,\icommit},\vec{\gamma},\vec{\mathsf{B}}[j\mapsto \epsilon], \mathsf{P}
        } 
    
        \inferrule[abort]{e\mbox{ fresh} \quad 
        \vec{\mathsf{B}}(j) = \iabort; B
        }{
          \hist,\vec{\gamma},\vec{\mathsf{B}}, \mathsf{P}
          \Rightarrow_I
          \hist \oplus_j \tup{e,\iabort},\vec{\gamma},\vec{\mathsf{B}}[j\mapsto \epsilon], \mathsf{P}
        }
      \end{mathpar}
      \caption{An operational semantics for transactional programs. Above, $\mathit{last}(h,j)$ denotes the last transaction log in the session order $\so(j)$ of $h$, and $\transC{\hist}$ denotes the set of transaction logs in $\hist$ that are committed}.
      \label{fig:op:sem:baseline}
    \end{figure}
    
Figure~\ref{fig:op:sem:baseline} uses the following notation. Let $\hist$ be a history that contains a representation of $\so$ as above. We use $\hist\oplus_j \tup{\tr,E,\po_t}$ to denote a history where $\tup{\tr,E,\po_t}$ is appended to $\so(j)$. 
Also, for an event $e$, $\hist\oplus_j e$ is the history obtained from $\hist$ by adding $e$ to the last transaction log in $\so(j)$ and as a last event in the program order of this log (i.e.,  if $\so(j)=\sigma; \tup{t,E,\po_t}$, then the session order $\so'$ of $\hist\oplus_j e$ is defined by $\so'(k)=\so(k)$ for all $k\neq j$ and $\so(j) =\sigma; \tup{t,E\cup\{e\},\po_t\cup \{(e',e): e'\in E\}}$). Finally, for a history $\hist = \tup{T, \so, \wro}$, $\hist\oplus\wro(\tr,e)$ is the history obtained from $\hist$ by adding $(\tr,e)$ to the write-read relation.

\textsc{spawn} starts a new transaction in a session $j$ provided that this session has no live transaction ($\vec{\mathsf{B}}(j) = \epsilon$). It adds a transaction log with a single $\ebegin$ event to the history and schedules the body of the transaction. \textsc{if-true} and \textsc{if-false} check the truth value of a Boolean condition of an $\mathtt{if}$ conditional. \textsc{local} models the execution of an assignment to a local variable which does not impact the stored history. \textsc{read-local} and \textsc{read-extern} concern read instructions. \textsc{read-local} handles the case where the read follows a write on the variable $x$ in the same transaction: the read returns the value written by the last write on $x$ in that transaction. Otherwise, \textsc{read-extern} corresponds to reading a value written in another transaction $\tr'$. The transaction $\tr'$ is chosen non-deterministically as long as extending the current history with the write-read dependency associated to this choice leads to a history that still satisfies $I$. \textsc{read-extern} applies only when the executing transaction contains no write on the same variable. \textsc{commit} confirms the end of a transaction making its writes visible while \textsc{abort} ends the transaction's execution immediately.

	\newpage

\section{Proof of Theorem~\ref{theorem:causalExtensibleModels-CC-RA-RC}}
\label{app:causalExtensible}

\causalExtensibleModels*

\begin{proof}
Let $I$ be an isolation level in $\{\CC, \RA, \RC\}$. We show that any commit order $\co$ justifying that a history $h$ is $I$-consistent can also be used to justify that a causal extension $\hist'$ of a $(\so \cup \wro)^*$-maximal pending transaction $t$ in $h$ with an event $e$ is $I$-consistent as well.
We consider a causal extension $\hist'$ where if $e$ is a read event, then it reads from the last transaction $t_w$ in $\co$ such that $t_w$ writes $\mathit{var}(e)$ and $(t_w, t) \in (\so \cup \wro)^+$.
Assume by contradiction that this is not the case. Let $\phi_{\CC}(h', t', e') = t' \ (\so \cup \wro)^+ \ \trans{h'}{e'}$, $\phi_{\RA}(h',t', e') = t' \ (\so \cup \wro) \ \trans{h'}{e'}$ and $\phi_{\RC}(h',t', e') = t' \ (\wro \circ \po) \ e'$ be sub-formulas of the axioms defining the corresponding isolation level. Then, 
%
%
$\hist'$ contains transactions $t_1, t_2, t_3$  such that $t_2$ writes some variable $x$, $t_3$ contains some read event $e'$, $(t_1, e') \in \wro_x$ and $\phi_{I}(h', t_2, e')$ but $(t_1, t_2) \in \co$. The assumption concerning $\co$ implies that the extended transaction $t$ is one of $t_1, t_2, t_3$ (otherwise, $\co$ would not be a ``valid'' commit order for $\hist$). Since $t$ is $(\so \cup \wro)^+$-maximal in $\hist$, we have that $t\not\in \{t_1,t_2\}$. If 
$e$ is \emph{not} a read event, or if $e$ is a read event different from $e'$, then $t \neq t_3$, as $t_1$, $t_2$ and $t_3$ would satisfy the same constraints in $h$, which is impossible by the hypothesis. Otherwise, if $e=e'$, then this contradicts the choice we made for the transaction $t_w$ that $e$ reads from. Since $(t_1, t_2) \in \co$ and $t_2$ writes $\mathit{var}(e)$, it means that $t_w=t_1$ is not maximal w.r.t. $\co$ among transactions that write $\mathit{var}(e)$ and precede $t$ in $(\so \cup \wro)^+$. Both cases lead to a contradiction, which implies that $\hist'$ is $I$-consistent, and therefore the theorem holds.
\end{proof}
\newpage
\section{Proof of Theorem~\ref{th:impos}}
\label{app:impossibility}

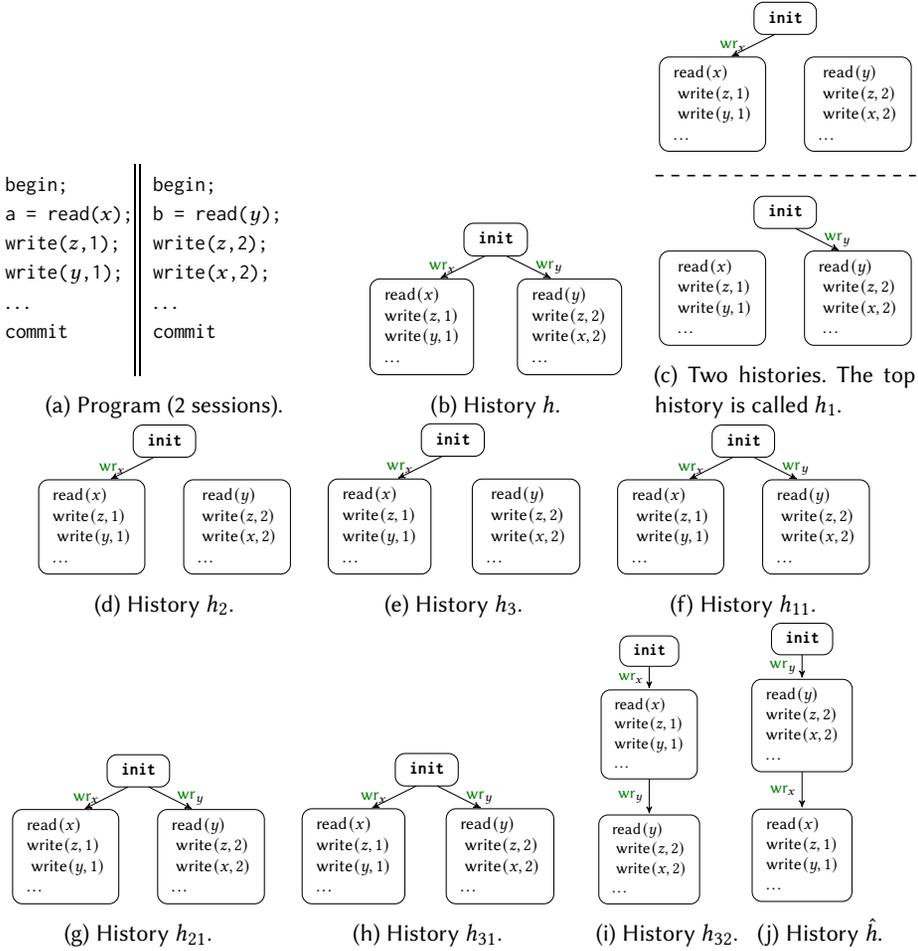
\begin{figure}[H]
	\centering
	\begin{subfigure}[b]{.33\textwidth}

		\begin{adjustbox}{max width=\textwidth}
			\begin{tabular}{c||c}
				\begin{lstlisting}[xleftmargin=5mm,basicstyle=\ttfamily\footnotesize,escapeinside={(*}{*)}, tabsize=1]
begin;
a = read((*$x$*));
write((*$z$*),1);
write((*$y$*),1);
(*$\ldots$*)
commit
				\end{lstlisting} &
				\begin{lstlisting}[xleftmargin=5mm,basicstyle=\ttfamily\footnotesize,escapeinside={(*}{*)}, tabsize=1]
begin;
b = read((*$y$*));
write((*$z$*),2);
write((*$x$*),2);
(*$\ldots$*)
commit
				\end{lstlisting} 
	
			\end{tabular} 
		\end{adjustbox}
%
%

		\caption{Program \nver{(2 sessions)}.}
		\label{fig:ser-non-optimal:prog}
	\end{subfigure}
	\hspace{.2cm}
	\centering
	\begin{subfigure}[b]{.25\textwidth}
		\resizebox{\textwidth}{!}{
			\begin{tikzpicture}[->,>=stealth',shorten >=1pt,auto,node distance=3cm,
				semithick, transform shape]
				\node[draw, rounded corners=2mm,outer sep=0] (t1) at (-1.5, -0.25) {\begin{tabular}{l} $\init$ \end{tabular}};
				\node[draw, rounded corners=2mm,outer sep=0] (t2) at (-3, -2) {\begin{tabular}{l} 
					$\rd{x}$ \\$\wrt{z}{1}$ \\  $\wrt{y}{1}$ \\ \ldots
				\end{tabular}};
				\node[draw, rounded corners=2mm,outer sep=0] (t3) at (0, -2) {\begin{tabular}{l} 
					$\rd{y}$ \\ $\wrt{z}{2}$ \\ $\wrt{x}{2}$ \\ \ldots
				\end{tabular}};		
				
				\path (t1.south west) -- (t1.south) coordinate[pos=0.67] (t1sw);
				\path (t1.south east) -- (t1.south) coordinate[pos=0.67] (t1se);
				\path (t2.north east) -- (t2.north) coordinate[pos=0.67] (t2x);
				\path (t3.north west) -- (t3.north) coordinate[pos=0.67] (t3x);
				
				\path (t1sw) edge [left] node {$\wro_x$} (t2x);
				\path (t1se) edge [right] node {$\wro_y$} (t3x);
			\end{tikzpicture}  
			
		}
		\caption{History $\hist$.}
		\label{fig:ser-non-optimal:h}
	\end{subfigure}
	\hspace{.2cm}
	\centering
	\begin{subfigure}[b]{.25\textwidth}
		\resizebox{\textwidth}{!}{
			\begin{tikzpicture}[->,>=stealth',shorten >=1pt,auto,node distance=3cm,
				semithick, transform shape]
				\node[draw, rounded corners=2mm,outer sep=0] (t1) at (-1.5, -0.25) {\begin{tabular}{l} $\init$ \end{tabular}};
				\node[draw, rounded corners=2mm,outer sep=0] (t2) at (-3, -2) {\begin{tabular}{l} 
					$\rd{x}$ \\  \pgfsetfillopacity{0.3}$\wrt{z}{1}$ \\ \pgfsetfillopacity{0.3}$\wrt{y}{1}$  \\ \ldots
				\end{tabular}};
				\node[draw, rounded corners=2mm,outer sep=0] (t3) at (0, -2) {\begin{tabular}{l} 
					\pgfsetfillopacity{0.3}$\rd{y}$ \\ 
					\pgfsetfillopacity{0.3}$\wrt{z}{2}$ \\ \pgfsetfillopacity{0.3}$\wrt{x}{2}$ \\ \ldots
				\end{tabular}};		
				
				\path (t1.south west) -- (t1.south) coordinate[pos=0.67] (t1sw);
				\path (t1.south east) -- (t1.south) coordinate[pos=0.67] (t1se);
				\path (t2.north east) -- (t2.north) coordinate[pos=0.67] (t2x);
				\path (t3.north west) -- (t3.north) coordinate[pos=0.67] (t3x);
				
				\path (t1sw) edge [left] node {$\wro_x$} (t2x);
			\end{tikzpicture}  
		}

			\vspace{-1mm}
			\begin{tikzpicture}[shorten >=1pt,auto,node distance=3cm,
				semithick, transform shape]
				\path (-1.5,0) edge[dashed] (2,0);
			\end{tikzpicture}
			\vspace{-2mm}

		\resizebox{\textwidth}{!}{
			\begin{tikzpicture}[->,>=stealth',shorten >=1pt,auto,node distance=3cm,
				semithick, transform shape]
				\node[draw, rounded corners=2mm,outer sep=0] (t1) at (-1.5, -0.25) {\begin{tabular}{l} $\init$ \end{tabular}};
				\node[draw, rounded corners=2mm,outer sep=0] (t2) at (-3, -2) {\begin{tabular}{l} 
					\pgfsetfillopacity{0.3}$\rd{x}$ \\ 
					\pgfsetfillopacity{0.3}$\wrt{z}{1}$ \\ \pgfsetfillopacity{0.3}$\wrt{y}{1}$  \\ \pgfsetfillopacity{0.3}\ldots
				\end{tabular}};
				\node[draw, rounded corners=2mm,outer sep=0] (t3) at (0, -2) {\begin{tabular}{l} 
					$\rd{y}$ \\  \pgfsetfillopacity{0.3}$\wrt{z}{2}$ \\ \pgfsetfillopacity{0.3}$\wrt{x}{2}$ \\ \pgfsetfillopacity{0.3}\ldots
				\end{tabular}};		
				
				\path (t1.south west) -- (t1.south) coordinate[pos=0.67] (t1sw);
				\path (t1.south east) -- (t1.south) coordinate[pos=0.67] (t1se);
				\path (t2.north east) -- (t2.north) coordinate[pos=0.67] (t2x);
				\path (t3.north west) -- (t3.north) coordinate[pos=0.67] (t3x);
				
				\path (t1se) edge [right] node {$\wro_y$} (t3x);
			\end{tikzpicture}  
		}
		\caption{Two histories. The top history is called $\hist_1$.}
		\label{fig:ser-non-optimal:1}
	\end{subfigure}
	\centering
	\begin{subfigure}[b]{.25\textwidth}
	\resizebox{\textwidth}{!}{
		\begin{tikzpicture}[->,>=stealth',shorten >=1pt,auto,node distance=3cm,
			semithick, transform shape]
			\node[draw, rounded corners=2mm,outer sep=0] (t1) at (-1.5, -0.25) {\begin{tabular}{l} $\init$ \end{tabular}};
			\node[draw, rounded corners=2mm,outer sep=0] (t2) at (-3, -2) {\begin{tabular}{l} 
				$\rd{x}$ \\ $\wrt{z}{1}$ \\ \pgfsetfillopacity{0.3}$\wrt{y}{1}$ \\ \ldots
			\end{tabular}};
			\node[draw, rounded corners=2mm,outer sep=0] (t3) at (0, -2) {\begin{tabular}{l} 
				\pgfsetfillopacity{0.3}$\rd{y}$ \\ \pgfsetfillopacity{0.3}$\wrt{z}{2}$ \\ \pgfsetfillopacity{0.3}$\wrt{x}{2}$ \\ \ldots
			\end{tabular}};		
			
			\path (t1.south west) -- (t1.south) coordinate[pos=0.67] (t1sw);
			\path (t1.south east) -- (t1.south) coordinate[pos=0.67] (t1se);
			\path (t2.north east) -- (t2.north) coordinate[pos=0.67] (t2x);
			\path (t3.north west) -- (t3.north) coordinate[pos=0.67] (t3x);
			
			\path (t1sw) edge [left] node {$\wro_x$} (t2x);
		\end{tikzpicture}  			
	}
	\caption{History $\hist_{2}$.}
	\label{fig:ser-non-optimal:2}
	\end{subfigure}
	\hspace{.2cm}
	\centering
	\begin{subfigure}[b]{.25\textwidth}
		\resizebox{\textwidth}{!}{
			\begin{tikzpicture}[->,>=stealth',shorten >=1pt,auto,node distance=3cm,
				semithick, transform shape]
				\node[draw, rounded corners=2mm,outer sep=0] (t1) at (-1.5, -0.25) {\begin{tabular}{l} $\init$ \end{tabular}};
				\node[draw, rounded corners=2mm,outer sep=0] (t2) at (-3, -2) {\begin{tabular}{l} 
						$\rd{x}$ \\ $\wrt{z}{1}$ \\ $\wrt{y}{1}$ \\ \ldots
				\end{tabular}};
				\node[draw, rounded corners=2mm,outer sep=0] (t3) at (0, -2) {\begin{tabular}{l} 
					\pgfsetfillopacity{0.3}$\rd{y}$ \\ \pgfsetfillopacity{0.3}$\wrt{z}{2}$ \\  \pgfsetfillopacity{0.3}$\wrt{x}{2}$ \\ \ldots
				\end{tabular}};		
				
				\path (t1.south west) -- (t1.south) coordinate[pos=0.67] (t1sw);
				\path (t1.south east) -- (t1.south) coordinate[pos=0.67] (t1se);
				\path (t2.north east) -- (t2.north) coordinate[pos=0.67] (t2x);
				\path (t3.north west) -- (t3.north) coordinate[pos=0.67] (t3x);
				
				\path (t1sw) edge [left] node {$\wro_x$} (t2x);
			\end{tikzpicture}  
		}
		\caption{History $\hist_{3}$.}
		\label{fig:ser-non-optimal:3}
	\end{subfigure}
	\hspace{.2cm}
	\begin{subfigure}[b]{.25\textwidth}
		\resizebox{\textwidth}{!}{
			\begin{tikzpicture}[->,>=stealth',shorten >=1pt,auto,node distance=3cm,
				semithick, transform shape]
				\node[draw, rounded corners=2mm,outer sep=0] (t1) at (-1.5, -0.25) {\begin{tabular}{l} $\init$ \end{tabular}};
				\node[draw, rounded corners=2mm,outer sep=0] (t2) at (-3, -2) {\begin{tabular}{l} 
					$\rd{x}$ \\ \pgfsetfillopacity{0.3}$\wrt{z}{1}$ \\  \pgfsetfillopacity{0.3}$\wrt{y}{1}$ \\ \ldots
				\end{tabular}};
				\node[draw, rounded corners=2mm,outer sep=0] (t3) at (0, -2) {\begin{tabular}{l} 
					$\rd{y}$ \\ \pgfsetfillopacity{0.3}$\wrt{z}{2}$ \\ \pgfsetfillopacity{0.3}$\wrt{x}{2}$ \\ \ldots
				\end{tabular}};		
				
				\path (t1.south west) -- (t1.south) coordinate[pos=0.67] (t1sw);
				\path (t1.south east) -- (t1.south) coordinate[pos=0.67] (t1se);
				\path (t2.north east) -- (t2.north) coordinate[pos=0.67] (t2x);
				\path (t3.north west) -- (t3.north) coordinate[pos=0.67] (t3x);
				
				\path (t1sw) edge [left] node {$\wro_x$} (t2x);
				\path (t1se) edge [right] node {$\wro_y$} (t3x);
			\end{tikzpicture}  
			
		}
		\caption{History $\hist_{11}$.}
		\label{fig:ser-non-optimal:11}
	\end{subfigure}
	\centering
	\begin{subfigure}[b]{.25\textwidth}
	\resizebox{\textwidth}{!}{
		\begin{tikzpicture}[->,>=stealth',shorten >=1pt,auto,node distance=3cm,
			semithick, transform shape]
			\node[draw, rounded corners=2mm,outer sep=0] (t1) at (-1.5, -0.25) {\begin{tabular}{l} $\init$ \end{tabular}};
			\node[draw, rounded corners=2mm,outer sep=0] (t2) at (-3, -2) {\begin{tabular}{l} 
				$\rd{x}$ \\ $\wrt{z}{1}$ \\ \pgfsetfillopacity{0.3}$\wrt{y}{1}$ \\ \ldots
			\end{tabular}};
			\node[draw, rounded corners=2mm,outer sep=0] (t3) at (0, -2) {\begin{tabular}{l} 
				$\rd{y}$ \\ \pgfsetfillopacity{0.3}$\wrt{z}{2}$ \\ \pgfsetfillopacity{0.3}$\wrt{x}{2}$ \\ \ldots
			\end{tabular}};		
			
			\path (t1.south west) -- (t1.south) coordinate[pos=0.67] (t1sw);
			\path (t1.south east) -- (t1.south) coordinate[pos=0.67] (t1se);
			\path (t2.north east) -- (t2.north) coordinate[pos=0.67] (t2x);
			\path (t3.north west) -- (t3.north) coordinate[pos=0.67] (t3x);
			
			\path (t1sw) edge [left] node {$\wro_x$} (t2x);
			\path (t1se) edge [right] node {$\wro_y$} (t3x);
		\end{tikzpicture}  			
	}
	\caption{History $\hist_{21}$.}
	\label{fig:ser-non-optimal:21}
	\end{subfigure}
	\hspace{.2cm}
	\centering
	\begin{subfigure}[b]{.25\textwidth}
		\resizebox{\textwidth}{!}{
			\begin{tikzpicture}[->,>=stealth',shorten >=1pt,auto,node distance=3cm,
				semithick, transform shape]
				\node[draw, rounded corners=2mm,outer sep=0] (t1) at (-1.5, -0.25) {\begin{tabular}{l} $\init$ \end{tabular}};
				\node[draw, rounded corners=2mm,outer sep=0] (t2) at (-3, -2) {\begin{tabular}{l} 
						$\rd{x}$ \\ $\wrt{z}{1}$ \\ $\wrt{y}{1}$ \\ \ldots
				\end{tabular}};
				\node[draw, rounded corners=2mm,outer sep=0] (t3) at (0, -2) {\begin{tabular}{l} 
						$\rd{y}$ \\ \pgfsetfillopacity{0.3}$\wrt{z}{2}$ \\  \pgfsetfillopacity{0.3}$\wrt{x}{2}$ \\ \ldots
				\end{tabular}};		
				
				\path (t1.south west) -- (t1.south) coordinate[pos=0.67] (t1sw);
				\path (t1.south east) -- (t1.south) coordinate[pos=0.67] (t1se);
				\path (t2.north east) -- (t2.north) coordinate[pos=0.67] (t2x);
				\path (t3.north west) -- (t3.north) coordinate[pos=0.67] (t3x);
				
				\path (t1sw) edge [left] node {$\wro_x$} (t2x);
				\path (t1se) edge [right] node {$\wro_y$} (t3x);
			\end{tikzpicture}  
		}
		\caption{History $\hist_{31}$.}
		\label{fig:ser-non-optimal:31}
	\end{subfigure}
	\hspace{.3cm}
	\centering
	\begin{subfigure}[b]{.14\textwidth}
	\resizebox{.75\textwidth}{!}{
		\begin{tikzpicture}[->,>=stealth',shorten >=1pt,auto,node distance=3cm,
			semithick, transform shape]
			\node[draw, rounded corners=2mm,outer sep=0] (t1) at (0, 0) {\begin{tabular}{l} $\init$ \end{tabular}};
			\node[draw, rounded corners=2mm,outer sep=0] (t2) at (0, -1.8) {\begin{tabular}{l} 
					$\rd{x}$ \\ $\wrt{z}{1}$ \\ $\wrt{y}{1}$ \\ \ldots
			\end{tabular}};
			\node[draw, rounded corners=2mm,outer sep=0] (t3) at (0, -4.5) {\begin{tabular}{l} 
					$\rd{y}$ \\ \pgfsetfillopacity{0.3}$\wrt{z}{2}$ \\ \pgfsetfillopacity{0.3}$\wrt{x}{2}$ \\ \ldots
			\end{tabular}};		
			
			\path (t1.south west) -- (t1.south) coordinate[pos=0.67] (t1sw);
			\path (t1.south east) -- (t1.south) coordinate[pos=0.67] (t1se);
			\path (t2.north east) -- (t2.north) coordinate[pos=0.67] (t2x);
			\path (t3.north west) -- (t3.north) coordinate[pos=0.67] (t3x);
			
			\path (t1.south) edge [left] node {$\wro_x$} (t2.north);
			\path (t2.south) edge [left] node {$\wro_y$} (t3.north);
		\end{tikzpicture}  
			
	}
	\caption{History $\hist_{32}$.}
	\label{fig:ser-non-optimal:32}
	\end{subfigure}
	\begin{subfigure}[b]{.14\textwidth}
		\resizebox{.75\textwidth}{!}{
			\begin{tikzpicture}[->,>=stealth',shorten >=1pt,auto,node distance=3cm,
				semithick, transform shape]
				\node[draw, rounded corners=2mm,outer sep=0] (t1) at (0, 0) {\begin{tabular}{l} $\init$ \end{tabular}};
				\node[draw, rounded corners=2mm,outer sep=0] (t3) at (0, -1.8) {\begin{tabular}{l} 
					$\rd{y}$ \\ $\wrt{z}{2}$ \\ $\wrt{x}{2}$ \\ \ldots
				\end{tabular}};	
				\node[draw, rounded corners=2mm,outer sep=0] (t2) at (0, -4.5) {\begin{tabular}{l} 
					$\rd{x}$ \\ $\wrt{z}{1}$ \\ $\wrt{y}{1}$ \\ \ldots
				\end{tabular}};

				\path (t1.south west) -- (t1.south) coordinate[pos=0.67] (t1sw);
				\path (t1.south east) -- (t1.south) coordinate[pos=0.67] (t1se);
				\path (t2.north east) -- (t2.north) coordinate[pos=0.67] (t2x);
				\path (t3.north west) -- (t3.north) coordinate[pos=0.67] (t3x);
				
				\path (t1.south) edge [left] node {$\wro_y$} (t3.north);
				\path (t3.south) edge [left] node {$\wro_x$} (t2.north);
			\end{tikzpicture}  
				
		}
		\caption{History $\hat{\hist}$.}
		\label{fig:ser-non-optimal:32}
		\end{subfigure}
\vspace{-2mm}
	\caption{A program and some partial histories. Events in grey are not yet added to the history. For $\hist_3$, $\hist_{31}$ and $\hist_{32}$, the number of events that follow $\ewrt{y,1}$ and $\ewrt{x,2}$ is not important (we use black $\ldots$  to signify that).}
	\label{fig:ser-non-optimal}
\end{figure}

\vspace{-0.5cm}

\begin{figure}[H]
	\resizebox{.5\textwidth}{!}{
		\begin{tikzpicture}[->,>=stealth',shorten >=1pt,auto,node distance=3cm,
			semithick, transform shape]
			\node[draw, rounded corners=2mm,outer sep=0] (h0) at (-3, -1.5) {$\emptyset$};
			\node[draw, rounded corners=2mm,outer sep=0] (h1) at (-1.5, -1.5) { $h_1$};
			\node[draw, rounded corners=2mm,outer sep=0] (h2) at (0, -1.5) {$h_2$};
			\node[draw, rounded corners=2mm,outer sep=0] (h3) at (1.5, -1.5) {$h_3$};
			\node[draw, rounded corners=2mm,outer sep=0] (h11) at (-1.5, -3) {$h_{11}$};
			\node[draw, rounded corners=2mm,outer sep=0] (h21) at (0, -3) { $h_{21}$};
			\node[draw, rounded corners=2mm,outer sep=0] (h31) at (1.5, -3) {$h_{31}$};
			\node[draw, rounded corners=2mm,outer sep=0] (h32) at (3, -3) {$h_{32}$};	
			\node[draw, red, rounded corners=2mm,outer sep=0] (h) at (0, -4.5) {$h$};
			\node[draw, rounded corners=2mm,outer sep=0] (h') at (3, -4.5) {$\hat{h}$};	
			\node [red] (x) at (3, -3.65) {$\bigtimes$};

			\path (h0.east) edge [below] node[right] {} (h1.west);
			\path (h1.east) edge [below] node[right] {} (h2.west);			
			\path (h2.east) edge [below] node[right] {} (h3.west);		
			\path (h1.south) edge [below] node[right] {} (h11.north);				
			\path (h2.south) edge [below] node[right] {} (h21.north);			
			\path (h3.south) edge [below] node[right] {} (h31.north);				
			\path (h3.south) edge [below] node[right] {} (h32.north);			
			\path (h11.south) edge [below, dashed] node[right] {} (h.north);			
			\path (h21.south) edge [below, dashed] node[right] {} (h.north);			
			\path (h31.south) edge [below, dashed] node[right] {} (h.north);
			\path (h32.south) edge [below, dashed] node[right] {} (h'.north);			
			
			
		\end{tikzpicture}
		
	}
	\caption{Summary of all possible execution paths from $\textsc{explore}$. Black arrows represent alternative explored options depending on $\genericNext$ while dashed arrows are mandatory visited histories from such state.}
	\label{fig:execution-ser-si-impossibility}
\end{figure}
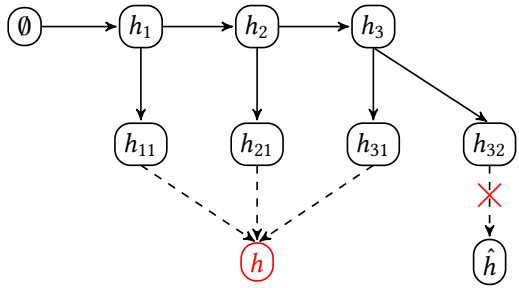

\impossibility*

\begin{proof}
\label{proof:impossibility}
	We consider the program in Figure~\ref{fig:ser-non-optimal:prog}, and show that any concrete instance of the $\textsc{explore}$ function in Algorithm~\ref{algorithm:algo-class} \emph{can not be both} $I$-complete and strongly optimal. This program contains two transactions, where only the first three instructions in each transaction are important.
	We show that if $\textsc{explore}$ is $I$-complete, then it will necessarily be called recursively on a history $\hist$ like in Figure~\ref{fig:ser-non-optimal:h} which does not satisfy $I$, thereby violating strong optimality. In the history $\hist$, both \textit{Snapshot Isolation} and \textit{Serializability} forbid the two reads reading initial values while the writes following them are also executed (committed). A diagram of the proof can be seen in Figure~\ref{fig:execution-ser-si-impossibility}.

	Assuming that the function $\genericNext$ is not itself blocking (which would violate strong optimality), \nver{the $\textsc{explore}$ will be called recursively on \emph{exactly one} of the two histories in Figure~\ref{fig:ser-non-optimal:1}, depending on which of the two reads is returned first by $\genericNext$. 
	We will continue our discussion with the history $\hist_1$ on the top of Figure~\ref{fig:ser-non-optimal:1}. The other case is similar (symmetric).}
	
	\nver{From $\hist_1$, depending on order defined by $\genericNext$ between $\ewrt{z,1}$ and $\erd{y}$, $\textsc{explore}$ can be called recursively either on $\hist_{11}$ in Figure~\ref{fig:ser-non-optimal:11} or on $\hist_{2}$ in Figure~\ref{fig:ser-non-optimal:2}. Analogously, from $\hist_2$ two alternatives arise depending on the order defined by $\genericNext$ between $\erd{y}$ and the rest of events in the left transaction: exploring $\hist_{21}$ in Figure~\ref{fig:ser-non-optimal:21} if $\erd{y}$ is added before $\ewrt{y,1}$ or $\hist_3$ in Figure~\ref{fig:ser-non-optimal:3} otherwise. Thus, from $\hist_3$ two alternatives arise when added $\erd{y}$ depending on where it reads from: $\hist_{31}$ in Figure~\ref{fig:ser-non-optimal:31} if it reads from $\init$ and $\hist_{32}$ in~Figure~\ref{fig:ser-non-optimal:32} if it reads from the left transaction.
	}



	\nver{However, from histories $\hist_{11}$, $\hist_{21}$ or $\hist_{31}$ $\textsc{explore}$ will necessarily be called recursively on a history $\hist$ like in Figure~\ref{fig:ser-non-optimal:h} which does not satisfy $I$, thereby violating strong optimality: $\textsc{explore}$ always explore branches that enlarge the current history. Thus, any $\textsc{explore}$ implementation that is strong optimal should only explore $\hist_{32}$. In such case, by the restrictions on the $\genericSwap$ function (defined in Section~\ref{sec:algs}), any extension of $\hist_{32}$ does not allow to explore the history $\hat{h}$ in Figure~\ref{fig:ser-non-optimal:3} where $\erd{x}$ reads from $\ewrt{x,2}$: any outcome of a re-ordering between two contiguous subsequences $\alpha$ and $\beta$ must be prefix of such extension when the events in $\alpha$ are taken out. In particular, for any extension $\hist'$ of $\hist_{32}$ and pair of contiguous sequences $\alpha, \beta$ such that $\hist' \setminus \alpha$ is a prefix of $\hist'$, if an event from the second transaction belongs to $\beta$, $\erd{y}$ must also be in $\beta$. Therefore, $\ewrt{x, 2}$ must be in $\beta$ as it is $\wro^{-1}(\erd{y})$. Hence, $\erd{x}$ must also be in $\beta$. Analogously, if $\erd{x}$ belongs to $\beta$, $\init$ belongs to it. Altogether, if $\beta$ contains any element, then $\alpha$ must be empty; so no swaps can be produced from $h_{32}$. To conclude, in this case $\textsc{explore}$ violates $I$-completeness.}

\end{proof}

	\newpage
\section{Proof of Theorem~\ref{th:corr}}
\label{sec:proofs-algorithm}

\algoCorrect*

As explained in Section~\ref{ssec:corr}, $I$-soundness, the polynomial space bound, and the part of strong completeness that refers to not engaging in fruitless explorations follow directly from definitions. In the following, we focus on $I$-completeness and then optimality. For the sake of the proof's readability, we will omit all local states of the algorithm's definition during the proof. Therefore, we consider programs where we can describe all their events.

%
%

\subsection{Completeness}

By definition, $\textsc{explore-ce}$ is $I$-complete if for any given program $\prog$, it outputs every history in $\histOf[I]{\prog}$. Let $h \in \histOf[I]{\prog}$. Our objective is to produce a computable path of ordered histories that lead to $h$ (i.e. a (finite) ordered collection of ordered histories such that $h_0 = \emptyset$ and for every $n$, if $e= \nextEvent(h_n)$, either $h_{n+1} = h_n \oplus e$, $h_{n+1} = h_n \oplus \wro(e, t)$ for some $t \in h_n$ or $h_{n+1} = \swap(h_n, r, t)$ for some $r, t \in h_n$).

However, the algorithm $\textsc{explore-ce}$ works with ordered histories. Therefore, we first have to furnish $h$ with a total order called \textit{canonical order} that, if $h$ were reachable, it would coincide with its history order.  Secondly, we describe a function $\prev$ defined over the set of all partial histories that, if $h$ is reachable, $\prev(h)$ returns the previous history of $h$ computed by $\textsc{explore-ce}$. Then, we prove that there exists a finite collection of histories $H = \{h_i\}_{i = 0}^n$ such that $h_n = h$, $h_0 = \emptyset$ and $h_i = \prev(h_{i+1})$. As it ends in the initial state, we can therefore prove that this collection conforms an actual computable path; which allow us to conclude that $h$ is reachable. Nevertheless, for proving both the equivalence between history order and canonical order and the soundness of function $\prev$ we will define the notion of \textit{$\ora$-respectfulness}, an invariant satisfied by every reachable history based on the events' relative positions in the oracle order.


\subsubsection{Canonical order} $ $\\

As mentioned, we need to formally define a total order for every history that coincide on reachable histories with the history order. For achieving it, we analyze how the algorithm orders transaction logs in a history. In particular, we observe that if two transactions $t, t'$ have a $(\so \cup \wro)^*$ dependency, the history order in the algorithm orders them analogously. But if they are $(\so \cup \wro)^*$-incomparable, the algorithm prioritizes the one that is read by a smaller $\iread$ event according $\ora$. Combining both arguments recursively we obtain a \textit{canonical order} for a history, which is formally defined with the function presented below. 

\begin{algorithm}[h]
	\caption{\textsc{Canonical order}}
	\begin{algorithmic}[1]
		
		\Statex
		\Procedure{\textsc{canonicalOrder}}{$h, t, t'$}
		\State \Return $t \ [\so \cup \wro]^* \ t' \lor$
		\State \qquad $(\lnot(t' \ [\so \cup \wro]^* \ t) \land \minimalDependency(h, t, t', \bot)$
		\EndProcedure
		
		\Procedure{\minimalDependency}{$h, t, t', e$}
		\Let $a = \min_{<_{\ora}} \dep(h, t, e)$; $a' = \min_{<_{\ora}} \dep(h, t', e)$
		\If{$a \neq a'$}
		\State \Return $a <_{\ora} a'$
		\Else
		\State \Return $\minimalDependency(h, t, t', a)$
		\EndIf
		\EndProcedure
		
		\Procedure{\dep}{$h, t, e$}
		\State \Return $\{r \ | \exists t' \text{ s.t. } t \ [\so \cup \wro]^* \ t' \land \ t' \ [\wro] \ r \land \trans{h}{r} \ [\so \cup \wro]^+ \trans{h}{e} \} \cup t$ 
		\EndProcedure	
	\end{algorithmic}
	\label{algorithm:canonical-order}
\end{algorithm}

The function $\textsc{canonicalOrder}$ produces a relation between transactions in a history, denoted $\leq^h$. In algorithm \ref{algorithm:canonical-order}'s description, we denote $\bot$ to represent the end of the program, which always exists, and that is $\so$-related with every single transaction.

Firstly, we prove our canonical order is well defined for every pair of transactions.

\begin{lemma}
	\label{lemma:dep_shrinks}
	For every history $h$, event $e$ and transaction $t$, $\dep(h, t, \min_{<_{\ora}} \dep(h, t, e)) \subseteq \dep(h, t, e)$. Moreover, if $\dep(h, t, e) \neq t$, the inclusion is strict.
	
\end{lemma}
\begin{proof}
	Let $r' = \min_{<_{\ora}} \dep(h, t, e)$. If $\dep(h, t, r') = t$ the lemma is trivially proved, so let's suppose there exists $r \in \dep(h, t,r') \setminus t$. Then, $\exists t' \text{ s.t. } t \ [\so \cup \wro]^* \ t' \land \ t' \ [\wro] \ r \land \trans{h}{r} \ [\so \cup \wro]^+ \trans{h}{r'}$ and $\exists t'' \text{ s.t. } t \ [\so \cup \wro]^* \ t'' \land \ t'' \ [\wro] \ r' \land \trans{h}{r'} \ [\so \cup \wro]^+ \trans{h}{e}$; so $\trans{h}{r} \ [\so \cup \wro]^+ \trans{h}{r'} \ [\so \cup \wro]^+ \trans{h}{e}$. In other words, $r \in \dep(h, t, e)$. The moreover comes trivially as $r' \not\in \dep(h, t, r')$.
\end{proof}

\begin{lemma}
	\label{lemma:minimalDependency-halts}
	For every pair of distinct transactions $t, t'$, $\minimalDependency(h,t,t',\bot)$ always halts.
\end{lemma}
\begin{proof}
	Let's suppose by contrapositive that $\minimalDependency(h,t, t',\bot)$ does not halt. Therefore, there would exist an infinite chain of events $e_n, n \in \mathbb{N}$ such that $e_0 = \bot, e_{n+1} = \min_{\ora}\dep(h, t, e_{n}) = \min_{\ora}\dep(h, t', e_{n})$. Firstly, as $h$ is finite, so are both $\dep(h, t, e_{n})$ and $\dep(h, t', e_{n})$. Moreover, if $e_n \not \in t$, $\dep(h, t, e_{n+1}) \subsetneq \dep(h, t, e_{n})$ (and analogously for $t'$). Therefore, there exist some indexes $n_0, m_0$ such that $e_{n_0} \in t$ and $e_{m_0} \in t'$. Let $k = \max\{n_0, m_0\}$. Because ; but if $e_n \in t$, $t = \dep(h, t, e_n)$ and $e_{n+1} = e_n$, so $e_k = e_{n_0}$ and $e_k = e_{m_0}$. Therefore $e_k \in t \cap t'$; so $t = t'$ as transaction logs do not share events; which contradict the assumptions.
\end{proof}

\begin{corollary}
	The relation $\leq^h$ is well defined for every pair of transactions.
\end{corollary}
\begin{proof}
	As by lemma \ref{lemma:minimalDependency-halts}, we know that $\minimalDependency(h, t, t', \bot)$ always halts if $t \neq t'$; it is clear that $\textsc{canonicalOrder}(h, t, t')$ also does it. Therefore, the relation is well defined.
\end{proof}

Now that $\leq^h$ has been proved a well defined relation between each pair of transactions, let us prove that it is indeed a total order.

\begin{lemma}
	\label{lemma:canonincal-total-order}
	The relation $\leq^h$ is a total order.
\end{lemma}
\begin{proof}
	$ $
	\begin{itemize}
		\item \underline{Strongly connection} Let $t_1, t_2$ s.t. $t_1 \not \leq^h t_2$. If $t_2 \ [\so \cup \wro]^* t_1$, then $t_2 \leq^h t_1$. Otherwise, as $\lnot(t_1 \ [\so \cup \wro]^* \ t_2)$ and $\minimalDependency$ halts (lemma \ref{lemma:minimalDependency-halts}) either \linebreak $\minimalDependency(h, t_1, t_2, \bot)$ or $\minimalDependency(h,t_2, t_1, \bot)$ holds. But as $t_1 \not\leq^h t_2$, $t_2 \leq^h t_1$.
		\item \underline{Reflexivity:} By definition, for every $t$, $t \leq^h t$.
		\item \underline{Transitivity:} Let $t_1, t_2, t_3$ three distinct transactions such that $t_1 \leq^h t_2 $ and $t_2 \leq^h t_3$. Clearly, if $t_1 \ [\so \cup \wro]^* \ t_3$, $t_1 \leq^h t_3$. However, if $t_3 \ [\so \cup \wro]^* \ t_1$, we would find one of the following three scenarios:
		\begin{itemize}
			\item $t_1 \ [\so \cup \wro]^* \ t_2$, which is impossible by strong connectivity as that would mean $t_3 \leq^h t_2$.
			\item $t_2 \ [\so \cup \wro]^* \ t_3$, which is also impossible by strong connectivity, as $t_2 \leq^h t_1$.
			\item $\lnot(t_1 \ [\so \cup \wro]^* \ t_2)$ and $\lnot(t_2 \ [\so \cup \wro]^* \ t_3)$. Then, let us call $e^i_0 = \bot$ and $e^i_{n+1} = \min_{<_{\ora}}\dep(h, t_i, e^i_{n})$ for $i \in \{1,2, 3\}$. Let's prove by induction that if for every $k < n$ $e^1_n \not\in t^1$, then $e^1_n = e^{2}_n = e^3_n$. Clearly this hold for $n = 0$ and, assuming it holds for every $k \leq n-1$, as $t_1 \leq^h t_2$, $t_2 \leq^h t_3$, we know $e^1_n \leq_{\ora} e^2_n \leq_{\ora} e^3_n$ and as $t^3 \ [\so \cup \wro]^* \ t^1$, if $e^1_n \not\in t^1$, $ e^3_n \leq_{\ora} e^1_n$. In other words, they coincide. However, by lemma \ref{lemma:minimalDependency-halts}, we know $\minimalDependency(h, t^1, t^3, \bot)$ halts, so there exists some minimal $n_0$ such that $e^1_{n_0} \in t^1$; so $e^2_{n_0} \in t_1$. That implies $t^2 \ [\so \cup \wro]^* \ t_1$; which is impossible as $t_1 \leq^h t_2$.
		\end{itemize}
		
		We deduce then that either $t_1 \ [\so \cup \wro]^* \ t_3$ or $\lnot(t_3 \ [\so \cup \wro]^* \ t_1)$. In the latter case, let's take the sequence $e^i_n$, $i \in \{1,2,3\}$ defined in the last paragraph. Then, as by lemma \ref{lemma:minimalDependency-halts} $\minimalDependency(h, t_1, t_3, \bot)$ halts, there exists a maximum index $n_0$ such that $e^1_{n_0} = e^2_{n_0} = e^3_{n_0}$. Then $e^1_{n_0 + 1} <_{\ora} e^2_{n_0+1}$ or $e^2_{n_0+1} <_{\ora} e^3_{n_0}$; so $t_1 \leq^h t_3$.
		
		\item \underline{Antisymmetric} Let $t_1, t_2$ s.t. $t_1 \leq^h t_2$ and $t_2 \leq^h t_1$. If $t_1 \ [\so \cup \wro]^* t_2$, then $t_1 = t_2$. If not, by the symmetric argument, $\lnot(t_2 \ [\so \cup \wro]^* t_1)$. In that situation, by lemma \ref{lemma:minimalDependency-halts} we know both $\minimalDependency(h,t_1,t_2,\bot)$ and $\minimalDependency(h,t_1,t_2,\bot)$ halt and cannot be satisfied at the same time. This contradicts that both $t_1 \leq^h t_2$ and $t_2 \leq^h t_1$ hold; so $t_1 = t_2$.		
	\end{itemize}
\end{proof}

\subsubsection{Oracle-respectful histories} $ $\\

The second step in this proof is characterizing all reachable histories with some general invariant that can be generalized to every total history. For doing so, we will show that for reachable histories any history order coincide with its canonical order; so any property based on a history order can be generalized to be based on its canonical order.

\begin{definition}
\label{def:oracle-respectful}
An ordered history $(h, \leq)$ is \callout{$\ora$-respectful} with respect to $\leq$ if it has at most one pending transaction log and for every pair of events $e \in \prog, e' \in h$ s.t. $e \leq_{\ora} e'$, either $e \leq e'$ or $\exists e'' \in h, \trans{h}{e''} \leq_{\ora} \trans{h}{e}$ s.t. $\trans{h}{e'} \ [\so \cup \wro]^* \ \trans{h}{e''}$, $e'' \leq e$ and $\swapped{h}{e''}$; where if $e \not\in h$ we state $e' \leq e$ always hold but $e \leq e'$ never does. We will denote it by $\oraRespectful{h}{\leq}$. 
\end{definition}

\begin{lemma}
\label{lemma:reachable-or-respectful}
Let $p$ a computable path. Every ordered history $(h, \leq_h)$ in $p$ is $\ora$-respectful with respect to $\leq_h$.
\end{lemma}
\begin{proof}
	
We will prove this property by induction on the number of histories this path has. The base case, the empty path, trivially holds; so let us prove the inductive case: for every path of at most length $n$ the property holds. Let $p$ a path of length $n+1$ and $h_<$ the last reachable history of this path. As $p \setminus \{h\}$ is a computable path of length $n$, the immediate predecessor of $h$ in $p$, $(h_p,{<_{h_p}})$ is $\ora$-respectful with respect to $<_p$. Let $a = \nextEvent(h_p)$.

Firstly, if $a$ is not a $\iread$ nor a $\ibegin$ event and $h = h_p \oplus a$, as $\leq_h$ is an extension of $\leq_{h_p}$, $a$ belongs to the only pending transaction and $\ora$ orders transactions completely, we can deduce that $h$ is $\ora$-respectful with respect to $\leq$. 

In addition, if $a$ is a $\ibegin$ event and $h = h_p \oplus a$, let $e \in \prog, e' \in h$ s.t. $e <_{\ora} e'$. If $e \in h_p$ or $e' \neq a$, as $\leq_{h}$ is an extension of $\leq_{h_p}$ and $\oraRespectful{h_p}{\leq_{h_p}}$ holds, the condition for satisfying $\oraRespectful{h}{\leq}$ holds with $e$ and $e'$. Moreover, as $a = \min_{\ora} \prog \setminus h_p$, there is no event $e \in \prog \setminus h_p$ s.t. $e \leq_{\ora} a$; so $\oraRespectful{h}{\leq}$ holds.

Moreover, if $a$ is a $\iread$ event and $h = h_p \oplus \wro(t, a)$ for some transaction log $t$, let us call $e \in \prog, e' \in h$ s.t. $e <_{\ora} e'$. Once again, if $e \in h$ or $e' \neq a$ the property holds; so let's suppose $e \in \prog \setminus h_p$ and $e' = a$. Let $b = \ibegin(\trans{h}{a})$, that also belongs to $h_p$. Firstly, as $e \leq_{\ora} \trans{h}{e'} = \trans{h}{b}$ we know that $e \leq_{\ora} b$. Secondly, as $\oraRespectful{h_p}{\leq_{h_p}}$, $e \not\in h_p$ and $e \leq_{\ora} b$; there exists $c \in h_p$, $\trans{h_p}{c} \leq_{\ora} \trans{h_p}{a}$ s.t. $(\trans{h_p}{b}, \trans{h_p}{c}) \in (\so \cup \wro)^*$, $c \leq b$ and $\swapped{(h_p,{<_{h_p}})}{c}$. As $\trans{h}{a} = \trans{h}{b}$ and $\swapped{(h_p,{<_{h_p}})}{c}$ implies $\swapped{h_<}{c}$, we conclude $\oraRespectful{h}{\leq}$.

But if no previous case is satisfied, it is because $h = \swap((h_p,{<_{h_p}}), r, t)$ for some $r, t \in h_p$ s.t. $\genericProtocol((h_p,{<_{h_p}}), r, t)$ holds. Let $e, e'$ two events s.t. $e \leq_{\ora} e'$. On one hand, if $e \leq e'$, $\oraRespectful{h}{e}$ holds. On the other hand, if $e' < e$ and $e' \leq_{h_p} e$, as $\oraRespectful{h_p}{\leq_{h_p}}$ holds and no swapped event is deleted by $\genericProtocol((h_p,{<_{h_p}}), r, t)$'s definition, the property is also satisfied. Finally, if \oldver{$e < e'$} \nver{$e' < e$} and $e \leq_{h_p} e'$, $e$ has to be a deleted event so $e \in \prog \setminus h$. As $r \leq_{h_p} e$, if $e \leq_{\ora} a$, \nver{as $e \not\leq a$,} there would exist a $c \in h_p$, $\trans{h_p}{c} \leq_{\ora} \trans{h_p}{e} \leq_{\ora} \trans{h_p}{r}$ s.t. $(\trans{h_p}{r}, \trans{h_p}{c}) \in (\so \cup \wro)^*$ and $\swapped{h_<}{c}$. However, \nver{this is impossible as $\trans{h_<}{r}$ has as maximal event $r$ and the algorithm preserves at most one pending transaction}\oldver{ contradicts $\genericProtocol((h_p,{<_{h_p}}), r, t)$ as $a$ must be in the last transaction of $h$, that is $(\so \cup \wro)^+$-maximal}; so $e \leq_{\ora} a$. Taking $e'' = r$ the property is witnessed. 
\end{proof}

\begin{proposition}
	\label{proposition:orders-coincide}
	For any reachable history $h$, $\leq^h \equiv \leq_h$.
\end{proposition}
\begin{proof}
For proving this equivalence, we will show that in any computable path and for any ordered history $(h, \leq_h)$, if $t \leq_h t'$, then $t \leq^h t'$, as by lemma \ref{lemma:canonincal-total-order} $\leq^h$ is a total order and therefore they have to coincide. We will prove this by induction on the number of histories a path has. The base case, the empty path, trivially holds; so let us prove the inductive case: for every path of at most length $n$ the property holds. Let $p$ a path of length $n+1$ and $h_{<_h}$ the last reachable ordered history of this path. As $p \setminus \{h\}$ is a computable path of length $n$, the immediate predecessor of $h$ in $p$, $\leq^{h_p} \equiv \leq_{h_p}$. Let $e = \nextEvent(h_p)$. Firstly, let's note that if $h$ is an extension of $h_p$, as $\oraRespectful{h_p}{<_{h_p}}$, the property can only fail while comparing a transaction $t$ with $\trans{h}{e}$.
	\begin{itemize}
		\item \underline{$h$ extends $h_p$ and $e$ is a $\ibegin$:} As $\dep(h_p, t, \bot) = \dep(h, t, \bot)$ for every transaction in $h_p$, if $t \leq^{h_p} t'$, then $t \leq^h t'$. Moreover, $\dep(h, \trans{h}{e}, \bot) = \{e\}= \min_{\ora} \prog \setminus h_p$. By lemma \ref{lemma:reachable-or-respectful} $h$ is $\ora$-respectful, so for every $t$, $\min_{\ora} \dep(h, t, \bot) <_{\ora} e$; which implies $t <^h \trans{h}{e}$. By lemma \ref{lemma:canonincal-total-order}, $\leq^h$ is a total order, so it coincides with $\leq_h$.
		
		\item \underline{$h$ extends $h_p$ and $e$ is not a $\ibegin$:} As no transaction depends on $\trans{h}{e}$ and $\trans{h}{e} = $ \linebreak $\last{h_p}$, if we prove that for every pair of transactions $\minimalDependency(h_p, t', t'', \bot) $ \linebreak $= \minimalDependency(h, t', t'', \bot)$, the lemma would hold. On one hand, \linebreak $\dep(h, \trans{h}{e}, \bot) = \dep(h_p, \trans{h}{e}, \bot) = \trans{h}{e}$ and in the other hand, by lemma \ref{lemma:reachable-or-respectful}, $\min_{\ora} \dep(h_p, t, \bot) <_{\ora} \trans{h}{e}$. Finally, as $e \not\in \dep(h, \hat{t}, e')$, for every $\hat{t} \neq \trans{h}{e}, e' \neq \bot$, for every pair of transactions $t', t''$, $\minimalDependency(h_p, t', t'' \bot) =$ \linebreak $ \minimalDependency(h, t', t'', \bot)$. 
		
		\item \underline{$h = \swap(h_p, r, t)$, where $t = \trans{h}{e}$:} As $\genericProtocol(h_p, r, t)$ is satisfied and $h$ is \linebreak $\ora$-respectful, for every event $e'$ and transaction $t'$ in $h$, $\min_{\ora} \dep(h_p, t', e') =$ \linebreak$ \min_{\ora} \dep(h, t', e')$, so for every pair of transactions $\minimalDependency(h_p, t', t'', \bot) =$ \linebreak $ \minimalDependency(h, t', t'', \bot)$. In particular, this implies $t' \leq^{h_p} t''$ if and only if $t' \leq^h t''$ for every  pair $t', t'' \in h$. Finally, as for every $t' \in h$, $t' \leq^h \trans{h}{r}$ (because $\trans{h}{r}$ is $(\so \cup \wro)^+$-maximal); we conclude that $\leq^h \equiv \leq_h$. 
	\end{itemize}
\end{proof}

Proposition \ref{proposition:orders-coincide} is a very interesting result as it express the following fact: regardless of the computable path that leads to a history, the final order between events will be the same. Therefore, all possible history orders collapse to one, the canonical one. This result will have a key role during both completeness and optimality, as it restricts the possible histories that precede another while describing the computable path leading to it. In addition, proposition \ref{proposition:orders-coincide} together with lemma \ref{lemma:reachable-or-respectful} justify enlarging definition \ref{def:oracle-respectful} with a general order as for reachable histories, $\oraRespectful{h}{\leq_h}$ is equivalent to $\oraRespectful{h}{\leq^h}$. From what follows, we will simply state $h$ is $\ora$-respectful and we will denote it by $\oraRespectfulCanon{h}$. Moreover, we will assume every history is ordered with the canonical order.

\begin{corollary}
	\label{corollary:soundness-swapped}
	Let $h_p$ a reachable history and let $h$ a immediate successor of $h_p$ whose last event $r$ is a $\iread$. Then $h_< = \swap((h_p, <_{h_p}), r, t)$ if and only if $\swapped{h}{r}$ does.
\end{corollary}

\newpage
\begin{proof}
	\underline{$\implies$}
	
	Let's suppose that $h_< = \swap((h_p, <_{h_p}), r, t)$ for some $t$ transaction. As the last event in $h$ is $r$ and by definition of $\swap$ function no event reads from \nver{$\wro^{-1}(r)$} in $h$ besides $r$, to prove $\swapped{h}{r}$ holds we just need to show that $r <_{\ora} t$. By lemma \ref{lemma:reachable-or-respectful}, $\oraRespectfulCanon{h_p}$ holds. As $r <_{h_p} t$, $\genericProtocol((h_p, <_{h_p}), r, t)$ holds and $t$ is $(\so \cup \wro)^+$-maximal, we conclude that $r <_{\ora} t$.
	
	\underline{$\impliedby$}
	Let's suppose that $h = h_p \oplus r \oplus \wro(r, t)$ for some transaction $t$. Let's suppose that $r <_{\ora} t$. As $\oraRespectfulCanon{h_p}$, there exists some event $e''$ s.t. $\trans{h_p}{e''} \leq \trans{h}{r}$, $t \ [\so \cup \wro]^* \trans{h}{e''}$ and $e'' \leq r$ so $\lnot(\swapped{h}{r})$.
	
\end{proof}

\begin{lemma}
	\label{lemma:total-respectful}
	Any total history is $\ora$-respectful.
	\begin{proof}
		Let $h$ be a total history and $t, t'$ a pair of transactions s.t. $t \leq_{\ora} t'$. If $t \leq^h t'$, then the statement is satisfied; so let's assume the contrary: $t' \leq^h t$. If $(t', t) \in (\so \cup \wro)^*$, then for every $e \in t, e' \in t'$ $\exists c \in h$ s.t. $\trans{h}{c} \leq_{\ora} \trans{h}{e}$, ($\trans{h}{e'}, \trans{h}{c}) \in  (\so \cup \wro)^* $, $\swapped{h}{c}$ and $c \leq^h e$; so the property is satisfied. Otherwise, by definition of $\minimalDependency$, there exists $r' \in h$ s.t. $(t', \trans{h}{r'}) \in (\so \cup \wro)^* $ and $\trans{h}{r'} \leq_{\ora} t$. Moreover, by \textsc{canonicalOrder}'s definition, $\trans{h}{r} \leq^h t$. Finally $\swapped{h}{r'}$ holds as it is the minimum element according $\ora$. To sum up, $\oraRespectfulCanon{h}$ holds.
	\end{proof}
\end{lemma}

\subsubsection{Previous of a history} 
$ $\\

As a third and final step in our proof, we define the function \textit{previous} that, for a every history $h$, if $\prev(h)$ is reachable, then $h$ is also reachable. Moreover, $\prev(h)$ will belong to the same computable path.

\begin{algorithm}[h]
		\label{algorithm:prev}
	\caption{\textsc{prev}}
	\begin{algorithmic}[1]
		
		\Statex
		\Procedure{\textsc{prev}}{$h$}
		\If{$h = \emptyset$}
		\State \Return $\emptyset$
		\EndIf
		\State $a \gets \last{h}$
		\If{$\lnot \swapped{h}{a}$}
		\State \Return $h \setminus a$
		\Else
		\Let $t$ s.t. $(t,r) \in \wro$.
		\State \Return $\maxCompletion(h\setminus a, \{e \ | \ e \not\in (h \setminus a) \land e <_{\ora} t \})$
		\EndIf
		\EndProcedure		
		\Procedure{\maxCompletion}{$h, D$}
		\If{$D \neq \emptyset$}
		\State $e \gets \min_{<_{\ora}} D$
		\If{$\mathit{type}(e) \neq \iread$}
		\State \Return $\maxCompletion(h \oplus e, D \setminus \{e\})$
		\Else
		\Let $t$ s.t. $\isMaximallyAdded{h \oplus e \oplus \wro(t, e)}{e}$ holds
		\State \Return $\maxCompletion(h \oplus e \oplus \wro(t, e), D \setminus \{e\})$
		\EndIf
		
		\Else
		\State \Return $h$
		\EndIf
		\EndProcedure

	\end{algorithmic}

\end{algorithm}




First, we show that the invariant of our algorithm is preserved via $\prev$.

\begin{lemma}
	\label{lemma:prev-respectful}
	For every $\ora$-respectful history $h$, $\prev(h)$ is also $\ora$-respectful.
\end{lemma}
\begin{proof}
	Let suppose $h \neq \emptyset$, $h_p = \prev(h)$, $a = \last{h}$, $e \in \prog$ and $ e' \in h_p$ s.t. $e \leq_{\ora} e'$. We explore different cases depending if $e, e'$ belong to $h$ or not. If $e' \in h_p \setminus h$, $\lnot(\swapped{h_p}{e})$ and $ \lnot(\swapped{h_p}{e'})$ holds. As $\min_{<_{\ora}} \dep(h, \trans{h}{e'}, \bot) = \ibegin(\trans{h}{e'})$, we obtain that $\min_{<_{\ora}}\dep(h,\trans{h}{e'}) \leq_{\ora} e' \leq_{\ora} \ibegin(\trans{h}{e'})$. Therefore, as $e' \in h_p \in h$, $\lnot(\trans{h}{e'} \ [\so \cup \wro]^+ \ \trans{h}{e})$, so $e \leq^h e'$. And if $e' \in h$, either $e \leq^h e'$ or $e' \leq^h e$. In the former case, both are in $h$ and therefore, in $h_p$. As it cannot happen that $e' \in \trans{h}{a}$ and $e \leq^{h_p} a$ because $\swapped{h}{a}$ and $e \leq_{\ora} e'$, we conclude that $e \leq^h e'$ ($\leq_{h_p}$ keeps the relative orders between transactions different from $\trans{h}{a}$ and by lemma \ref{lemma:reachable-or-respectful} they coincide). In the latter case,  by $\oraRespectfulCanon{h}$, there exists $e''$ that witness it. In particular, $\swapped{h}{e''}$ holds, so $e'' \in h_p$. $e''$ witness $\oraRespectfulCanon{h_p}$ holds. In the three cases we deduce that $\oraRespectfulCanon{h_p}$.
	

	
\end{proof}

Next, we have to prove that previous is a sound function, i.e. the composition between $\textsc{explore-ce}$ and $\prev$ give us the identity. For doing so, in the case a history is a swap, we deduce that both histories should contain the same elements and they read the same; so they have to coincide.

\begin{lemma}
	\label{lemma:soundness-prev}
	For every consistent history $\ora$-respectful $h$, if $\prev(h)$ is reachable, then $h$ is also reachable.
\end{lemma}
\begin{proof}
	Let suppose $h \neq \emptyset$, $h_p = \prev(h)$ and $a = \last{h}$. If $\lnot \swapped{h}{a}$, let $h_n = h_p \oplus a$ if $a$ is not a read and $h_n = h_p \oplus a \oplus \wro(t, a)$, where $t$ is the transaction s.t. $(t, r) \in \wro$, otherwise. Either way, $h_n$ is always reachable and it coincides with $h$. On the contrary, if $\swapped{h}{a}$, $a$ is a $\iread$ event and it swapped; so let us call $t$ to the transaction s.t. $(t, a) \in \wro$. Firstly, as $\swapped{h}{a}$, $a <_{\ora} t$, and by lemma \ref{lemma:reachable-or-respectful}, $\oraRespectfulCanon{h_p}$ holds, so $a <_{h_p} t$ does; which let us conclude $\genericCompute(h_p)$ will always return $(a, t)$ as a possible swap pair. In addition, all transactions in $h_p$ are non-pending and $(t,a) \in \wro$, so in particular $\last{h_p}$ is an $\icommit$ event. If we call $h_s = \swap(h_p, a, t)$, and we prove that $h_p \setminus h = h_p \setminus h_s$ holds, then we would deduce $h = h_s$ as $\wro(t, a)$ in both $h_p, h_s$ and $h \subseteq h_p, h_s \subseteq h_p$; which would allow us to conclude $h$ is reachable from $h_p$.
	
	On one hand, if $e \in h_p \setminus h$, we deduce that $e \not\in h$ and $e <_{\ora} t$. In particular, $\lnot (\trans{h}{e} \ [\so \cup \wro]^* \ t)$. Moreover, if $ e \leq_{\ora} a$, by $\oraRespectfulCanon{h}$, either $e \leq^h a$ or $\exists e''\in h, e'' \leq_{\ora} e$ s.t. $\tr(a) \ [\so \cup \wro]^* \trans{h}{e''}$, $e'' \leq^h e$ and $\swapped{h}{e''}$; both impossible situations as $e \not\in h$ and $a = \last{h}$; so $a \leq_{\ora} e$. In other words, $e \in h_p \setminus h_s$.
	
	On the other hand, $e \in h_p \setminus h_s$ if and only if $\lnot (\trans{h}{e} \ [\so \cup \wro]^* \ \tr(w))$ and $a <_{\ora} e <_{\ora} w$. If $e$ would belong to $h$ then $e \leq^{h} a$. As $h$ is $\ora$-respectful and $a \leq_{\ora} e$, we deduce there exists a $e'' \in h$ s.t. $\trans{h}{e''} \leq_{\ora} \tr(a)$, $\trans{h}{e} \ [\so \cup \wro]^* \trans{h}{e''}$ and $\swapped{h}{e''}$. Moreover, as $e'' \in h$, $e'' \in h_p$. By corollary \ref{corollary:soundness-swapped} $\swapped{h_p}{e''}$ and $\genericProtocol(h_p, a, t)$ hold,  $e'' \in h_s$ and so $e$ does. This result leads to a contradiction, so $e \not\in h$; i.e. $e \in h_p \setminus h$.
\end{proof}

\begin{corollary}
	\label{corollary:prev-swap-identity}
	In a consistent $\ora$-respectful history $h$ whose previous history is reachable, if $a = \last{h}$, $\swapped{h}{a}$ and $t$ is a transaction such that $(t, a) \in \wro$, $h$ coincides with $\swap(\prev(h), a, t)$.
\end{corollary}
\begin{proof}
	It comes straight away from the proof of lemma \ref{lemma:soundness-prev}.
\end{proof}

Once proven that $\prev$ is sound, let us prove that for every history we can compose $\prev$ a finite number of times obtaining the empty history. We are going to prove it by induction on the number of swapped events, so we prove first the recursive composition finishes in finite time and then we conclude our claim.

\begin{lemma}
	\label{lemma:prev-reduces-one}
	For every non-empty consistent $\ora$-respectful history $h$, $h_p = \prev(h)$ and $a = \last{h}$, if $\swapped{h}{a}$ then $\{e \in h_p \ | \ \swapped{h_p}{e}\} = \{e \in h \ | \ \swapped{h}{e}\} \setminus \{a\}$, otherwise $h_p = h \setminus a$.
	\begin{proof}
		Let $a = \last{h}$ and $h' = h \setminus a$. If $\lnot(\swapped{h}{a})$, then $h_p = h'$ and the lemma holds trivially. Otherwise, as $h_p =  \maxCompletion(h')$, we will show that every event not belonging to $h_p \setminus h'$ is not swapped by induction on every recursive call to $\maxCompletion$. Let us call $D = \{e \ | \ e \not\in h' \land e <_{\ora} \}$. This set, intuitively, contain all the events that would have been deleted from a reachable history $h$ to produce $h_p$. In this setting, let us call $h_{|D|} = h'$, $D_{|D|} = D$ and $D_k = D_{k+1} \setminus \{\min_{<_{\ora}} D_{k+1}\}, \; e_k = \min_{<_{\ora}}D_k$ for every $k, 0 \leq k < |D|$ (i.e. $D_k = D_{k+1} \setminus \{ e_{k+1}\}$). We will prove the lemma by induction on $n = |D| - k$, constructing a collection of \nver{$\ora$-respectful} histories $h_k$, $0 \leq k < |D|$, such that each one is an extension of its predecessor with a non-swapped event.
		
		The base case, $h_{|D|}$ is trivial as by its definition it corresponds with $h'$. Let's prove the inductive case: $\{e \ | \ \swapped{h_{k+1}}{e}\} = \{e \ | \ \swapped{h'}{e}\}$. If $e_{k+1}$ is not a $\iread$ event, $h_k = h_{k+1} \oplus e_{k+1}$, \nver{$\oraRespectfulCanon{h_k}$} and $\{e \ | \ \swapped{h_{k}}{e}\} = \{e \ | \ \swapped{h'}{e}\}$; as only $\iread$ events can be swapped. Otherwise, $e_{k+1}$ is a read event. By the isolation level's causal-extensibility there exists a transaction $f_{k+1}$ that writes the same variable as $e_{k+1}$, $(f_{k+1}, \trans{h}{e_{k+1}}) \in (\so \cup \wro)^* $ and $h_{k+1} \oplus e_{k+1} \oplus \wro( f_{k+1}, e_{k+1})$ is consistent. 
		\oldver{Moreover, by corollary \ref{corollary:soundness-swapped} $ \{e \ | \ \swapped{h_{k+1}}{e}\} $ $= \{e \ | \ \swapped{h_{k+1} \oplus e_{k+1} \oplus \wro(f_{k+1}, e_{k+1})}{e}\}$.}
		\nver{Moreover, if $e_{k+1}$ reads from any causal dependent element $f'$, $f'$ in $h_{k+1}$, it cannot be swapped: as $\oraRespectfulCanon{h_{k+1}}$ holds, if $e_{k+1} <_{\ora} f'$ there must be an event $c_{k+1}$ s.t. $ \trans{h}{c_{k+1}} \leq_{\ora} \trans{h}{e_{k+1}}$ and $(f', \trans{h}{c_{k+1}}) \in (\so \cup \wro)^*$. Hence, $ \{e \ | \ \swapped{h_{k+1}}{e}\} $ $= \{e \ | \ \swapped{h_{k+1} \oplus e_{k+1} \oplus \wro(f', e_{k+1})}{e}\}$.
		} 
		
		Let $E_{k+1} = \{t \ | \ h_{k+1} \oplus e_{k+1} \oplus \wro(t, e_{k+1}) \models I$ $\land\  \{e \ | \ \swapped{h_{k+1}}{e}\} = s \{e \ | \ \swapped{h_{k+1} \oplus e_{k+1} \oplus \wro(t, e_{k+1}) }{e}\}\}$ and let \nver{$t_{k+1} = \max_{\leq^{h_{k+1}}} \{ t \in E_{k+1} \ | \ (t, \trans{h_{k+1}}{e_{k+1}}) \in (\so \cup \wro)^* \}$.} This element is well defined as $f_{k+1}$ belongs to $E_{k+1}$. Therefore, $h_k = h_{k+1} \oplus e_{k+1} \oplus \wro(t_{k+1}, e_{k+1})$ is consistent and $\{e \ | \ \swapped{h_{k}}{e}\} =$ $ \{e \ | \ \swapped{h'}{e}\}$. Moreover, let's remark that as $t_{k+1}$ is the maximum transaction according to $\leq_{h_{k+1}}$ s.t. is consistent and $\{e \ | \ \swapped{h_{k}}{e}\} = \{e \ | \ \swapped{h'}{e}\}$. In addition, by construction, it also satisfies $\isMaximallyAdded{h_{k}}{ e_{k+1}, w_{k+1}}$. \nver{Finally, $h_{k}$ is also $\ora$-respectful as $e_{k+1}$ is not swapped and $\oraRespectfulCanon{h_{k+1}}$ holds.}
		
		Thus, after applying induction, we obtain $h_p = h_0$; which let us conclude $\{e \in h_p \ | \ \swapped{h_p}{e}\} = \{e \in h' \ | \ \swapped{h'}{e}\} = \{e \in h \ | \ \swapped{h}{e}\}\setminus \{a\}$.
	\end{proof}
\end{lemma}

\begin{lemma}
	\label{lemma:prev-leads-empty}
	For every consistent $\ora$-respectful history $h$ there exists some $k_h \in \mathbb{N}$ such that $\prev^{k_h}(h) = \emptyset$.
	\begin{proof}
		This lemma is immediate consequence of lemma \ref{lemma:prev-reduces-one}. Let us call $\xi(h) = $ \linebreak $|\{e \in h \ | \ \swapped{h}{e}\}|$, the number of swapped events in $h$, and let us prove the lemma by induction on $(\xi(h), |h|)$. The base case, $\xi(h) = |h| = 0$ is trivial as $h$ would be $\emptyset$; so let's assume that for every history $h$ such that $\xi(h) < n$ or $\xi(h) =h \land |h| < m$ there exists such $k_h$. Let $h$ then a history s.t. $\xi(h) = n$ and $|h| = m$. $h_p = \prev(h)$. On one hand, if $h_p = h \setminus a$ then $\xi(x_p) = \xi(h)$ and $|h_p| = |h|-1$. On the other hand, if $h_p \neq h \setminus a$, $\xi(h_p) = \xi(h) - 1$. In any case, by induction hypothesis on $h_p$, there exists an integer $k_{h_p}$ such that $\prev^{k_{h_p}}(h_p) = \emptyset$. Therefore, $k_h = k_{h_p}+ 1$ satisfies $\prev^{k_h}(h) = \emptyset$.
	\end{proof}
\end{lemma}

\begin{proposition}
	\label{proposition:chain-histories-proof}
	For every consistent $\ora$-respectful history $h$ exists $k \in \mathbb{N}$ and some sequence of $\ora$-respectful histories $\{h_n\}_{n = 0}^k$, $h_0 = \emptyset$ and $h_k = h$ such that the algorithm will compute.
	\begin{proof}
		Let $h$ a history, $k$ the minimum integer such that $\prev^k(h) = \emptyset$, which exists thanks to lemma \ref{lemma:prev-leads-empty} and $C = \{\prev^{k-n}(h)\}_{n = 0}^k$ a set of indexed histories. By the collection's definition and lemma \ref{lemma:prev-respectful}, $h_0 = \prev^k(h) = \emptyset$, $h_k = \prev^0(h) = h$ and $\oraRespectfulCanon{h_n}$ for every $n \in \mathbb{N}$; so let us prove by induction on $n$ that every history in $C$ is reachable. The base case, $h_0$, is trivially achieved; as it is always reachable. In addition, by lemma \ref{lemma:soundness-prev}, we know that if $h_n$ is reachable, $h_{n+1}$ is it too; which proves the inductive step. 
	\end{proof}
\end{proposition}

\begin{theorem}
	\label{theorem:completeness}
	The algorithm \textup{\textsc{explore-ce}} is complete.
	\begin{proof}
		By lemma \ref{lemma:total-respectful}, any consistent total history is $\ora$-respectful. As a consequence of proposition \ref{proposition:chain-histories-proof}, there exist a sequence of reachable histories which $h$ belongs to; so in particular, $h$ is reachable.
	\end{proof}
\end{theorem}

\subsection{Optimality}

For proving optimality we are going to exploit two properties already studied for completeness: $\ora$-respectfulness and the canonical order. Then, as algorithm \textsc{explore-ce} is sound and complete, we will prove that any computable path leading to a consistent history is the one computed in the completeness' proof.

\begin{theorem}
	Algorithm \textup{\textsc{explore-ce}} is strongly optimal.
	\begin{proof}
		As the model is causal-extensible, any algorithm optimal is also strongly optimal. Let us prove that for every reachable history there is only a computable path that leads to it from $\emptyset$.
		Let's suppose there exists a history $h$ that is reached $p_1$, $p_2$ by two computable paths. By lemma \ref{proposition:orders-coincide}, we know that $\leq_h \equiv \leq^h$. However, $\leq^h$ is an order that does not depend on the computable path that leads to $h$; so neither does $\leq_h$.  Therefore, we can assume without loss of generality that $h$ is a history with minimal value of $\xi(h) =|\{e \in h \ | \swapped{h}{e}\}|$ and in case of tie, that is minimal with respect $|h|$; values independent of the computable path that leads to $h$. 
		
		We can also assume without loss of generality that the predecessor of $h$ in $p_1$ is $h_1 = \prev{h}$, and $h_2$ is the predecessor of $h$ in $p_2$. If we prove $h_1$ and $h_2$ are identical, $p_1$ and $p_2$ have to also be identical and therefore, the algorithm would be optimal. Firstly, if $\last{h}$ is not a swapped $\iread$ event, by the definition of $\nextEvent$ function $h_2 = h \setminus \last{h} = h_1$. On the contrary, let's suppose $r = \last{h}$ is a swapped event that reads from a transaction $t$. Because $\swapped{h}{r}$ holds, from $h_2$ to $h$ it has to have happened a swap between $r$ and $w$. But by corollary \ref{corollary:prev-swap-identity}, $h = \swap(h_1, r, w)$, so $h_1 \restriction_{h \setminus r} = h_2 \restriction_{h \setminus r}$. As $h_1, h_2$ are both $\ora$-respectful, $e \in h_1 \setminus h \iff e \in h_2 \setminus h$. Finally, as $\genericProtocol(h_i, r, w)$ holds for $i \in \{1,2\}$, for every $\iread$ event $e$ in $h_1 \cap h_2$ there exists a transaction $t_e$ s.t. $\wro(e, t_e)$ for both histories.
	\end{proof}
\end{theorem}

	\newpage
\section{Experimental data}
\label{sec:appendix:experiments}

\subsection{Application Scalability}
\label{ssec:appendix:application-scalability}


\begin{table}[H]
    \resizebox{\textwidth}{!}{

    \begin{tabular}{ccccccccccccccc}
        \toprule
        & \multicolumn{4}{c}{CC}                     &  & \multicolumn{4}{c}{CC + SI}                     &  & \multicolumn{4}{c}{CC + SER}                    \\ \cline{2-5} \cline{7-10} \cline{12-15} 
        & Histories & End states & Time   & Mem.     &  & Histories & End states & Time   & Mem.     &  & Histories & End states & Time   & Mem.     \\ \hline
        courseware-1 & 216 & 216 & 00:00:22 & 370 & ~ & 81 & 216 & 00:00:25 & 370 & ~ & 72 & 216 & 00:00:23 & 370 \\ 
        courseware-2 & 46 & 46 & 00:00:06 & 316 & ~ & 34 & 46 & 00:00:06 & 308 & ~ & 34 & 46 & 00:00:06 & 314 \\ 
        courseware-3 & 12790 & 12790 & 00:12:45 & 533 & ~ & 6197 & 12790 & 00:13:59 & 533 & ~ & 960 & 12790 & 00:12:37 & 557 \\ 
        courseware-4 & 69 & 69 & 00:00:07 & 314 & ~ & 39 & 69 & 00:00:08 & 324 & ~ & 17 & 69 & 00:00:07 & 370 \\ 
        courseware-5 & 388 & 388 & 00:00:25 & 308 & ~ & 136 & 388 & 00:00:27 & 370 & ~ & 71 & 388 & 00:00:24 & 370 \\ 
        shoppingCart-1 & 444 & 444 & 00:00:19 & 308 & ~ & 108 & 444 & 00:00:22 & 308 & ~ & 81 & 444 & 00:00:19 & 308 \\ 
        shoppingCart-2 & 2934 & 2934 & 00:00:55 & 308 & ~ & 811 & 2934 & 00:01:13 & 444 & ~ & 480 & 2934 & 00:00:58 & 308 \\ 
        shoppingCart-3 & 1594 & 1594 & 00:00:55 & 308 & ~ & 1077 & 1594 & 00:01:05 & 308 & ~ & 338 & 1594 & 00:01:00 & 308 \\ 
        shoppingCart-4 & 58677 & 58677 & \red{TL}  & 444 & ~ & 12440 & 49589 & \red{TL}  & 444 & ~ & 779 & 60194 & \red{TL}  & 383 \\ 
        shoppingCart-5 & 4686 & 4686 & 00:02:56 & 444 & ~ & 1986 & 4686 & 00:03:07 & 444 & ~ & 780 & 4686 & 00:02:41 & 308 \\ 
        tpcc-1 & 165 & 165 & 00:00:43 & 794 & ~ & 47 & 165 & 00:00:47 & 808 & ~ & 47 & 165 & 00:00:45 & 796 \\ 
        tpcc-2 & 353 & 353 & 00:01:25 & 699 & ~ & 35 & 353 & 00:01:29 & 879 & ~ & 31 & 353 & 00:01:25 & 704 \\ 
        tpcc-3 & 1593 & 1593 & 00:10:12 & 966 & ~ & 232 & 1593 & 00:10:29 & 1054 & ~ & 116 & 1593 & 00:10:05 & 803 \\ 
        tpcc-4 & 105 & 105 & 00:00:15 & 450 & ~ & 22 & 105 & 00:00:16 & 485 & ~ & 1 & 105 & 00:00:15 & 396 \\ 
        tpcc-5 & 7836 & 7836 & \red{TL}  & 1732 & ~ & 695 & 6973 & \red{TL}  & 1647 & ~ & 271 & 7617 & \red{TL}  & 1640 \\ 
        twitter-1 & 36 & 36 & 00:00:05 & 256 & ~ & 29 & 36 & 00:00:06 & 308 & ~ & 18 & 36 & 00:00:05 & 256 \\ 
        twitter-2 & 876 & 876 & 00:00:48 & 459 & ~ & 263 & 876 & 00:01:03 & 1066 & ~ & 122 & 876 & 00:00:56 & 513 \\ 
        twitter-3 & 1072 & 1072 & 00:01:24 & 444 & ~ & 576 & 1072 & 00:01:33 & 569 & ~ & 216 & 1072 & 00:01:21 & 444 \\ 
        twitter-4 & 12915 & 12915 & 00:08:36 & 444 & ~ & 1680 & 12915 & 00:10:59 & 640 & ~ & 1680 & 12915 & 00:09:12 & 533 \\ 
        twitter-5 & 12915 & 12915 & 00:07:53 & 444 & ~ & 1680 & 12915 & 00:11:18 & 533 & ~ & 1680 & 12915 & 00:07:43 & 444 \\ 
        wikipedia-1 & 649 & 649 & 00:03:59 & 820 & ~ & 95 & 649 & 00:02:34 & 699 & ~ & 95 & 649 & 00:02:32 & 695 \\ 
        wikipedia-2 & 3610 & 3610 & 00:13:51 & 792 & ~ & 328 & 3610 & 00:14:13 & 696 & ~ & 292 & 3610 & 00:13:50 & 798 \\ 
        wikipedia-3 & 2339 & 2339 & 00:05:44 & 640 & ~ & 175 & 2339 & 00:06:04 & 640 & ~ & 175 & 2339 & 00:05:41 & 640 \\ 
        wikipedia-4 & 691 & 691 & 00:01:54 & 774 & ~ & 246 & 691 & 00:02:01 & 768 & ~ & 108 & 691 & 00:01:55 & 774 \\ 
        wikipedia-5 & 21317 & 21317 & \red{TL}  & 620 & ~ & 292 & 19840 & \red{TL}  & 533 & ~ & 220 & 22307 & \red{TL}  & 444 \\
        \bottomrule
    \end{tabular}
    }
\end{table}

\begin{table}[H]
    \resizebox{\textwidth}{!}{
    \begin{tabular}{ccccccccccccccccccc}
        \toprule
        & \multicolumn{4}{c}{RA + CC}                     &  & \multicolumn{4}{c}{RC + CC}                     &  & \multicolumn{4}{c}{\texttt{true} + CC}           &  & \multicolumn{3}{c}{DFS(CC)}      \\ \cline{2-5} \cline{7-10} \cline{12-15} \cline{17-19} 
        & Histories & End states & Mem.   & Time     &  & Histories & End states & Mem.   & Time     &  & Histories & End states & Mem.  & Time &  & End states & Mem.   & Time     \\ \hline
        courseware-1 & 216 & 893 & 00:00:52 & 370 & ~ & 216 & 11751 & 00:07:02 & 370 & ~ & 216 & 124399 & \red{TL} & 444 & ~ & 58072 & 00:13:52 & 308 \\ 
        courseware-2 & 46 & 106 & 00:00:06 & 308 & ~ & 46 & 588 & 00:00:11 & 308 & ~ & 46 & 1074 & 00:00:15 & 308 & ~ & 18010 & 00:02:12 & 308 \\ 
        courseware-3 & 10585 & 47570 & \red{TL} & 444 & ~ & 40 & 65075 & \red{TL} & 308 & ~ & 40 & 119388 & \red{TL} & 370 & ~ & 186758 & \red{TL} & 308 \\ 
        courseware-4 & 69 & 88 & 00:00:07 & 308 & ~ & 69 & 2392 & 00:00:29 & 308 & ~ & 69 & 3779 & 00:00:39 & 315 & ~ & 37956 & 00:07:19 & 308 \\ 
        courseware-5 & 388 & 765 & 00:00:39 & 308 & ~ & 388 & 66557 & 00:27:12 & 370 & ~ & 320 & 96681 & \red{TL} & 660 & ~ & 68074 & \red{TL} & 308 \\ 
        shoppingCart-1 & 444 & 1620 & 00:00:34 & 370 & ~ & 444 & 202066 & \red{TL} & 370 & ~ & 370 & 173904 & \red{TL} & 308 & ~ & 69396 & \red{TL} & 308 \\ 
        shoppingCart-2 & 2934 & 32976 & 00:19:58 & 450 & ~ & 366 & 224700 & \red{TL} & 370 & ~ & 366 & 508967 & \red{TL} & 370 & ~ & 93549 & \red{TL} & 308 \\ 
        shoppingCart-3 & 1594 & 6291 & 00:01:41 & 308 & ~ & 1594 & 131226 & 00:18:23 & 308 & ~ & 1594 & 223740 & \red{TL} & 308 & ~ & 99522 & \red{TL} & 533 \\ 
        shoppingCart-4 & 19945 & 53687 & \red{TL} & 533 & ~ & 151 & 256686 & \red{TL} & 308 & ~ & 11 & 267433 & \red{TL} & 308 & ~ & 270996 & \red{TL} & 370 \\ 
        shoppingCart-5 & 4686 & 16323 & 00:06:43 & 370 & ~ & 2469 & 265924 & \red{TL} & 469 & ~ & 371 & 420084 & \red{TL} & 444 & ~ & 191813 & \red{TL} & 404 \\ 
        tpcc-1 & 165 & 958 & 00:02:13 & 839 & ~ & 7 & 50588 & \red{TL} & 1045 & ~ & 3 & 84272 & \red{TL} & 545 & ~ & 18489 & \red{TL} & 1383 \\ 
        tpcc-2 & 353 & 3958 & 00:11:24 & 809 & ~ & 3 & 7029 & \red{TL} & 1153 & ~ & 1 & 23097 & \red{TL} & 670 & ~ & 25253 & \red{TL} & 1029 \\ 
        tpcc-3 & 1475 & 10969 & \red{TL} & 1029 & ~ & 20 & 22934 & \red{TL} & 688 & ~ & 2 & 119267 & \red{TL} & 459 & ~ & 18124 & \red{TL} & 1251 \\ 
        tpcc-4 & 105 & 114 & 00:00:15 & 474 & ~ & 17 & 50203 & \red{TL} & 640 & ~ & 3 & 112330 & \red{TL} & 670 & ~ & 22645 & \red{TL} & 948 \\ 
        tpcc-5 & 271 & 9178 & \red{TL} & 1629 & ~ & 3 & 4059 & \red{TL} & 662 & ~ & 3 & 48306 & \red{TL} & 768 & ~ & 36060 & \red{TL} & 1284 \\ 
        twitter-1 & 36 & 44 & 00:00:05 & 256 & ~ & 36 & 4104 & 00:01:07 & 370 & ~ & 36 & 12384 & 00:02:56 & 473 & ~ & 35056 & 00:20:25 & 533 \\ 
        twitter-2 & 876 & 2917 & 00:02:01 & 592 & ~ & 876 & 18219 & 00:09:30 & 548 & ~ & 876 & 37943 & 00:17:35 & 544 & ~ & 145070 & \red{TL} & 533 \\ 
        twitter-3 & 1072 & 2272 & 00:02:05 & 576 & ~ & 1072 & 9514 & 00:08:10 & 533 & ~ & 1072 & 20164 & 00:16:14 & 588 & ~ & 108792 & 00:24:34 & 452 \\ 
        twitter-4 & 12915 & 48363 & 00:29:09 & 476 & ~ & 10 & 114588 & \red{TL} & 370 & ~ & 1 & 147462 & \red{TL} & 533 & ~ & 50404 & \red{TL} & 533 \\ 
        twitter-5 & 12915 & 48363 & 00:27:07 & 444 & ~ & 84 & 70376 & \red{TL} & 444 & ~ & 84 & 136241 & \red{TL} & 370 & ~ & 57654 & \red{TL} & 444 \\ 
        wikipedia-1 & 649 & 2296 & 00:04:16 & 672 & ~ & 64 & 37382 & \red{TL} & 832 & ~ & 4 & 66814 & \red{TL} & 699 & ~ & 54510 & \red{TL} & 660 \\ 
        wikipedia-2 & 2049 & 8451 & \red{TL} & 979 & ~ & 17 & 19697 & \red{TL} & 795 & ~ & 2 & 85523 & \red{TL} & 930 & ~ & 43629 & \red{TL} & 650 \\ 
        wikipedia-3 & 2339 & 6170 & 00:10:06 & 640 & ~ & 100 & 28952 & \red{TL} & 522 & ~ & 50 & 23974 & \red{TL} & 581 & ~ & 43962 & \red{TL} & 682 \\ 
        wikipedia-4 & 691 & 1781 & 00:02:57 & 925 & ~ & 5 & 44937 & \red{TL} & 567 & ~ & 3 & 61334 & \red{TL} & 543 & ~ & 24873 & \red{TL} & 1188 \\ 
        wikipedia-5 & 13159 & 26384 & \red{TL} & 533 & ~ & 29 & 72930 & \red{TL} & 444 & ~ & 23 & 78413 & \red{TL} & 695 & ~ & 97881 & \red{TL} & 444 \\ 
        \bottomrule
    \end{tabular}
    }
\end{table}

\subsection{Session Scalability}
\label{ssec:appendix:session-scalability}

\begin{table}[H]
    \resizebox{\textwidth}{!}{
\begin{tabular}{cccccccccccccccccccc}
    \toprule
    &\multicolumn{3}{c}{One session}&  & \multicolumn{3}{c}{Two sessions} &  & \multicolumn{3}{c}{Three sessions} &  & \multicolumn{3}{c}{Four sessions}  &  & \multicolumn{3}{c}{Five sessions}  \\ \cline{2-4} \cline{6-8} \cline{10-12} \cline{14-16} \cline{18-20} 
    & Histories & Time   & Mem.     &  & Histories & Time   & Mem.     &  & Histories & Time   & Mem.     &  & Histories & Time   & Mem.     &  & Histories & Time   & Mem.     \\ \hline
    tpcc-1 & 1 & 00:00:02 & 256 &  & 6 & 00:00:03 & 256 &  & 1540 & 00:05:42 & 804 &  & 3081 & \red{TL} & 4096 &  & 14525 & \red{TL} & 4096 \\ 
    tpcc-2 & 1 & 00:00:03 & 256 &  & 66 & 00:00:17 & 587 &  & 9630 & \red{TL} & 2900 &  & 17637 & \red{TL} & 4076 &  & 2442 & \red{TL} & 4096 \\ 
    tpcc-3 & 1 & 00:00:03 & 256 &  & 12 & 00:00:09 & 384 &  & 4824 & 00:25:43 & 1503 &  & 3463 & \red{TL} & 4096 &  & 2940 & \red{TL} & 4096 \\ 
    tpcc-4 & 1 & 00:00:03 & 256 &  & 90 & 00:00:41 & 674 &  & 6355 & \red{TL} & 1728 &  & 1722 & \red{TL} & 4096 &  & 2634 & \red{TL} & 4096 \\ 
    tpcc-5 & 1 & 00:00:03 & 256 &  & 96 & 00:00:41 & 692 &  & 3659 & \red{TL} & 1765 &  & 1343 & \red{TL} & 4092 &  & 1481 & \red{TL} & 4096 \\ 
    wikipedia-1 & 1 & 00:00:02 & 256 &  & 199 & 00:00:19 & 370 &  & 19654 & \red{TL} & 640 &  & 16377 & \red{TL} & 533 &  & 12419 & \red{TL} & 4096 \\ 
    wikipedia-2 & 1 & 00:00:02 & 256 &  & 38 & 00:00:14 & 536 &  & 7055 & 00:22:07 & 768 &  & 21000 & \red{TL} & 3520 &  & 16985 & \red{TL} & 4096 \\ 
    wikipedia-3 & 1 & 00:00:02 & 256 &  & 67 & 00:00:14 & 444 &  & 9346 & \red{TL} & 768 &  & 9451 & \red{TL} & 4096 &  & 3264 & \red{TL} & 4096 \\ 
    wikipedia-4 & 1 & 00:00:02 & 256 &  & 7 & 00:00:07 & 374 &  & 73 & 00:00:14 & 602 &  & 3940 & 00:20:40 & 4096 &  & 1325 & \red{TL} & 4096 \\ 
    wikipedia-5 & 1 & 00:00:02 & 256 &  & 28 & 00:00:08 & 308 &  & 336 & 00:05:41 & 662 &  & 10914 & \red{TL} & 4096 &  & 563 & \red{TL} & 3936 \\ 
    \bottomrule
\end{tabular}
}
\end{table}

\subsection{Transaction Scalability}
\label{ssec:appendix:transaction-scalability}


\begin{table}[H]
    \resizebox{\textwidth}{!}{
        \begin{tabular}{cccccccccccccccccccc}
        \toprule
        &\multicolumn{3}{c}{One transaction}&  & \multicolumn{3}{c}{Two transactions}   &  & \multicolumn{3}{c}{Three transactions} &  & \multicolumn{3}{c}{Four transactions}  &  & \multicolumn{3}{c}{Five transactions}  \\ \cline{2-4} \cline{6-8} \cline{10-12} \cline{14-16} \cline{18-20} 
        & Histories & Mem.   & Time     &  & Histories & Mem.   & Time     &  & Histories & Mem.   & Time     &  & Histories & Mem.   & Time     &  & Histories & Mem.   & Time     \\ \hline
        tpcc-1 & 4 & 00:00:03 & 256 & ~ & 107 & 00:00:46 & 674 & ~ & 303 & 00:01:32 & 812 & ~ & 13780 & 00:27:59 & 3904 & ~ & 13431 & \red{TL} & 4096 \\ 
        tpcc-2 & 18 & 00:00:11 & 444 & ~ & 4030 & 00:16:10 & 1063 & ~ & 5162 & \red{TL} & 2012 & ~ & 5351 & \red{TL} & 4096 & ~ & 3243 & \red{TL} & 4096 \\ 
        tpcc-3 & 3 & 00:00:04 & 256 & ~ & 219 & 00:01:38 & 881 & ~ & 6679 & \red{TL} & 1327 & ~ & 6533 & \red{TL} & 4093 & ~ & 2036 & \red{TL} & 4096 \\ 
        tpcc-4 & 20 & 00:00:13 & 444 & ~ & 5187 & 00:20:04 & 1046 & ~ & 3262 & \red{TL} & 2066 & ~ & 1548 & \red{TL} & 4096 & ~ & 2045 & \red{TL} & 4096 \\ 
        tpcc-5 & 1 & 00:00:03 & 256 & ~ & 23 & 00:00:15 & 596 & ~ & 171 & 00:01:31 & 901 & ~ & 1812 & \red{TL} & 3933 & ~ & 4091 & \red{TL} & 4096 \\ 
        wikipedia-1 & 16 & 00:00:04 & 308 & ~ & 2428 & 00:02:30 & 444 & ~ & 22289 & \red{TL} & 533 & ~ & 17113 & \red{TL} & 640 & ~ & 14648 & \red{TL} & 4096 \\ 
        wikipedia-2 & 9 & 00:00:06 & 256 & ~ & 56 & 00:00:20 & 533 & ~ & 739 & 00:02:02 & 690 & ~ & 5364 & \red{TL} & 4068 & ~ & 3568 & \red{TL} & 3977 \\ 
        wikipedia-3 & 18 & 00:00:07 & 308 & ~ & 1109 & 00:01:24 & 640 & ~ & 26110 & \red{TL} & 768 & ~ & 15339 & \red{TL} & 3621 & ~ & 15138 & \red{TL} & 3822 \\ 
        wikipedia-4 & 4 & 00:00:05 & 256 & ~ & 43 & 00:00:12 & 444 & ~ & 3919 & 00:29:57 & 1649 & ~ & 3501 & \red{TL} & 4096 & ~ & 2506 & \red{TL} & 4096 \\ 
        wikipedia-5 & 2 & 00:00:03 & 256 & ~ & 20 & 00:00:10 & 444 & ~ & 46 & 00:00:20 & 370 & ~ & 754 & 00:05:42 & 2521 & ~ & 2573 & \red{TL} & 4096 \\ 
        \bottomrule
    \end{tabular}
    }
\end{table}

\end{document}